\documentclass[a4paper,UKenglish,cleveref, autoref, thm-restate]{lipics-v2021}

\title{The Complexity of\\Homomorphism Reconstructibility} 

\titlerunning{The Complexity of Homomorphism Reconstructibility} 

\author{Jan Böker}{RWTH Aachen University, Germany}{boeker@informatik.rwth-aachen.de}{https://orcid.org/0000-0003-4584-121X}{European Union (ERC, SymSim, 101054974).}
\author{Louis Härtel}{RWTH Aachen University, 
Germany}{haertel@informatik.rwth-aachen.de}{https://orcid.org/0009-0004-3446-5874}{}
\author{Nina Runde}{RWTH Aachen University, Germany}{runde@lics.rwth-aachen.de}{https://orcid.org/0009-0000-4547-1023}{German Research Foundation (DFG) ---  453349072.}
\author{Tim Seppelt}{RWTH Aachen University, Germany}{seppelt@cs.rwth-aachen.de}{https://orcid.org/0000-0002-6447-0568}{German Research Foundation (DFG) --- GRK 2236/2 (UnRAVeL).}
\author{Christoph Standke}{RWTH Aachen University, Germany}{standke@informatik.rwth-aachen.de}{https://orcid.org/0000-0002-3034-730X}{German Research
	Foundation (DFG) --- GR 1492/16-1 and GRK 2236/2 (UnRAVeL).}

\authorrunning{J. Böker, L. Härtel, N. Runde, T. Seppelt, C. Standke} 

\Copyright{Jan Böker, Louis Härtel, Nina Runde, Tim Seppelt, and Christoph Standke} 



\category{} 

\relatedversion{}

\acknowledgements{We would like to thank Martin Grohe, Athena Riazsadri, and Jan Tönshoff for fruitful discussions.}

\nolinenumbers

\hideLIPIcs

\usepackage{todonotes}
\usepackage{multicol}
\usepackage{csquotes}

\usepackage{amsmath, amsfonts, amsthm, amssymb, mathtools, stmaryrd}
\usepackage{tikz}
\usetikzlibrary{calc,
	cd,
	decorations.pathreplacing,
	calligraphy,
	positioning}
\tikzstyle{vertex}=[circle, draw=lipicsGray, fill=lipicsGray,inner sep=0pt, minimum size=2mm]
\tikzstyle{smallvertex}=[circle, draw=lipicsGray, fill=lipicsGray,inner sep=0pt, minimum size=1.2mm]
\tikzstyle{kneser}=[circle, draw=lipicsGray, minimum size=0.8cm]
\tikzstyle{cycle}=[circle, draw=lipicsGray, minimum size=1.5cm]

\newsavebox{\fminibox}
\newlength{\fminilength}
\newenvironment{fminipage}[1][\linewidth]{
	\setlength{\fminilength}{#1-2\fboxsep-2\fboxrule}\begin{lrbox}{\fminibox}\begin{minipage}{\fminilength}}{
	\end{minipage}\end{lrbox}\noindent\fbox{\usebox{\fminibox}}
}
\usepackage{xspace}

\newcommand{\abs}[1]{\left\lvert #1 \right\rvert}

\renewcommand\phi\varphi
\renewcommand\epsilon\varepsilon

\DeclareMathOperator{\surj}{surj}

\DeclareMathOperator{\sub}{sub}
\DeclareMathOperator{\indsub}{indsub}
\newcommand\multiset[1]{\left\{ \!\! \left\{ #1 \right\} \!\! \right\}}

\newcommand{\CC}{\mathcal{C}}

\newcommand{\CF}{\mathcal{F}}
\newcommand{\CG}{\mathcal{G}}
\newcommand{\CI}{\mathcal{I}}
\newcommand{\CL}{\mathcal{L}}
\newcommand{\CS}{\mathcal{S}}

\newcommand{\NN}{\mathbb{N}}

\newcommand{\Aut}{\mathsf{Aut}}
\newcommand{\CEP}{\mathsf{C_= P}}
\newcommand{\NEXP}{\mathsf{NEXP}}
\newcommand{\NP}{\mathsf{NP}}
\newcommand{\PH}{\mathsf{PH}}
\newcommand{\SP}{\#\mathsf{P}}
\newcommand{\Poly}{\mathsf{P}}
\newcommand{\FPT}{\mathsf{FPT}}
\newcommand{\GapP}{\mathsf{GapP}}
\newcommand{\FP}{\mathsf{FP}}

\newtheorem{proviso}[definition]{Proviso}
\Crefname{proviso}{Proviso}{Provisos}

\newcommand{\dproblem}[3]{
	\begin{center}
		\begin{fminipage}[.95\linewidth]
			\textup{\textsc{#1}\begin{description}
				\item[Input] #2
				\item[Question] #3
			\end{description}}
\end{fminipage}
\end{center}}
\newcommand{\pproblem}[4]{
	\begin{center}
		\begin{fminipage}[.95\linewidth]
			\textup{\textsc{#1}\begin{description}
				\item[Input] #2
				\item[Parameter] #3
				\item[Question] #4
			\end{description}}
\end{fminipage}
\end{center}}

\newcommand{\FHomROf}[1]{\textup{\textsc{HomRec($#1$)}}\xspace}
\newcommand{\FHomR}{\FHomROf{\CF}}
\newcommand{\FGHomROf}[2]{\textup{\textsc{HomRec($#1, #2$)}}\xspace}
\newcommand{\FGHomR}{\FGHomROf{\CF}{\CG}}
\newcommand{\pFGHomR}{$p$-\FGHomR}

\newcommand{\BoundedFHomROf}[1]{\textup{\textsc{BHomRec($#1$)}}\xspace}
\newcommand{\BoundedFHomR}{\BoundedFHomROf{\CF}}
\newcommand{\BoundedFGHomROf}[2]{\textup{\textsc{BHomRec($#1, #2$)}}\xspace}
\newcommand{\BoundedFGHomR}{\BoundedFGHomROf{\CF}{\CG}}

\newcommand{\FSubROf}[1]{\textup{\textsc{SubRec($#1$)}}\xspace}

\newcommand{\FGSubROf}[2]{\textup{\textsc{SubRec($#1, #2$)}}\xspace}
\newcommand{\FGSubR}{\FGSubROf{\CF}{\CG}}
\newcommand{\pFGSubR}{$p$-\FGSubR}
\newcommand{\BoundedFSubROf}[1]{\textup{\textsc{BSubRec($#1$)}}\xspace}
\newcommand{\BoundedFSubR}{\BoundedFSubROf{\CF}}
\newcommand{\BoundedFGSubROf}[2]{\textup{\textsc{BSubRec($#1, #2$)}}\xspace}
\newcommand{\BoundedFGSubR}{\BoundedFGSubROf{\CF}{\CG}}

\newcommand{\SCR}{\textup{\textsc{$p$-SingleHomRec($\mathcal{F}, \mathcal{G}$)}}\xspace}
\newcommand{\subsize}{\textup{\textsc{$p$-EquiSizeSubRec($\mathcal{F}, \mathcal{G}$)}}\xspace}

\newcommand{\SHom}{\#\textsc{Hom}\xspace}

\newcommand{\SSat}{\#\textup{\textsc{Sat}}\xspace}
\newcommand{\SetSplitting}{\textup{\textsc{SetSplitting}}\xspace}
\newcommand{\QuadraticPolynomial}{\textup{\textsc{QPoly}}\xspace}
\newcommand{\BPoly}{\textup{\textsc{BPoly}}\xspace}
\newcommand{\DecProblemA}{\textup{\textsc{A}}\xspace}
\newcommand{\DecProblemB}{\textup{\textsc{B}}\xspace}

\newcommand{\threecolouring}{\textup{\textsc{$3$-Colouring}}\xspace}
\newcommand{\existsEqualsThreeColouring}{\textup{\textsc{EC-}}\threecolouring}
\newcommand{\threeSAT}{\textup{$3$-\textsc{Sat}}\xspace}
\newcommand{\existsEqualsThreeSAT}{\textup{\textsc{EC-}}\threeSAT}

\renewcommand\theta\vartheta
\renewcommand\phi\varphi
\renewcommand\epsilon\varepsilon

\begin{document}

\maketitle

\begin{abstract}
    Representing graphs by their homomorphism counts
    has led to the beautiful theory of homomorphism indistinguishability
    in recent years.
    Moreover, homomorphism counts have promising applications
    in database theory and machine learning,
    where one would like
    to answer queries or classify graphs
    solely based on the representation of a graph $G$ as
    a finite vector of homomorphism counts
    from some fixed finite set of graphs to $G$.
    We study the computational complexity of the arguably
    most fundamental computational problem
    associated to these representations,
    the \emph{homomorphism reconstructability problem}:
    given a finite sequence of graphs and a corresponding vector of natural numbers,
    decide whether
    there exists a graph $G$ that realises the given vector
    as the homomorphism counts from the given graphs.

    We show that this problem yields a natural example of an $\NP^{\SP}$-hard problem,
    which still can be $\NP$-hard when restricted
    to a fixed number of input graphs of bounded treewidth
    and a fixed input vector of natural numbers, or alternatively,
    when restricted to a finite input set of graphs.
    We further show that,
    when restricted to a finite input set of graphs
    and given an upper bound on the order of the graph $G$ as additional input,
    the problem cannot be $\NP$-hard unless $\Poly = \NP$.
    For this regime, we obtain partial positive results.
    We also investigate the problem's parameterised complexity
    and provide fpt-algorithms for the case
    that a single graph is given
    and that multiple graphs of the same order with subgraph
    instead of homomorphism counts are given.
\end{abstract}

\newpage

\section{Introduction}

Representing a graph in terms of homomorphism counts has proven to be fruitful in theory and applications.
Many graph properties studied in logic~\cite{dvorak_recognizing_2010,grohe_counting_2020}, algebraic graph theory~\cite{dell_lovasz_2018}, quantum information theory~\cite{mancinska_quantum_2019}, and convex optimisation~\cite{roberson_lasserre_2023,grohe_homomorphism_2022} can be expressed as homomorphism counts from some family of graphs.
Homomorphism counts provide a basis for other counting problems~\cite{curticapean_homomorphisms_2017} and have been studied extensively using diverse tools ranging from algorithmics \cite{dell_lovasz_2018} to algebra \cite{grohe_homomorphism_2022,mancinska_quantum_2019}, from combinatorics \cite{roberson_oddomorphisms_2022,seppelt_logical_2023} to category theory \cite{dawar_lovasz-type_2021,abramsky_discrete_2022}.
In database theory,
they correspond to evaluations of Boolean conjunctive queries under bag semantics \cite{chaudhuri_optimization_1993,kwiecien_determinacy_2021}.
In graph learning, representations of graphs as vectors of homomorphism counts yield embeddings into a continuous latent space and underpin theoretically meaningful and successfully field-tested machine learning architectures~\cite{NguyenM20,bouritsas_improving_2023,wolf_structural_2023}.

In this work, we consider representations of a graph $G$ as a finite vector of homomorphisms counts $\hom(\mathcal{I}, G) \in \mathbb{N}^{\mathcal{I}}$ for some finite set of graphs $\mathcal{I}$,
which we call a \emph{homomorphism embedding}.
The rich theory of homomorphism counts calls for algorithmic applications of homomorphism embeddings:
in database theory,
one would ideally like to decide properties of the graph~$G$
having access to the vector $\hom(\mathcal{I}, G)$ only~\cite{grohe_word2vec_2020,chen_queries_2021}.
In graph learning, certain entries of the homomorphism embedding might be
associated with desirable properties of the graph being embedded,
and one would like to be able to synthesise a graph having these desirable properties
from a vector in the latent space \cite{bonifati_graph_2021,guo_systematic_2022}.
Despite its ubiquity in the contexts described above, homomorphism embeddings have not undergone a systematic complexity-theoretic analysis yet.
The arguably most fundamental computational problem associated to them is to decide whether a vector $h \in \mathbb{N}^{\mathcal{I}}$ actually represents a graph, i.e.\@ it is of the form $h = \hom(\mathcal{I}, G)$ for some graph~$G$.
Therefore, we consider the following \emph{homomorphism reconstructability problem} for graph classes $\mathcal{F}$ and $\mathcal{G}$:
\dproblem{\FGHomROf{\CF}{\CG}}
{Pairs $(F_1, h_1), \dots, (F_m,h_m) \in \CF \times \NN$ where $h_1, \dots, h_m$ are given in binary.}
{Is there a graph $G \in \CG$ such that $\hom(F_i, G) = h_i$ for every $i \in [m]$?}

We simply write $\FHomR$ if $\CG$ is the class of all graphs.
While the problem has a clean-cut motivation in practical applications,
one quickly encounters surprising connections to deep theoretical
results and long-standing open questions.
One does not only have to carefully keep distance to the notorious graph reconstruction conjecture of Ulam~\cite{oneil_ulam_1970},
but also be aware that various decision problems involving the set $R(\mathcal{I})$ of all vectors $\hom(\mathcal{I}, G) \in \mathbb{N}^{\mathcal{I}}$ where $G$ ranges over all graphs have received much attention recently, e.g.\ 
the homomorphism domination problem~\cite{kopparty_homomorphism_2010}, whose decidability is open,
the homomorphism determinacy problem~\cite{kwiecien_determinacy_2021}, or the undecidable problem of determining whether inequalities of homomorphism counts hold~\cite{ioannidis_containment_1995,hatami_undecidability_2011}.
Beyond computational concerns, an abstract characterisation of the set $R(\mathcal{I})$ akin to \cite{freedman_reflection_2007,lovasz_semidefinite_2009} is desirable, yet elusive~\cite{atserias_expressive_2021}.

In this paper, we establish a firm grasp on the computational complexity
of the homomorphism reconstructability problem $\FGHomR$ by
exploring from which of its aspects computational hardness
arises and then finding restrictions for which efficient algorithms can be found.
Despite the interest in the problem, surprisingly little progress on this
question has been made.
The only result related to our work is a theorem from~\cite{ko_three_1991},
which asserts that a variant of \FGHomR with Boolean subgraph constraints
instead of homomorphism counts is $\NP^\NP$-complete.
In particular, a formal definition of $\FGHomR$
has not been made before, and we would like to remark on
a curious peculiarity of the definition we have chosen:
a polynomial-time algorithm for $\FGHomR$ may exploit
algebraic properties of homomorphism numbers and deduce
via arithmetic operations on the given numbers
whether these can be realised by a graph $G$ or not;
let us call an algorithm of this type \emph{arithmetic}.
However, it is also conceivable that an algorithm for $\FGHomR$ would operate
by explicitly constructing the graph $G$, for example, in a dynamic-programming
fashion;
let us call such an algorithm \emph{constructive}.
A constructive algorithm seems just as reasonable as an arithmetic one but
may not be a polynomial-time algorithm for $\FGHomR$
since the order of $G$ does not have to be polynomial
in the length of the binary encoding of the given homomorphism numbers.
Hence, it also seems reasonable to define the
\emph{bounded homomorphism reconstructability problem} $\BoundedFGHomR$
where an additional input $n \in \mathbb{N}$
given in unary imposes a bound on the number of vertices in $G$.
This, however, poses an additional constraint to the graph $G$, which
may make the design of arithmetic algorithms more difficult or impossible.
Hence, both $\FGHomR$ and $\BoundedFGHomR$ are arguably reasonable
definitions of the reconstructability problem, each in their own right,
and we consider both.

\paragraph*{The Ocean of Hardness}
Let $\CEP$ denote the class of all decision problems $L$
for which there is a function $f \in \SP$
and a polynomial-time computable function $g$
such that, for every instance $x$ of $L$, we have
$x \in L \iff f(x) = g(x)$~\cite{simon_central_1975,wagner_observations_1986,hemaspaandra_satanic_1995}.
We first show that
both the unbounded and the bounded reconstructability problem
are $\NP^{\CEP}$-hard when not restricted in any way.
Note that $\NP^{\CEP} = \NP^{\SP}$ since,
whenever we issue a call to the $\SP$-oracle, we may guess the output
using nondeterminism and verify it using a $\CEP$-oracle instead~\cite{toran_complexity_1991,curticapean_parity_2016};
we refer to \Cref{sec:ct} for more details about counting classes.

\begin{theorem}
    \label{th:NPCEPHardnessIntro}
    Let $\CG$ denote the class of all graphs.
    Then, $\FHomROf{\CG}$ is in $\NEXP$ and
    $\BoundedFHomROf{\CG}$ is in $\NP^{\CEP}$.
    Moreover, both problems are $\NP^{\CEP}$-hard.
	Hence, they are not contained
    in the polynomial hierarchy unless it collapses.
\end{theorem}

This result illustrates
two intertwined sources of hardness in the reconstructability problem:
first, the \emph{reconstruction hardness} of finding a graph $G$,
which corresponds to the class $\NP$,
and secondly, the \emph{counting hardness} of verifying that $G$
actually satisfies the given constraints,
which corresponds to the $\CEP$-oracle.
The reconstruction hardness is what we are interested in
since the counting hardness simply reflects the
hardness of counting homomorphisms, which is well understood:
the problem \SHom of counting the number of homomorphisms
from a graph $F$ to a graph $G$ is $\SP$-complete
and becomes tractable if and only if one restricts
the graphs $F$ to come from a (recursively enumerable)
class of bounded treewidth \cite{dalmau_complexity_2004}
under the complexity-theoretic assumption that \textsf{FPT} $\neq$ \textsf{\#W[1]}.

To isolate the reconstruction hardness,
we restrict $\CF$ to be a class of bounded treewidth,
which makes verifying that
a graph $G$ satisfies all given constraints possible in polynomial time.
In particular, the bounded problem is then in $\NP$
since one can guess a graph $G$ of the given size and
then verify that it satisfies all constraints in polynomial time.
We show that the intuition conveyed by \Cref{th:NPCEPHardnessIntro} is correct:
$\FHomR$ is still $\NP$-hard if $\CF$ has bounded treewidth
and even remains $\NP$-hard when only a constant number of constraints
is allowed to appear in the input.

\begin{restatable}{theorem}{hardnessboundedtwConstantNumberConstraintsIntro}
    \label{thm:bounded-number-constraints}
	There is a class $\CF$ of graphs of bounded treewidth
	such that $\FHomROf{\CF}$ and $\BoundedFHomROf{\CF}$ are $\NP$-hard.
\end{restatable}

The reduction used to prove \Cref{thm:bounded-number-constraints} further demonstrates that
both problems remain $\NP$-hard
when only allowing some fixed number of constraints:
it produces a family of instances with graphs from $\CF$
where the number of constraints,
all homomorphism numbers, and all but one constraint graph  are fixed.
This raises the question what happens if
we fix all input graphs and allow the homomorphism numbers to vary instead,
i.e.\@ does $\FHomR$ become tractable
if $\CF$ is \emph{finite}?
Even though the input to $\FHomR$ for finite $\CF$
essentially consists only of natural numbers
encoded in binary, we are still able to show that $\FHomR$
is also $\NP$-hard in this case.
In contrast, however, $\BoundedFHomR$ is \emph{sparse}
for finite $\CF$, i.e.\
it only has polynomial number of yes-instances, and hence,
unlikely to be $\NP$-hard~\cite{mahaney_sparse_1982}.

\begin{theorem} \label{thm:finite-set-hardness}
	There is a finite set $\CF$ of graphs such that
	\FHomROf{\CF} is $\NP$-hard.
	If \BoundedFHomROf{\CF} is $\NP$-complete
    for a finite set $\CF$ of graphs, then $\Poly = \NP$.
\end{theorem}

Hence, the complexity of the reconstructability problem
becomes much more nuanced for finite $\CF$, and
in order to design efficient algorithms for it,
we seemingly have to focus on $\BoundedFHomR$
for finite $\CF$.
The fact that there are only polynomially many feasible combinations
of homomorphism numbers in this case
might be somehow exploitable,
e.g.\ by a dynamic-programming algorithm
operating on a table indexed by them.

\paragraph*{Islands of Tractability}

\newcommand{\smallKone}{\tikz[baseline=-3pt]{
		\node[smallvertex] (R) {};}}
\newcommand{\smallKtwo}{\tikz[baseline=-3pt]{
		\node[smallvertex] (R) {};\node[smallvertex, left = 0.2cm of R] (M) {};\draw (M) edge (R);
}}
\newcommand{\smallKthree}{\tikz[baseline=.5pt]{
		\node[smallvertex] (R) {};\node[smallvertex, left = 0.1cm of R] (M) {};\node[smallvertex, fill = white, draw = white, above = 0.1cm of R] (Tghost) {};\node[smallvertex, fill = white, draw = white, above = 0.1cm of M] (Sghost) {};\node[smallvertex] (T) at ($(Tghost)!0.5!(Sghost)$){};

		\draw (M) edge (R);
		\draw (T) edge (R);
		\draw (M) edge (T);
}}
\newcommand{\smallLKthree}{\tikz[baseline=-5.5pt]{
        \node[smallvertex, label={right:\small$\ell_\mathsf{any}$}] (R) {};\node[smallvertex, left = 0.1cm of R, label={left:\small$\ell_\mathsf{any}$}] (M) {};\node[smallvertex, fill = white, draw = white, above = 0.1cm of R] (Tghost) {};\node[smallvertex, fill = white, draw = white, above = 0.1cm of M] (Sghost) {};\node[smallvertex, label={left:\small$\ell_\mathsf{any}$}] (T) at ($(Tghost)!0.5!(Sghost)$){};

        \draw (M) edge (R);
        \draw (T) edge (R);
        \draw (M) edge (T);
}}
\newcommand{\smallPtwo}{\tikz[baseline=.5pt]{
		\node[smallvertex] (R) {};\node[smallvertex, left = 0.1cm of R] (M) {};\node[smallvertex, fill = white, draw = white, above = 0.1cm of R] (Tghost) {};\node[smallvertex, fill = white, draw = white, above = 0.1cm of M] (Sghost) {};\node[smallvertex] (T) at ($(Tghost)!0.5!(Sghost)$){};

		\draw (M) edge (R);
		\draw (M) edge (T);
}}
\newcommand{\smallKfour}{\tikz[baseline=.5pt]{
		\node[smallvertex] (R) {};\node[smallvertex, left = 0.1cm of R] (M) {};\node[smallvertex, above = 0.1cm of R] (T) {};\node[smallvertex, above = 0.1cm of M] (S) {};\draw (M) edge (R);
		\draw (M) edge (T);
		\draw (M) edge (S);
		\draw (R) edge (T);
		\draw (R) edge (S);
		\draw (T) edge (S);
}}
\newcommand{\smallDiamond}{\tikz[baseline=.5pt]{
		\node[smallvertex] (R) {};\node[smallvertex, left = 0.1cm of R] (M) {};\node[smallvertex, above = 0.1cm of R] (T) {};\node[smallvertex, above = 0.1cm of M] (S) {};\draw (M) edge (R);
		\draw (M) edge (S);
		\draw (R) edge (T);
		\draw (R) edge (S);
		\draw (T) edge (S);
}}

The first tractable instance of \FHomR that comes to mind is given by $F_1 = \smallKone$ and $F_2 = \smallKtwo$ and $h_1, h_2 \in \mathbb{N}$.
In this case, we need to decide whether $h_2 \leq h_1(h_1-1)$ and $h_2$ is even.
Although this is fairly trivial, we encounter severe combinatorial difficulties when attempting to generalise this even to $F_1 = \smallKone$ and $F_2 = \smallKthree$.
\Cref{fig:plot-triangles} shows that the set of reconstructible vectors has a non-trivial shape. 
In particular, while highly engineered results from extremal combinatorics \cite{lovasz_large_2012,razborov_minimal_2008,huang_3-local_2014,glebov_densities_2017} provide insights in the (asymptotic) behaviour of the upper and lower boundary of that set,
we are unable to characterise the seemingly erratic gaps and spikes depicted in \cref{fig:plot-triangles}.
On the positive side, we are able map out large regions of reconstructible vectors using number-theoretic insights:

\begin{figure}
	\centering
	\includegraphics[width=0.75\textwidth]{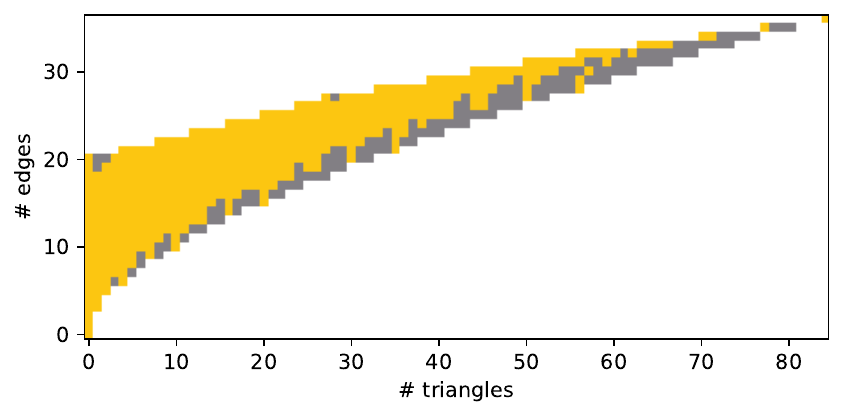}
	\caption{Reconstructable $\smallKthree$ and $\smallKtwo$ subgraph counts (yellow) for graphs on nine vertices.
		The grey area depicts the values which are not ruled out by the Kruskal--Katona bound \cite{kruskal_number_1963,katona_theorem_1968}, cf.\@ \cite[13.31b]{lovasz_combinatorial_1993}, or the Razborov bound \cite{razborov_minimal_2008}, cf.\@ \cite[Theorem~16.14]{lovasz_large_2012}.	
		The $\smallKthree$-counts realisable by graphs on nine vertices correspond to columns with at least one yellow box.
}	\label{fig:plot-triangles}
\end{figure}

\begin{restatable}{theorem}{thmcliques} \label{thm:reconstruct-cliques}
	There exists a function $\gamma \colon \mathbb{N} \to \mathbb{N}$ such that
	for every $k \geq 2$, $n \geq 1$, $h \leq \binom{n}{k}$, there exists a graph $G$ on $n+\gamma(k-1)-1$ vertices such that $\sub(K_k, G) = h$.
\end{restatable}

\Cref{thm:reconstruct-cliques} allows for the construction of graphs with almost all possible numbers of clique subgraphs. 
Indeed, the proportion of covered values is $\left.\binom{n}{k}\middle/\binom{n+\gamma(k-1)-1}{k}\right. = 1 - o(1)$ for $n \to \infty$.
In other words, all sensible values can be realised by only slightly deviating from the stipulated size constraint.

The problem arising when dealing with the remaining admissible parameters, i.e.\@ $\binom{n}{k} < h \leq \binom{n+\gamma(k-1)-1}{k}$, 
seems to be that the constraints on the number of vertices and cliques interact in an elusive fashion.
Although understanding such interactions better remains a direction for future investigations,
we are able to identify certain combinatorial conditions under which the constraints are somewhat independent.
This yields fixed-parameter algorithms for variants of \FGHomR contrasting \cref{thm:finite-set-hardness}.
Here, \cref{proviso:fpt} stipulates mild constraints on the graph classes 
$\mathcal{F}$ and $\mathcal{G}$. For example, $\mathcal{G}$ can be taken to be 
the class of all graphs and $\mathcal{F}$ to be the class of all connected 
graphs.

\begin{restatable}{theorem}{thmfpt} \label{thm:fpt}
	For graph classes $\mathcal{F}$, $\mathcal{G}$ as in \cref{proviso:fpt},
	the following problem is in~$\FPT$:
	\pproblem{\SCR}{a graph $F \in \mathcal{F}$, an integer $h \in \mathbb{N}$ given in binary}{$\lvert V(F) \rvert$}{Does there exist a graph $G \in \mathcal{G}$ such that $\hom(F, G) = h$?}

\end{restatable}

Curiously, any fpt-algorithm for \SCR has to be arithmetic: 
In $\FPT$, one can neither construct the graph $G$ nor count homomorphisms from $F$ to $G$ \cite{dalmau_complexity_2004}.
Indeed, our algorithm essentially only operates with integers and exploits number-theoretic properties of the set of reconstructible numbers.
In \cref{thm:fpt2}, we apply similar ideas to derive an fpt-algorithm for a version of \FGHomR with multiple equi-sized subgraph constraints.

\section{Preliminaries and Conventions}

Write $\mathbb{N} = \{0,1,2,\dots\}$ for the set of natural numbers.
$\DecProblemA \le_p \DecProblemB$ denotes that
a decision problem $\DecProblemA$ is polynomial-time
many-one reducible to the decision problem $\DecProblemB$.
A \emph{graph} is a pair $G = (V, E)$ of a set of \emph{vertices} $V$
and a set of \emph{edges} $E \subseteq \binom{V}{2}$.
We usually write $V(G)$ and $E(G)$ for $V$ and $E$, respectively,
and use $n$ to denote the \emph{order} $n  \coloneqq \lvert V(G) \rvert$
of $G$.
For ease of notation, we denote an edge $\{u, v \}$ by $uv$ or $vu$.
A \emph{homomorphism} from a graph $F$ to a graph $G$ is a mapping
$h \colon V(F) \to V(G)$ such that $h(uv) \in E(G)$ for every $uv \in E(F)$.
A \emph{($\CC$-vertex-)coloured graph} is a triple $G = (V, E, c)$ where $(V, E)$ is a graph,
the \emph{underlying graph},
and $c \colon V(G) \to \CC$ a function assigning a \emph{colour} from
a set $\CC$ to every vertex of $G$.
An \emph{($\CL$-)labelled} graph is defined analogously
with a function $\ell \colon \CL \to V(G)$ assigning a vertex of~$G$
to every \emph{label} from a set of labels $\CL$ instead.
Homomorphisms between coloured graphs and between labelled graphs are then
defined as homomorphisms of the underlying graphs that
respect colours and labels, respectively.

A graph $G'$ is a \emph{subgraph} of a graph $G$, written $G' \subseteq G$,
if $V(G') \subseteq V(G)$ and $E(G') \subseteq E(G)$.
The \emph{subgraph induced by a set} $U \subseteq V(G)$,
written $G[U]$, is the
subgraph of $G$ with vertices $U$ and edges $E(G) \cap \binom{U}{2}$.
We write $\hom(F, G)$ for the number of homomorphisms from $F$ to $G$,
$\sub(F, G)$ for the number of subgraphs $G' \subseteq G$ such that $G' \cong F$, and
$\indsub(F, G)$ for the number of subsets $U \subseteq V(G)$ such that $G[U] \cong F$.
This notation generalises to coloured and labelled graphs
in the straightforward way.

The definition of \FGHomR from the introduction directly generalises
to classes $\CF$ and $\CG$ of relational structures over
the same signature, and in particular, classes of labelled and coloured graphs.
We call a pair $(F, h)$ of a structure $F$ and a number $h \in \mathbb{N}$
a \textit{constraint} and also denote the pair by $\hom(F) = h$,
or for example in the context of reconstructability of subgraph counts,
by $\sub(F) = h$.
As stipulated in the introduction, if $\CF$ is a class of graphs,
we abbreviate $\FGHomROf{\CF}{\CG}$
for the class of all graphs $\CG$ to $\FHomROf{\CF}$.
We use the abbreviation $\BoundedFHomR$ in the same way, and
for a class of labelled or coloured graphs $\CF$,
we analogously abbreviate the problem name
if $\CG$ is the class of all labelled or all coloured graphs,
respectively.
We define the problems $\FGSubR$ and $\BoundedFGSubR$ analogously
to $\FGHomR$ and $\BoundedFGHomR$, respectively,
with subgraph counts instead of homomorphism counts
and also follow the conventions agreed upon above.
See \cref{app:prelim} for details.

\section{Decidability}
The problem \BoundedFGHomR is trivially decidable if membership in $\CG$ is decidable
since it is possible to perform a brute-force search for $G$ in time bounded in the size of the instance.
For \FGHomR, decidability is implied by the following lemma:

\begin{lemma}\label{lem:locality}
	Let $\mathcal{F}$ and $\mathcal{G}$
    be classes of structures over the same signature.
    Suppose that $\mathcal{G}$ is closed under taking induced substructures.
	Let $(F_1, h_1), \dots, (F_m,h_m) \in \CF \times \NN$. 
	Let $G \in \mathcal{G}$ be such that $\hom(F_i, G) = h_i$ for all $i \in [m]$.
	Then there exists $H \in \mathcal{G}$ such that $\abs{H} \leq \sum_{i=1}^m h_i \abs{F_i}$ and $\hom(F_i, H) = h_i$ for all $i \in [m]$.
\end{lemma}
\begin{proof}
	Let $U$ denote the union over all images of homomorphisms $F_i \to G$, $i \in [m]$. Clearly, $\abs{U} \leq \sum_{i=1}^m \hom(F_i, G)\abs{F_i}$. Let $H \coloneqq G[U] \in \mathcal{G}$. For all $i \in [m]$, $\hom(F_i, H) = \hom(F_i, G)$. Indeed, every homomorphism $F_i \to H$ gives rise to a homomorphism $F_i \to G$ by composition with the embedding $H \hookrightarrow G$. Conversely, observe that every homomorphism $F_i \to G$ is in fact a homomorphism $F_i \to H$ since its image is contained in $U$. It remains to observe that this correspondence establishes a bijection.
\end{proof}

\begin{theorem}
	Let $\mathcal{F}$ and $\CG$ be classes of structures over the same signature.
	Suppose that membership in $\mathcal{G}$ is decidable.
    Then \BoundedFGHomR is decidable.
    If $\CG$ is closed under taking induced substructures, then also
	\FGHomR is decidable.
\end{theorem}

\section{Hardness}
\label{sec:complexity}

In this section,
we prove the hardness results
presented in the introduction.
First, we remark that, for the class $\CG$ of all graphs,
$\FHomROf{\CG}$ is in $\NEXP$
since we can non-deterministically guess a graph~$G$
of exponential size by \Cref{lem:locality} and then count homomorphisms to~$G$
in exponential time by simply going through all mappings
from the given constraints graphs to~$G$.
For the bounded problem,
we are given a size bound on $G$ as part of the input,
which means that $\BoundedFGHomROf{\CG}{\CG}$ is in $\NP^\CEP$ since
we can non-deterministically guess a graph~$G$ of linear size
and then verify that $G$ satisfies all constraints
by using the $\CEP$-oracle.

\begin{theorem}
    \label{th:BoundedIsInNPCEP}
	Let $\CG$ denote the class of all graphs.
	Then,
    $\FHomROf{\CG}$ is in $\NEXP$
    and
    $\BoundedFHomROf{\CG}$ is in $\NP^\CEP$.
\end{theorem}

Let $\CG$ denote the class of all graphs.
The fact that $\BoundedFHomROf{\CG} \in \NP^\CEP$
reflects the intuition on the hardness
on $\BoundedFHomROf{\CG}$
that we gave in the introduction, i.e.\
that there are two intertwined sources of hardness,
the \emph{reconstruction hardness},
manifested as $\NP$,
and the \emph{counting hardness},
manifested as the $\CEP$-oracle.
In \Cref{sec:unboundedTreewidth}, we show that this intuition is in fact correct
and that $\FHomROf{\CG}$ and $\BoundedFHomROf{\CG}$ are $\NP^\CEP$-hard.
In \Cref{sec:threeConstraints}, we further reinforce this intuition
by presenting a reduction from the well-known $\NP$-complete problem $\SetSplitting$
to $\FHomR$ and $\BoundedFHomR$ for a class of bounded treewidth $\CF$
that proves that these problems are $\NP$-hard for a family of inputs
where the number of constraints,
all homomorphism numbers, and all graphs but one are fixed.
This isolates the reconstruction hardness and again underlines that
our intuition of the reconstruction hardness being $\NP$-hardness is correct.

In our reduction from $\SetSplitting$,
the given instance to $\SetSplitting$
is encoded as a constraint graph that grows with the size
of the given instance
while the number of produced constraints and the produced homomorphism numbers
remain fixed.
This raises the question if
we can achieve tractability by
restricting the order of the constraint graphs instead, i.e.\
if $\FHomR$ and $\BoundedFHomR$ become tractable
if $\CF$ is \emph{finite}.
In \Cref{sec:finiteGraphSet},
we show that there is a finite class $\CF$ for which $\FHomR$ is NP-hard
by reducing from the $\NP$-complete problem $\QuadraticPolynomial$
of solving an equation involving a quadratic polynomial.
We further show that this reduction
cannot be adapted to work for $\BoundedFGHomR$
and then prove that $\BoundedFGHomR$ is sparse, i.e.\
it only has polynomial number of yes-instances,
which means that it cannot be $\NP$-hard
under the assumption that $\Poly \neq \NP$.
In \Cref{sec:hardnessSubgraphCounts},
we briefly discuss which of our hardness results also hold for
the subgraph reconstructability problem.

For the sake of presentability, we only
provide the reductions to the reconstructability
problem for labelled
or coloured graphs in the main body of this paper.
In \Cref{sec:uncoloured}, we show that
all these reductions be adapted to
(unlabelled and uncoloured) graphs
via gadget constructions that employ \emph{Kneser graphs},
cf.\ \Cref{sec:kneserGraphs}.

\subsection{$\NP^\CEP$-Hardness}
\label{sec:unboundedTreewidth}

We first show that
both problems $\FHomR$ and $\BoundedFHomR$ are $\NP^\CEP$-hard.
We reduce from the following $\NP^\CEP$-complete variant of $\threecolouring$,
cf.\ \Cref{sec:completeness}.

\dproblem{$\existsEqualsThreeColouring$}
    {A graph $G$, a subset of vertices $S\subseteq V(G)$, and a $k \in \NN$ given in binary.}
    {Is there a homomorphism $c \colon G[S] \to \smallKthree$ such that there are exactly $k$ homomorphisms $\hat{c} \colon G \to \smallKthree$ such that $\hat{c}|_S = c$?}

The idea is that the number of homomorphisms from a graph $F$
to the complete graph on three vertices $\smallKthree$ is precisely
the number of $3$-colourings of $F$.
Then, one can formulate constraints
that can only be satisfied by $\smallKthree$,
and by adding an additional constraint $\hom(F) = h$,
one obtains a yes-instance to the reconstructability problem
if and only if the number of $3$-colourings of $F$ is exactly $h$.
By additionally employing labels on $F$ and $\smallKthree$, we obtain a reduction from 
$\existsEqualsThreeColouring$.

\begin{theorem}
    \label{th:NPCEPReduction}
	Let $\CL\CG$ denote the class of all labelled graphs.
	Then,
    $\existsEqualsThreeColouring \le_p \FHomROf{\CL\CG}$
    and
    $\existsEqualsThreeColouring \le_p \BoundedFHomROf{\CL\CG}$.
\end{theorem}
\begin{proof}
	Given an instance $(F,S,k)$ of $\existsEqualsThreeColouring$,
    let $m := |S|$. Fix an arbitrary linear order on $S$,
    i.e. $S = \{s_1,\dots,s_m\}$. We use the labels
    $\ell_1,\dots,\ell_m$ to classify vertices of $F$
    as members of $S$:
    let $F'$ be the labelled graph obtained
    from $F$ and $S$ by assigning label $\ell_i$
    to the vertex $s_i \in S$ for every $i \in [m]$.
    The reduction then produces the following constraints, where for
    $\BoundedFHomROf{\CL\CG}$, we set the additional size constraint to three:
	\begin{multicols}{2}
		\begin{alphaenumerate}
        \item\label{hcfg} $\hom(F') = k$,
        \item\label{hck1} $\hom(\smallKone) = 3$,
        \item\label{hck2} $\hom(\smallKtwo) = 6$, and
\item\label{hcell} $\hom(\tikz[baseline=-3pt]{
            \node[smallvertex, label={right:\small$\ell_i$}] (ell_i) {}
        }\hspace{-3pt}) = 1$ for every $i\in[m]$.
        \end{alphaenumerate}
    \end{multicols}
    Note that $\smallKthree$ is the unique graph that satisfies $\hom(\smallKone) = 3$ and $\hom(\smallKtwo) = 6$.
The remaining constraints enforce 
    that each label from 
    $\{\ell_1,\dots,\ell_m\}$
    appears exactly once in~$G$. For a partition $\CL = \{L_1,L_2,L_3\}$
    of these labels into at most three parts,
    let $G_\CL$ denote the $\smallKthree$ with its three vertices labelled by
    the labels in $L_1$, $L_2$, and $L_3$, respectively.
    Then, $(F, S, k)$ is in $\existsEqualsThreeColouring$
    if and only if
there is a partition $\CL$ as above such that $\hom(F', G_\CL) = k$,
    which again is the case if and only if
    there is a labelled graph $G$
    that satisfies the constraints produced by the reduction.
\end{proof}

By encoding vertex labels by gadgets consisting of Kneser graphs,
we can also obtain a reduction for uncoloured graphs.
Since the number of labels used in the reduction above
depends on the input instance, we view every label
as a binary number and construct a gadget that encodes this binary number
via Kneser graphs. This allows us to only use a constant and finite set of Kneser graphs,
which means that we do not have to worry about the size of the Kneser graphs,
which guarantees that the resulting reduction still runs in polynomial time.
The proof can be found in \Cref{sec:uncolouredGeneralHardness}.

\begin{theorem}
    \label{th:uncolouredGeneralHardness}
	Let $\CG$ denote the class of all graphs.
	Then,
    $\existsEqualsThreeColouring \le_p \FHomROf{\CG}$
    and
    $\existsEqualsThreeColouring \le_p \BoundedFHomROf{\CG}$.
\end{theorem}

Together with the following observation,
this then finishes the proof of \Cref{th:NPCEPHardnessIntro}.

\begin{corollary}\label{todascor}
	Let $\CG$ denote the class of all graphs.
	Then,
    $\FHomROf{\CG}$ and
    $\BoundedFHomROf{\CG}$ are not in $\PH$
	unless $\PH$ collapses.
\end{corollary}
\begin{proof}
	This follows from Toda's Theorem \cite{toda_pp_1991}
    since every oracle query in the computation of a ${\Poly}^{\SP}$-machine
	can be simulated by a nondeterministic polynomial-time machine
    guessing the answer and then verifying it with an oracle query to a
    problem in $\CEP$.
\end{proof}

\subsection{$\NP$-Hardness for Constraints of Bounded Treewidth}
\label{sec:threeConstraints}

The reduction used to prove $\NP^\CEP$-hardness
in the previous section uses the
graph given as input for $\existsEqualsThreeColouring$
as a constraint, which
has the side effect that the treewidth of the produced instances
is not bounded.
Moreover, the same reduction cannot be easily adapted to prove
$\NP$-hardness of $\FHomR$ for a class of graphs $\CF$
of bounded treewidth by simply considering input graphs of bounded treewidth:
the $\NP$-complete problem $\threecolouring$~\cite{garey_simplified_1976}
restricted to graphs of bounded treewidth
is polynomial-time solvable~\cite{arnborg_linear_1989}.
Hence, we need a different approach to such a reduction.
We reduce from the following problem $\SetSplitting$,
which is well-known to be $\NP$-complete \cite{lovasz_coverings_1973}.

\dproblem{$\SetSplitting$}
{A collection $\CC$ of subsets of a finite set $S$.}
{Is there a partition of $S$ into two subsets $S_1$ and $S_2$ such that no subset in $\CC$ is entirely contained in either $S_1$ or $S_2$?}

The idea is that we represent every element $i$ of $S$
by a vertex of colour $i$.
A set $T \in \CC$ is encoded by a star that has a leaf for every element of $T$
and a root vertex of some fixed colour that is distinct
from the colours used for elements of $S$.
Intuitively, the graph $G$ then consists of two stars that encode
the two sets $S_1$ and $S_2$. To ensure that no subset in $\CC$
is entirely contained in $S_1$ or $S_2$, we use the constraints to
require that there are no homomorphisms from our constraint graphs to $G$.

The idea sketched above produces a single constraint for every
set $T \in \CC$.
We can use the following trick to combine these constraints into a single one:
a graph $F$ consisting has exactly one homomorphism to $G$ if and only
if all its connected components have exactly one homomorphism to $G$.
Hence, if we choose $G$ to consist of two stars that encode $S_1$ and $S_2$
and also
add a star that has all elements of $S$ as its leaves,
we can instead require to have exactly one homomorphism
from all our constraint graphs to $G$
and, thus, combine all these constraint graphs into a single (disconnected) graph
from which we require to have exactly one homomorphism.

\begin{theorem}
	\label{th:hardnesscolouredConstant}
    Let $\CC\CS$ denote the class of all
    disjoint unions of coloured stars.
	Then, $\SetSplitting \le_p \FHomROf{\CC\CS}$
	and $\SetSplitting \le_p \BoundedFHomROf{\CC\CS}$,
	where the reduction only produces three constraints
	$\hom(F_1) = 1$, $\hom(F_2) = 2$, and $\hom(F_3) = 0$.
\end{theorem}
\begin{proof}
	Given a collection $\CC$ of subsets of a finite set $S$, we may
	assume that $S = [k]$ by re-labelling the elements of $S$.
	In the construction of the graphs $F_1, F_2, F_3$, we use
	the colours $1, \dots, k$ and also $B$ (\enquote{black}),
	$E$ (\enquote{everything}),
	and $P$ (\enquote{partition}).
	We construct these graphs as shown
    in \Cref{fig:hardnesscolouredConstantConstruction}
	and add the constraints $\hom(F_1) = 1$, $\hom(F_2) = 2$, $\hom(F_3) = 0$.
	Note that the constraint $\hom(F_1) = 1$ is equivalent
	to $\hom(F) = 1$ for every connected component $F$ of $F_1$.
	For $\BoundedFHomROf{\CC\CS}$, we set the size bound
	to $2k + 6$.

	\begin{figure}
		\centering
		\begin{subfigure}[b]{0.55\textwidth}
			\centering
			\begin{tikzpicture}[node distance = 0.5cm]
				\node[vertex, label={right:\small$B$}] (everyB) {};
				\node[vertex, label={left:\small$E$}, left = 0.3cm of everyB] (everyE) {};
				\node[vertex, label={above:\small$1$}, above left = 0.5cm and 0.25cm of everyB] (every1) {};
				\node[vertex, label={above:\small$k$}, above right = 0.5cm and 0.25cm of everyB] (everyk) {};
				\node[] (everydots) at ($(every1)!0.5!(everyk)$) {\scriptsize$\ldots$};
				\draw
				(everyB) edge (everyE)
				edge (every1)
				edge (everyk)
				;
\node[below = 0.3cm of everyB] {};

				\node[vertex, label={right:\small$B$}, right = 2.2cm of everyB] (partB) {};
				\node[vertex, label={left:\small$P$}, left = 0.3cm of partB] (partP) {};
				\node[vertex, label={above:\small$i$}, above = of partB] (parti) {};
				\node[below = 0.17cm of partB, overlay] (partDescription) {\small for $i \in [k]$};
				\draw
				(partB) edge (partP)
				edge (parti)
				;

				\node[vertex, label={right:\small$B$}, right = 2.2cm of partB] (setB) {};
				\node[vertex, above left = 0.5cm and 0.25cm of setB] (set1) {};
				\node[vertex, above right = 0.5cm and 0.25cm of setB] (setk) {};
				\node[] (setdots) at ($(set1)!0.5!(setk)$) {\scriptsize$\ldots$};
				\draw[very thick, decorate, decoration = {calligraphic brace}]
				([xshift = -2pt, yshift = 5pt]set1.west)
				-- node[above, yshift = 3pt] {\small$T$}
				([xshift = 2pt, yshift = 5pt]setk.east)
				;
				\node[below = 0.20cm of setB, overlay] (setDescription) {\small for $T \in \CC$};
				\draw
				(setB) edge (set1)
				edge (setk)
				;
			\end{tikzpicture}
			\caption{$\hom(F_1) = 1$.}
		\end{subfigure}
		\begin{subfigure}[b]{0.20\textwidth}
			\centering
			\begin{tikzpicture}[node distance = 0.5cm]
				\node[vertex, label={right:\small$B$}] (B) {};
				\node[vertex, label={left:\small$P$}, left = 0.3cm of B] (P) {};
\draw
				(B) edge (P)
				;
			\end{tikzpicture}
			\caption{$\hom(F_2) = 2$.}
		\end{subfigure}
		\begin{subfigure}[b]{0.20\textwidth}
			\centering
			\begin{tikzpicture}[node distance = 0.5cm]
				\node[vertex, label={right:\small$B$}] (B) {};
				\node[vertex, label={left:\small$E$}, above left = 0.1cm and 0.4cm of B] (E) {};
				\node[vertex, label={left:\small$P$}, below left = 0.1cm and 0.4cm of B] (P) {};
\draw
				(B) edge (P)
				(B) edge (E)
				;
			\end{tikzpicture}
			\caption{$\hom(F_3) = 0$.}
		\end{subfigure}
		\caption{The three constraints produced by
			the reduction of \Cref{th:hardnesscolouredConstant}.}
		\label{fig:hardnesscolouredConstantConstruction}
	\end{figure}
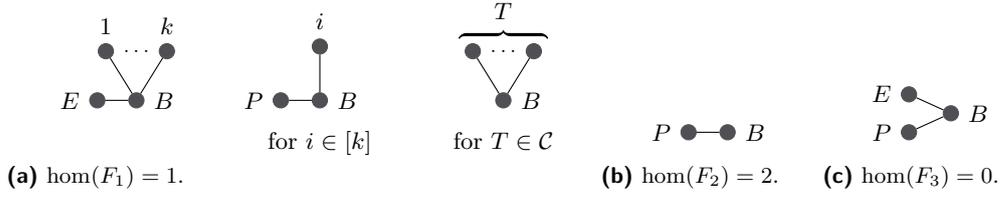

	Given a partition of $S$ into sets $S_1$ and $S_2$
	such that no subset in $\CC$ is entirely contained in either
	$S_1$ or $S_2$, the graph $G_{S_1, S_2}$ in \Cref{fig:hardnesscolouredConstantGraph}
	satisfies all constraints.
	Conversely, let $G$ be a coloured graph that satisfies all constraints.
	By the constraint $\hom(F_2) = 2$, the graph $F_2$ occurs exactly
	twice as a subgraph in $G$;
    call these occurrences $G_1$ and $G_2$.
    Let $S_1 \subseteq [k]$
    be the set of all $i \in [k]$
    such that a vertex of colour $i$ is connected to the $B$-vertex of $G_1$.
	If the $B$-vertex in $G_2 \subseteq G$
	is distinct from the $B$-vertex in $G_1 \subseteq G$,
    then let $S_2 \subseteq [k]$
    be the set of all $i \in [k]$
    such that a vertex of colour $i$ is connected to the $B$-vertex of $G_2$;
	otherwise, let $S_2 \coloneqq \emptyset$.
	The first constraint yields that the $i$-$B$-$P$-graph
	occurs exactly once as a subgraph of $G$ for every $i \in [k]$.
	By the second constraint, every $i$-$B$-$P$ graph
	has to be a supergraph of one of the two occurrences $G_1$ and $G_2$ of $F_2$.
	Hence, the sets $S_1$ and $S_2$ cover $S$.
	Moreover, as every $i$-$B$-$P$ graph occurs exactly once as a subgraph
	of $G$, the sets $S_1$ and $S_2$ have to be a partition of $S$.
	Finally, by the first constraint,
	the $1$-$k$-$B$-$E$-graph is a subgraph of~$G$.
	Observe that, for every $T \in \CC$,
	the set $T$ cannot be a subset of either $S_1$ or $S_2$
	as the $T$-$B$-graph
	occurs exactly once in $G$ by the first constraint
	and it already is a subgraph of the $1$-$k$-$B$-$E$-graph
	whose $B$-vertex has to be distinct from the $B$-vertices
	of the occurrences of $F_2$ by the third constraint.
\end{proof}

\begin{figure}
	\centering
	\begin{tikzpicture}[node distance = 0.5cm]
		\node[vertex, label={right:\small$B$}] (everyB) {};
		\node[vertex, label={left:\small$E$}, left = 0.3cm of everyB] (everyE) {};
		\node[vertex, label={above:\small$1$}, above left = 0.5cm and 0.25cm of everyB] (every1) {};
		\node[vertex, label={above:\small$k$}, above right = 0.5cm and 0.25cm of everyB] (everyk) {};
		\node[] (everydots) at ($(every1)!0.5!(everyk)$) {\scriptsize$\ldots$};
		\draw
		(everyB) edge (everyE)
		edge (every1)
		edge (everyk)
		;

		\node[vertex, label={right:\small$B$}, right = 2.5cm of everyB] (partOneB) {};
		\node[vertex, label={left:\small$P$}, left = 0.3cm of partOneB] (partOneP) {};
		\node[vertex, above left = 0.5cm and 0.25cm of partOneB] (partOne1) {};
		\node[vertex, above right = 0.5cm and 0.25cm of partOneB] (partOnek) {};
		\node[] (partOnedots) at ($(partOne1)!0.5!(partOnek)$) {\scriptsize$\ldots$};
		\draw[very thick, decorate, decoration = {calligraphic brace}]
		([xshift = -2pt, yshift = 5pt]partOne1.west)
		-- node[above, yshift = 3pt] {\small$S_1$}
		([xshift = 2pt, yshift = 5pt]partOnek.east)
		;
		\draw
		(partOneB) edge (partOneP)
		edge (partOne1)
		edge (partOnek)
		;

		\node[vertex, label={right:\small$B$}, right = 2.5cm of partOneB] (partTwoB) {};
		\node[vertex, label={left:\small$P$}, left = 0.3cm of partTwoB] (partTwoP) {};
		\node[vertex, above left = 0.5cm and 0.25cm of partTwoB] (partTwo1) {};
		\node[vertex, above right = 0.5cm and 0.25cm of partTwoB] (partTwok) {};
		\node[] (partTwodots) at ($(partTwo1)!0.5!(partTwok)$) {\scriptsize$\ldots$};
		\draw[very thick, decorate, decoration = {calligraphic brace}]
		([xshift = -2pt, yshift = 5pt]partTwo1.west)
		-- node[above, yshift = 3pt] {\small$S_2$}
		([xshift = 2pt, yshift = 5pt]partTwok.east)
		;
		\draw
		(partTwoB) edge (partTwoP)
		edge (partTwo1)
		edge (partTwok)
		;
	\end{tikzpicture}
	\caption{The graph $G_{S_1, S_2}$ constructed from $S_1$ and $S_2$ in
		the proof of
		\Cref{th:hardnesscolouredConstant}.}
	\label{fig:hardnesscolouredConstantGraph}
\end{figure}
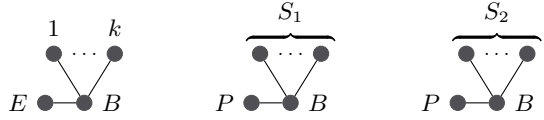

Note the following curiosity of
\Cref{th:hardnesscolouredConstant}:
the numbers in the constraints produced by the reduction are constant, i.e.\
they do not depend on the specific problem instance given as an input.
Hence, the hardness solely lies in the graphs and not the homomorphism numbers;
more specifically,
there is only a single constraint graph that depends on the input instance
and is the cause of hardness.

We again use Kneser graphs to turn this reduction into one
for uncoloured graphs.
We view the colours as numbers and use gadgets of Kneser graphs
that encode these numbers in binary.
This allows us to argue that the reduction is still correct,
where in particular, we have to argue that, if there is a
graph $G$ satisfying all constraints, then we can extract
a solution to the given $\SetSplitting$ instance from it;
this is not straightforward since such a graph $G$
does not have to adhere to our encoding of coloured graphs.
This then yields the following theorem,
which implies \Cref{thm:bounded-number-constraints}.
The proof can be found in \Cref{sec:uncolouredThreeConstraints}.

\begin{theorem}
	\label{th:ConstantEncoding}
    There is a class of graphs $\CF$ of bounded treewidth such that
	$\SetSplitting \le_p \FHomROf{\CF}$ and $\SetSplitting \le_p \BoundedFHomROf{\CF}$,
    where the number of constraints,
    the homomorphism numbers, and
    all constraint graphs but one are constant.
\end{theorem}

\subsection{$\NP$-Hardness for a Finite Set of Graphs}
\label{sec:finiteGraphSet}

The previous sections shows that restricting the constraint graphs
to be from a class $\CF$ of bounded treewidth and the number of
constraints to a constant is not enough to achieve tractability:
both $\FHomR$ and $\BoundedFHomR$ remain $\NP$-hard.
What happens if we go even further and consider a finite class $\CF$?
Then, the treewidth of $\CF$ is trivially bounded and so is the number of constraints.
At first glance, hardness
in this case seems unlikely since the reductions presented in the previous sections
rely heavily on encoding the input instance as the constraint graphs
and only used small constants for the homomorphism numbers.
Now, the input essentially consists just of homomorphism
numbers encoded in binary and, for $\BoundedFHomR$,
also the size of the desired graph encoded in unary.
This makes it all the more surprising that we can actually prove the
$\NP$-hardness of $\FHomR$.
We reduce from the following decision problem,
which only takes three natural numbers in binary encoding as input
and is $\NP$-complete \cite[Theorem~1]{manders_npcomplete_1976}:
\dproblem{$\QuadraticPolynomial$}
{Natural numbers $a$, $b$, and $c$ in binary encoding.}
{Are there natural numbers $x$ and $y$ such that $ax^2 + by = c$?}

The idea of the reduction is simple:
we encode the polynomial $ax^2 + by$ as a coloured star~$F$
from which we require exactly $c$ homomorphisms.
This star has a leaf of colour $A$ and two leaves of colour $X$
to encode the monomial $ax^2$.
Furthermore, it has a leaf of colour $B$ and of colour $Y$
to encode the monomial $by$.
Then, the sum $ax^2 + by$ is realised by encoding $x$ and $y$
as two separate components of $G$---additional
constraints are used to enforce that~$G$ has precisely two components,
that the first component has exactly $a$ leaves of colour $A$,
and that the second component has exactly $b$ leaves of colour $B$.

\newcommand{\pictoColouredEdgeOf}[2]{\hspace{-3pt}
    \tikz[baseline=-3pt]{
        \node[smallvertex, label={[label distance = -2pt]right:\footnotesize$#2$}] (R) {};\node[smallvertex, label={[label distance = -2pt]left:\footnotesize$#1$}, left = 0.1cm of R] (A) {};\draw (A) edge (R); }
    \hspace{-2pt}
}

\begin{theorem}
	\label{th:hardnesscolouredFixed}
	There is a finite set $\CF$ of coloured stars such that
    $\QuadraticPolynomial \le_p \FHomROf{\CF}$.
\end{theorem}
\begin{proof}
	We use the colours $R, A, X, B, Y, M_1$ and $M_2$.
	The main observation is that, for the graphs $F_\mathsf{poly}$ and $G_{a,b,x,y}$
	from \Cref{fig:hardnesscolouredFixed},
	we have $\hom(F_\mathsf{poly}, G_{a,b,x,y}) = a x^2 + by$.
	Note that, however, the size of $G_{a,b,x,y}$ is not polynomial in $\log a + \log b + \log c$, i.e.\
	in the size of an instance $(a,b,c)$ of \QuadraticPolynomial.
	Given an instance $(a,b,c)$ of \QuadraticPolynomial, we produce the following constraints
    and denote the set of all coloured graphs used in these constraints by $\CF$;
	note that $\CF$ is independent of the instance $(a,b,c)$:
	\begin{multicols}{3}
		\begin{alphaenumerate}
\item\label{hcf9} $\hom(\pictoColouredEdgeOf{A}{R}) = a + 1$,
\item\label{hcf11} $\hom(\pictoColouredEdgeOf{B}{R}) = b + 1$,
			\item\label{hcf1} $\hom(F_\mathsf{poly}) = c$,
			\item\label{hcf2} $\hom(\tikz[baseline=-3pt]{\node[smallvertex, label={[label distance = -2pt]right:\footnotesize$R$}] (R) {};}\hspace{-2pt}) = 2$,
			\item\label{hcf3} $\hom(\pictoColouredEdgeOf{M_1}{R}) = 1$,
			\item\label{hcf4} $\hom(\pictoColouredEdgeOf{M_2}{R}) = 1$,
			\item\label{hcf5} $\hom\Bigl(\hspace{-3pt}\tikz[baseline=-2.5pt]{
				\node[smallvertex, label={[label distance = -2pt]right:\footnotesize$R$}] (R) {};\node[smallvertex, label={[label distance = -2pt]left:\footnotesize$M_1$}, above left = 0.04cm and 0.12cm of R] (M) {};\node[smallvertex, label={[label distance = -2pt]left:\footnotesize$M_2$}, below left = 0.04cm and 0.12cm of R] (MT) {};\draw (M) edge (R)
				(MT) edge (R);
			}\hspace{-2pt}\Bigr) = 0$,
\item\label{hcf6} $\hom\Bigl(\hspace{-3pt}\tikz[baseline=2pt]{
				\node[smallvertex, label={[label distance = -2pt]right:\footnotesize$R$}] (R) {};\node[smallvertex, label={[label distance = -2pt]left:\footnotesize$M_1$}, left = 0.1cm of R] (M) {};\node[smallvertex, label={[yshift=3pt, label distance = -2pt]left:\footnotesize$B$}, above left = 0.12cm and 0.04cm of R] (B) {};\node[smallvertex, label={[yshift=3pt, label distance = -2pt]right:\footnotesize$Y$}, above right = 0.12cm and 0.04cm of R] (Y) {};\draw (R) edge (M)
				edge (B)
				edge (Y);
			}\hspace{-2pt}\Bigr) = 1$,
			\item\label{hcf7} $\hom\Bigl(\hspace{-3pt}\tikz[baseline=2pt]{
				\node[smallvertex, label={[label distance = -2pt]right:\footnotesize$R$}] (R) {};\node[smallvertex, label={[label distance = -2pt]left:\footnotesize$M_2$}, left = 0.1cm of R] (M) {};\node[smallvertex, label={[yshift=3pt, label distance = -2pt]left:\footnotesize$A$}, above left = 0.12cm and 0.04cm of R] (B) {};\node[smallvertex, label={[yshift=3pt, label distance = -2pt]right:\footnotesize$X$}, above right = 0.12cm and 0.04cm of R] (Y) {};\draw (R) edge (M)
				edge (B)
				edge (Y);
			}\hspace{-2pt}\Bigr) = 1$.
\end{alphaenumerate}
	\end{multicols}

	If $(a,b,c)$ is an instance of $\QuadraticPolynomial$
    with $x,y \in \NN$ such that $ax^2 + by = c$,
	then $G_{a,b,x,y}$
    from \Cref{fig:hardnesscolouredFixedGraph} satisfies all these constraints.
	Conversely, if there is a coloured graph $G$ that satisfies all constraints,
	we know by (\ref{hcf2})--(\ref{hcf5}) that there are two $R$-coloured vertices $v_1$ and $v_2$ such that
	$v_1$ is connected to an $M_1$-coloured vertex but not an $M_2$-coloured vertex
	and $v_2$ is connected to an $M_2$-coloured vertex but not an $M_1$-coloured vertex.
	By (\ref{hcf6}) and (\ref{hcf7}), $v_1$ has exactly one $B$-coloured neighbor and exactly one $Y$-coloured neighbor
    and $v_2$ has exactly one $A$-coloured neighbor and exactly one $X$-coloured neighbor.
    Hence, by (\ref{hcf9}) and (\ref{hcf11}), $v_1$ has exactly $a$ neighbors of colour $A$
    and $v_2$ has exactly $b$ neighbors of colour $B$.
	Let $x$ be the number of $X$-coloured neighbors of $v_1$
	and $y$ be the number of $Y$-coloured neighbors of $v_2$.
	Then, we have $c = \hom(F_{\mathsf{poly}}, G) = ax^2 + by$.
\end{proof}
	\begin{figure}
	\centering
	\begin{subfigure}{0.3\textwidth}
		\centering
		\begin{tikzpicture}[node distance = 0.5cm]
			\node[vertex, label={right:\small$R$}] (R) {};
			\node[vertex, label={below:\small$A$}, below left = 0.5cm and 0.25cm of R] (A) {};
			\node[vertex, label={below:\small$X$}, below right = 0.5cm and -0.15cm of R] (X1) {};
			\node[vertex, label={below:\small$X$}, below right = 0.5cm and 0.25cm of R] (X2) {};
			\node[vertex, label={above:\small$B$}, above left = 0.5cm and 0.25cm of R] (B) {};
			\node[vertex, label={above:\small$Y$}, above right = 0.5cm and 0.25cm of R] (Y) {};
			\draw
			(R) edge (A)
			edge (X1)
			edge (X2)
			edge (B)
			edge (Y)
			;

\draw[white, pen colour = {white}, very thick, decorate, decoration = {calligraphic brace}]
			([xshift = 2pt, yshift = -16pt]X2.east)
			-- node[below, yshift = -3pt] {$b$}
			([xshift = -2pt, yshift = -16pt]A.west)
			;
		\end{tikzpicture}
		\caption{The graph $F_\mathsf{poly}$.}
		\label{fig:hardnesscolouredFixedConstruction}
	\end{subfigure}
	\begin{subfigure}{.6\textwidth}
		\centering
		\begin{tikzpicture}
			\node[vertex, label={right:\small$R$}] (leftR) {};
			\node[vertex, label={below:\small$X$}, below right = 0.5cm and 0.1cm of leftR] (leftX1) {};
			\node[vertex, label={below:\small$X$}, below right = 0.5cm and 0.7cm of leftR] (leftXx) {};

			\node[vertex, label={below:\small$A$}, below left = 0.5cm and 0.7cm of leftR] (leftA1) {};
			\node[vertex, label={below:\small$A$}, below left = 0.5cm and 0.1cm of leftR] (leftAa) {};

			\node[vertex, label={above:\small$Y$}, above right = 0.5cm and 0.25cm of leftR] (leftY) {};
			\node[vertex, label={above:\small$B$}, above left= 0.5cm and 0.25cm of leftR] (leftB) {};
			\node[vertex, label={left:\small$M_1$}, left = 0.3cm of leftR] (leftM) {};

			\node[] (Xdots) at ($(leftX1)!0.5!(leftXx)$) {\scriptsize$\ldots$};
			\draw[very thick, decorate, decoration = {calligraphic brace}]
			([xshift = 2pt, yshift = -16pt]leftXx.east)
			-- node[below, yshift = -3pt] {$x$}
			([xshift = -2pt, yshift = -16pt]leftX1.west)
			;

			\node[] (Adots) at ($(leftA1)!0.5!(leftAa)$) {\scriptsize$\ldots$};
			\draw[very thick, decorate, decoration = {calligraphic brace}]
			([xshift = 2pt, yshift = -16pt]leftAa.east)
			-- node[below, yshift = -3pt] {$a$}
			([xshift = -2pt, yshift = -16pt]leftA1.west)
			;

			\draw
			(leftR) edge (leftY)
			edge (leftB)
			edge (leftM)
			edge (leftX1)
			edge (leftXx)
			edge (leftA1)
			edge (leftAa)
			;

			\node[vertex, label={right:$R$}, right = 4cm of leftR] (rightR) {};
			\node[vertex, label={below:$Y$}, below right = 0.5cm and 0.1cm of rightR] (rightX1) {};
			\node[vertex, label={below:$Y$}, below right = 0.5cm and 0.7cm of rightR] (rightXx) {};

			\node[vertex, label={below:$B$}, below left = 0.5cm and 0.7cm of rightR] (rightA1) {};
			\node[vertex, label={below:$B$}, below left = 0.5cm and 0.1cm of rightR] (rightAa) {};

			\node[vertex, label={above:$X$}, above right = 0.5cm and 0.25cm of rightR] (rightY) {};
			\node[vertex, label={above:$A$}, above left = 0.5cm and 0.25cm of rightR] (rightB) {};
			\node[vertex, label={left:$M_2$}, left = 0.3cm of rightR] (rightM) {};

			\node[] (Xdots) at ($(rightX1)!0.5!(rightXx)$) {\scriptsize$\ldots$};
			\draw[very thick, decorate, decoration = {calligraphic brace}]
			([xshift = 2pt, yshift = -16pt]rightXx.east)
			-- node[below, yshift = -3pt] {$y$}
			([xshift = -2pt, yshift = -16pt]rightX1.west)
			;

			\node[] (rightAdots) at ($(rightA1)!0.5!(rightAa)$) {\scriptsize$\ldots$};
			\draw[very thick, decorate, decoration = {calligraphic brace}]
			([xshift = 2pt, yshift = -16pt]rightAa.east)
			-- node[below, yshift = -3pt] {$b$}
			([xshift = -2pt, yshift = -16pt]rightA1.west)
			;

			\draw
			(rightR) edge (rightY)
			edge (rightB)
			edge (rightM)
			edge (rightX1)
			edge (rightXx)
			edge (rightA1)
			edge (rightAa)
			;

		\end{tikzpicture}
		\caption{The graph $G_{a,b,x,y}$ constructed from $a,b,x,y \in \mathbb{N}$.}
		\label{fig:hardnesscolouredFixedGraph}
	\end{subfigure}
	\caption{The most important graphs used in the reduction of \Cref{th:hardnesscolouredFixed}.}
	\label{fig:hardnesscolouredFixed}
\end{figure}
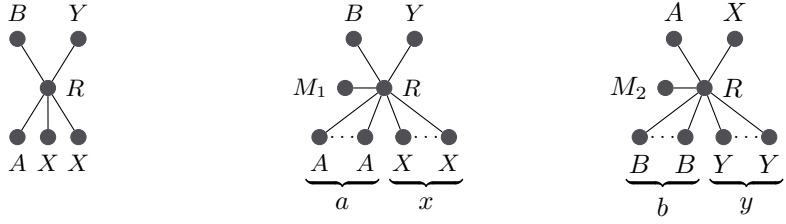

We can again use Kneser graphs to obtain a reduction for uncoloured graphs,
which implies the hardness stated in \Cref{thm:finite-set-hardness}.
The proof can be found in \Cref{sec:uncolouredFiniteSet}.

\begin{theorem}
    \label{th:hardnessFixed}
	There is a finite set $\CF$ of graphs such that
    $\QuadraticPolynomial \le_p \FHomROf{\CF}$.
\end{theorem}

Can such a reduction also be used to prove $\NP$-hardness of $\BoundedFHomR$
for a finite class $\CF$?
This is not possible, since for a fixed graph $F$,
the number of homomorphisms from $F$ to a graph $G$
is polynomial in the order of $G$.
Hence, with a finite set $\CF$ of graphs and a graph $G$ of order up to $n$,
we can only realise polynomially many distinct homomorphism numbers,
while the hardness of solving the equation $ax^2 + by = c$ stems from the fact
that there is an exponential number of solution candidates.
This implies that $\BoundedFHomROf{\CF}$ is sparse for every finite $\CF$,
which means that it cannot be $\NP$-hard unless $\Poly = \NP$~\cite{mahaney_sparse_1982}.

\begin{theorem}
	If \BoundedFHomROf{\CF}\ is $\NP$-hard for a finite set of graphs $\CF$, then $\Poly = \NP$.
\end{theorem}

\subsection{Subgraph Counts}
\label{sec:hardnessSubgraphCounts}

While subgraphs counts can be expressed as linear combinations
of homomorphism numbers via injective homomorphism numbers,
cf.\ \cite{curticapean_homomorphisms_2017},
adapting our reductions to subgraph counts is not as straightforward.
There is no obvious way of adapting the reduction of \Cref{th:NPCEPReduction}
since non-injectivity is crucial to encode colourability of graphs.
The reduction of \Cref{th:hardnesscolouredConstant}
can partially be salvaged by producing individual constraints
instead of only three constraint.
Then, all constraint graphs are \emph{colourful}, which means that
their subgraph counts are homomorphism counts,
which can be computed efficiently. However, the gadget construction used in
\Cref{th:ConstantEncoding} then produces constraint graphs
of unbounded vertex-cover number.
Hence, this reduction is uninteresting since
the determining factor for tractability of
subgraphs counts is the vertex-cover number~\cite{curticapean_complexity_2014a}.

The results for finite $\CF$ transfer to subgraphs
and are meaningful since subgraph counts can
trivially be computed in polynomial time in this case.
However, for subgraph counts, the reduction of \Cref{thm:finite-set-hardness}
encodes the binomial equation $a \binom{x}{2} + by = c$
instead of $ax^2 + by = c$.
Luckily, the problem $\BPoly$ of solving this equation is still $\NP$-complete,
cf.\ \Cref{sec:binomialEquations}.
Moreover, $\BoundedFSubR$ is still sparse for every finite class of graphs $\CF$.

\begin{theorem}
    \label{th:subgraphHardnessFixed}
	There is a finite set $\CF$ of graphs such that
    $\BPoly\le_p \FSubROf{\CF}$.
	If \BoundedFSubROf{\CF}\ is $\NP$-hard for a finite set of graphs $\CF$, then $\Poly = \NP$.
\end{theorem}

\section{Towards Tractability: Reconstructing Clique Counts}

In this section, we show that \BoundedFHomR is tractable when the only constraint graph is a clique and the constraints are in a certain range.
The proofs rely on number-theoretic insights and tailored combinatorial constructions.
Using a number-theoretic result by Kamke~\cite{kamke_verallgemeinerungen_1921}, we show that for almost all sensible $n, h, k \in \mathbb{N}$ there exists a graph $G$ on slightly more than $n$ vertices containing $h$ copies of the $k$-vertex clique~$K_k$ as a subgraph.

\begin{theorem}[{Kamke~\cite{kamke_verallgemeinerungen_1921}, cf.\@ \cite[Theorem~11.10]{nathanson_elementary_2000}}] \label{thm:numbertheory}
	There exists a function $\gamma \colon \mathbb{N} \to \mathbb{N}$ such that for every $k \geq 1$ and $n \geq 1$ there exist $a_1, \dots, a_{\gamma(k)} \in \mathbb{N}$ such that $n = \sum_{i=1}^{\gamma(k)} \binom{a_i}{k}$.
\end{theorem}

Ne\v{c}aev \cite{necaev_representation_1953} showed that $\gamma(k)$ is of order $O(k\log k)$ and gave a similar lower bound in~\cite{necaev_question_1976}. 
Specific values of $\gamma$ include $\gamma(1) = 1$, Gau{\ss}' Eureka Theorem, cf.\ \cite{ono_representation_1995}, stating $\gamma(2) = 3$ and the unproven Tetrahedral Numbers Conjecture of Pollock~\cite{pollock_extension_1851} asserting $\gamma(3) = 5$.

\thmcliques*

\begin{proof}
	Let $\gamma$ denote the function from \cref{thm:numbertheory}.
	For every fixed $k$, the proof is by induction on $n$. For $n \leq k$, the claim is trivial.
	Suppose subsequently that $n > k$.
	Inductively, we may suppose that there exists a graph~$G$ on $n-1 + \gamma(k-1) - 1$ vertices with $h \leq \binom{n-1}{k}$ copies of $K_k$. One may add an isolated vertex to obtain a graph on $n + \gamma(k-1)-1$ vertices with $h$ copies of $K_k$, as desired.
	
	Thus, it remains to construct a graph with $\binom{n-1}{k} < h \leq \binom{n}{k}$ copies of $K_k$ and $n+\gamma(k-1)-1$ vertices. Write $h' \coloneqq h - \binom{n-1}{k} \leq \binom{n-1}{k-1}$. 
	By \cref{thm:numbertheory}, there exist non-negative integers $a_1, \dots, a_{\gamma(k-1)}$ such that $h' = \sum_{i=1}^{\gamma(k-1)} \binom{a_i}{k-1}$. 
	It can be easily seen that $\binom{h}{k} > \binom{n}{k}$ for all integers $h > n \geq k \geq 1$.
	Hence, $\binom{a_i}{k-1} \leq h' \leq \binom{n-1}{k-1}$ implies that $a_i \leq n-1$ for all $1 \leq i \leq \gamma(k-1)$.
	
	Define the graph $G$ by taking the disjoint union of a clique $K_{n-1}$ and fresh vertices $v_1, \dots, v_{\gamma(k-1)}$.
	For $1 \leq i \leq \gamma(k-1)$, the vertex $v_i$ is connected to an arbitrary selection of $a_i$ many vertices of the clique. Note that this adds $\binom{a_i}{k-1}$ copies of $K_k$ to the graph. The resulting graph on $n-1 + \gamma(k-1)$ vertices satisfies $\sub(K_k, G) = \binom{n-1}{k} + \sum_{i=1}^{\gamma(k-1)} \binom{a_i}{k-1} = h$. For an example with $k=3$, see \cref{fig:clique-construct}.
\end{proof}

\begin{figure}
	\centering
\begin{tikzpicture}
	\node[vertex] (A) {};
	\node[vertex, below right = 0.3cm and 0.4cm of A] (B) {};
	\node[vertex, below right = 0.5cm and 0.1cm of B] (C) {};
	\node[vertex, below left = 0.5cm and 0.1cm of C] (D) {};
	\node[vertex, below left = 0.3cm and 0.4cm of D] (E) {};
	\node[vertex, left = 0.4cm of A] (A2) {};
	\node[vertex, left = 0.4cm of E] (E2) {};
	
	\node[label={center:$\ldots$}, left = 1.8cm of C] (Dots) {};

	\node[vertex, right = 1.5cm of C, label={right:$v_2$ with $\deg(v_2)=a_2$}] (Z1) {};
	\node[vertex, below = 0.5cm of Z1, label={right:$v_3$ with $\deg(v_3)=a_3$}] (Z2) {};
	\node[vertex, label={right:$v_1$ with $\deg(v_1)=a_1$}, above = 0.5cm of Z1] (Z3) {};
	
	\draw[very thick, decorate, decoration = {calligraphic brace}]
	([xshift = 5pt, yshift = -40pt]C.east)
	-- node[below, yshift = -3pt] {$K_{n-1}$}
	([xshift = -5pt, yshift = -40pt]Dots.west)
	;
	\draw
	(E2) edge (A2) 
	edge (A)
	edge (B)
	edge (C)
	edge (D)
	edge (E);
	\draw
	(A2) edge (A)
	edge (B)
	edge (C)
	edge (D)
	edge (E);
	\draw
	(A) edge (B)
	edge (C)
	edge (D)
	edge (E);
	\draw
	(B) edge (C)
	edge (D)
	edge (E);
	\draw
	(C)	edge (D)
	edge (E);
	\draw
	(D) edge (E);
	
	\draw
	(Z3) edge (A)
	edge (B)
	edge (C);

\end{tikzpicture}
\caption{
	Example for $k=3$, building a graph $G$ with $n-1 + \gamma(k-1)$ vertices and $\sub(\smallKthree,G) = h > \binom{n-1}{3}$. Since $\gamma(2)=3$, there are $a_1, a_2, a_3$ such that $h - \binom{n-1}{3} = \binom{a_1}{2}+\binom{a_2}{2}+\binom{a_3}{2}$. Connecting a fresh vertex $v_i$ with $a_i$ vertices from the $K_{n-1}$, adds $\binom{a_i}{2}$ subgraphs $\smallKthree$ to $G$.
}	
\label{fig:clique-construct}
\end{figure}
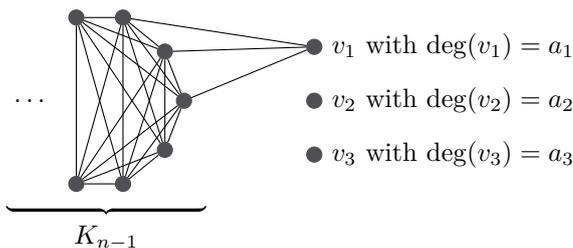

For the special case of $k = 3$, i.e.\@ $\smallKthree$, we can do slightly better than in \cref{thm:reconstruct-cliques}.
While \cref{thm:reconstruct-cliques} requires $\gamma(2)-1 = 2$ extra vertices to realise any sensible~$h$, we show in \cref{thm:triangles} that for large $n$ one additional vertex suffices. While \cref{thm:reconstruct-cliques} builds on an $(n-1)$-vertex clique to which new vertices and edges are added, for \cref{thm:triangles} we start with an $(n+1)$-vertex clique and remove edges to realise the precise subgraph count. The proof of \cref{thm:triangles} can be found in \cref{sec:triangles}.

\begin{restatable}{theorem}{thmtriangles} \label{thm:triangles}
	For every $n \geq 130$ and $h \leq \binom{n}{3}$, there is a graph $G$ on $n+1$ vertices such that $\sub(\smallKthree, G) = h$.
\end{restatable}

\section{Parametrised Complexity}
\label{sec:parameterized}

For graph classes $\CF$ and $\CG$,
we consider the parametrised version \pFGHomR of the homomorphism reconstructability problem.
For an instance $(F_1, h_1), \dots, (F_m, h_m)$, the parameter is $k \coloneqq \sum_{i=1}^m \lvert V(F_i)\rvert$.
We aim for an fpt-algorithm, i.e.\@ an algorithm that runs in $f(k) \operatorname{polylog}(h_1, \dots, h_m)$ for some computable function $f$.
By \cref{th:hardnessFixed}, \pFGHomR is para-$\NP$-hard and thus we cannot expect to obtain an fpt-algorithm  unless $\Poly = \NP$, cf.\@ \cite[Corollary~2.13]{flum_parameterized_2006}, which means that we have to restrict the problem in some way.
Surprisingly, it turns out that a certain restriction of \pFGHomR and of the analogous \pFGSubR are in $\FPT$.
Curiously, in $\FPT$, one cannot even count homomorphisms or subgraphs from arbitrary graphs \cite{dalmau_complexity_2004,curticapean_complexity_2014a}.
Our algorithm has to make do with the integers from the input and cannot construct the graph $G$ explicitly.

Given constraint graphs $\mathcal{I} \subseteq \mathcal{F}$,
the overall strategy is to compute in time only depending on $\mathcal{I}$ a data structure representing the (infinite) set of all reconstructable vectors of homomorphism counts
\begin{equation} \label{eq:defR}
	R(\mathcal{I})  \coloneqq \left\{\hom(\mathcal{I}, G) \in \mathbb{N}^{\mathcal{I}} \mid G \in \mathcal{G} \right\}
\end{equation}
or analogously of subgraph counts.
This data structure is required to admit a polynomial-time procedure for testing whether a given vector $h \in \mathbb{N}^{\mathcal{I}}$ is in this set. 
We identify various combinatorial conditions sufficient for guaranteeing the feasibility of this approach.
To start with, we impose the following conditions on the graph classes $\mathcal{F}$ and $\mathcal{G}$:
\begin{proviso} \label{proviso:fpt}
\begin{romanenumerate}
	\item membership in $\mathcal{G}$ is decidable,\label{fpt1}
	\item $\mathcal{G}$ is closed under taking induced subgraphs,\label{fpt3}
	\item $\mathcal{G}$ is closed under disjoint unions,\label{fpt2}
	\item all $F \in \mathcal{F}$ are connected.\label{fpt4}
\end{romanenumerate}
\end{proviso}
Items (\ref{fpt2}) and (\ref{fpt4}) imply that the 
set $R(\mathcal{I})$ of all reconstructable vectors 
is closed under addition. Indeed,
for all connected graphs $F$ and graphs $G$ and $H$ it holds that $\hom(F, G + H) = \hom(F, G) + \hom(F, H)$ and $\sub(F, G + H) = \sub(F, G) + \sub(F, H)$.
Thus, we can use the vectors realised by small graphs, i.e.\@ those which can be inspected in $\FPT$ time, to construct vectors realised by bigger graphs. 
More formally, writing
\begin{equation} \label{eq:defS}
	S(\mathcal{I}) \coloneqq \left\{\hom(\mathcal{I}, G) \in \mathbb{N}^{\mathcal{I}}\ \middle|\ G \in \mathcal{G} \text{ with } \lvert V(G) \rvert \leq \max_{I \in \mathcal{I}} \lvert V(I) \rvert \right\},
\end{equation}
it holds that the set of all finite linear combinations of elements in $S(\mathcal{I})$ with coefficients from~$\mathbb{N}$ is contained in $R(\mathcal{I})$, i.e.\@ $\mathbb{N}S(\mathcal{I}) \subseteq R(\mathcal{I})$.
The challenge is to ensure that all reconstructable vectors can be constructed in this way. 

We require item (\ref{fpt3}) to relate vectors 
realised by large graphs to those realised by small 
graphs, cf.\@ \cref{lem:subsmall,lem:homindsum}. 
However, this assumption is not sufficient to yield 
an fpt-algorithm.
To that end, we make further assumptions which ensure that the set of realised vectors is in a sense linear and thus admits the aforementioned desired data structure.
Combinatorially, these assumptions mean that the various constraints may not interact non-trivially.

\subsection{Equi-Size Subgraph Constraints}

The restriction we impose on \pFGSubR to put it in $\FPT$ is the following:
\begin{theorem}\label{thm:fpt2}
	For graph classes $\mathcal{F}$, $\mathcal{G}$ as in \cref{proviso:fpt},
	the following problem is in~$\FPT$:
	\pproblem{\subsize}{Pairs $(F_1, h_1), \dots, (F_m,h_m) \in \CF \times \NN$ with $\lvert V(F_1) \rvert = \dots = \lvert V(F_m) \rvert \eqqcolon k$}{$km$}
	{Is there a $G \in \CG$ such that $\sub(F_i, G) = h_i$ for every $i \in [m]$?}
\end{theorem}

Given an instance $\mathcal{I} \subseteq \mathcal{F}$, define $R(\mathcal{I})$ and $S(\mathcal{I})$ as in \cref{eq:defR,eq:defS} but with $\sub$ instead of $\hom$.
As argued in the previous section, we have $\mathbb{N}S(\mathcal{I}) \subseteq R(\mathcal{I})$.
In the context of \subsize, the following \cref{lem:subsmall} yields that  $\mathbb{N}S(\mathcal{I}) = R(\mathcal{I})$.

\begin{lemma} \label{lem:subsmall}
	Let $F$ and $G$ be graphs. Then
	\[
	\sub(F, G) = \sum_{H \text{ s.t.\@ } \lvert V(H) \rvert = \lvert V(F) \rvert} \sub(F, H) \indsub(H, G)
	\]
	where the sum ranges over all isomorphism types of graphs $H$ on $\lvert V(F) \rvert$~vertices.
\end{lemma}
Indeed, by \cref{lem:subsmall}, every vector $\sub(\mathcal{I}, G) \in R(\mathcal{I})$ is an $\mathbb{N}$-linear combination of the vectors $\sub(\mathcal{I}, H)$ where $H$ has exactly $k$ vertices.
It is crucial that all graphs in $\mathcal{I}$ have the same number of vertices since otherwise any statement akin to \cref{lem:subsmall} would involve negative coefficients stemming from Inclusion--Exclusion.
In virtue of \cref{lem:subsmall}, testing membership in $R(\mathcal{I})$ reduces to solving a system of linear equations over $\mathbb{N}$.

\begin{example} \label{ex:subrec}
	Let $p, c \geq 0$ be integers. There exists a graph $G$ with $\sub(\smallPtwo, G) = p$ and $\sub(\smallKthree, G) = c$ if and only if $p \geq 3c$.
\end{example}
\begin{proof}
	By \cref{lem:subsmall}, there exists a graph $G$ with $\sub(\smallPtwo, G) = p$ and $\sub(\smallKthree, G) = c$ if and only if the system $\left( \begin{smallmatrix}
		p \\c
	\end{smallmatrix} \right) = \left( \begin{smallmatrix}
		1 & 3 \\ 0 & 1
	\end{smallmatrix} \right) \left( \begin{smallmatrix}
		x \\y
	\end{smallmatrix} \right)$ has solutions $x, y \in \mathbb{N}$, i.e.\@
	$\sub(\smallPtwo, y\smallKthree + x\smallPtwo) = x + 3y$ and $\sub(\smallKthree, y\smallKthree + x\smallPtwo) = y$. 
	The columns of this matrix correspond to the two graphs on three vertices which have a subgraph $\smallPtwo$ or $\smallKthree$, namely $\smallPtwo$ and $\smallKthree$. 
	Solving this system yields that $y = c$ and $x = p - 3c$, as desired.
\end{proof}

The matrix constructed in \cref{ex:subrec} via \cref{lem:subsmall} can clearly be computed in $\FPT$. It remains to solve its system of linear equations, which is also possible in $\FPT$ \cite{damschke_sparse_2013}.

\begin{proof}[Proof of \cref{thm:fpt2}]
	Let $\mathcal{I} \subseteq \mathcal{F}$ denote the instance.
	As observed above, $\mathbb{N}S(\mathcal{I}) = R(\mathcal{I})$.	
	Write $\mathcal{H}$ for the set of all isomorphism types of graphs on $k$ vertices. 
	Write $A \in \mathbb{N}^{\mathcal{I} \times \mathcal{H}}$ for the matrix with entries $\sub(F, H)$ for $(F, H) \in \mathcal{I} \times \mathcal{H}$. This matrix can be computed in $\FPT$.
	Then $b = (h_1, \dots, h_m) \in R(\mathcal{I})$ if and only if $Ax = b$ has a solution over the non-negative integers. Testing the latter condition can be done in $\FPT$ by \cite{damschke_sparse_2013}, cf.\@ \cref{thm:damschke}.
\end{proof}

\subsection{A Single Homomorphism Constraint}

For \pFGHomR, the restriction to ensure fixed parameter-tractability is more restrictive.

\thmfpt*

Define $R(F)$ and $S(F)$ as in \cref{eq:defR,eq:defS} replacing $\mathcal{I}$ by the singleton set $\{F\}$.
As before, $\mathbb{N}S(F) \subseteq R(F)$.
Dealing with \SCR is more complicated than tackling \subsize in the sense that we will not be able to prove that $\mathbb{N}S(F) = R(F)$.
In fact, these sets only coincide for large enough numbers.
The key combinatorial identity is the following, whose proof is deferred to \cref{app:parameterised}:
\begin{lemma}\label{lem:homindsum}
	Let $F$ be a graph on $k$ vertices.
	Then for all graphs $G$ on more than $k$ vertices,
\[
	\hom(F, G) = \sum_{H \text{ s.t.\@ } \abs{V(H)} \leq k} \hom(F, H) \indsub(H, G) (-1)^{k-\abs{V(H)}}\binom{\abs{V(G)}-\abs{V(H)}-1}{k-\abs{V(H)}},
	\]
	where the sum ranges over all isomorphism types of graphs $H$ on at most $k$~vertices.
\end{lemma}

\Cref{lem:homindsum} yields that $R(F) \subseteq \mathbb{Z}S(F)$, i.e.\@ every realised number is a linear coefficient of numbers in $S(F)$ with (not necessarily non-negative) integer coefficients.
What allows us to obtain \cref{thm:fpt} is the purely number-theoretic observation that $\mathbb{N}S(F)$ and $\mathbb{Z}S(F)$ coincide on sufficiently large numbers.
\Cref{lem:bezout} is based on B\'ezout's identity and proven in  \cref{app:parameterised}.

\begin{lemma}\label{lem:bezout}
	Given $y_1, \dots, y_n \in \mathbb{N}$, one can compute integers $y$ and $N$ such that
	\[
	X \cap \{N, N+1, \dots\} = y\mathbb{N} \cap \{N, N+1, \dots\}
	\]
	where $X \coloneqq \mathbb{N}\{y_1, \dots, y_n\}$.
\end{lemma}

\begin{example} \label{ex:bezout}
	Let $y_1 = 6$ and $y_2 = 16$. Their greatest common divisor is $2$. The set of their $\mathbb{N}$-linear combinations is
	$
	X = \mathbb{N}\{6,16\}= 2\mathbb{N} \setminus \{2,4,8,10,14,20\},
	$
	i.e.\@ the set of all even numbers except $2,4,8,10,14,20$.
\end{example}

This concludes the preparations for the proof of \cref{thm:fpt}:

\begin{proof}[Proof of \cref{thm:fpt}]
	The algorithm operates as follows:
	Compute $S(F)$ in time only depending on $|V(F)|$.
	Let $y$ denote the greatest common divisor of the numbers in $S(F)$.
	By \cref{lem:homindsum,lem:bezout},  there exists a number $N$ only depending on $\abs{V(F)}$ such that for every $h \geq N$ there exists a graph $G$ with $\hom(F, G) = h$ if and only if $h$ is a multiple of $y$. 
	This settles the question for all $h \geq N$. 
	It remains to consider the case $h < N$.
	By \cref{lem:locality}, it suffices to consider graphs $G$ of size bounded in $\abs{V(F)}$ to conclude.
\end{proof}

To illustrate our algorithm, we include an example:
\begin{example} 
 Consider the constraint graph $F = \smallDiamond$.
 Enumerating all graphs on at most $4$ vertices yields that $S(\smallDiamond) = \{0, 6, 16, 48\}$. 
 For example, $\hom(\smallDiamond, \smallKthree) = 6$, 
 $\hom(\smallDiamond, \smallDiamond) = 16$, 
 and $\hom(\smallDiamond, \smallKfour) = 48$.
 Hence, $R(\smallDiamond)$ is a subset of the set in \cref{ex:bezout}.
 It remains to check the finitely many exceptions $2,4,8,10,14,20$. By \cref{lem:locality}, this can be done by inspecting graphs on at most $20 \cdot 4 = 80$ vertices.
\end{example}

\section{Conclusion}

This paper provides the first systematic study of the homomorphism reconstructability problem.
Our results show that this deceivingly simple-to-state problem generally is hard---not
only in terms of its computational complexity but also
in terms of finding efficient algorithms for the simplest of cases,
being subject to intricate phenomena from combinatorics and number theory.
The following questions remain open and warrant further investigation:
\begin{itemize}
	\item Is $\BoundedFHomROf{\CF}$ $\NP^{\SP}$-complete for every class 
	of unbounded treewidth $\CF$, analogous to the $\SP$-completeness of 
	$\SHom(\CF)$ \cite{dalmau_complexity_2004}?
	Is $\FHomROf{\CG}$ $\NEXP$-complete for the class of all graphs $\CG$?
    Is there a graph class $\CF$ for which $\FHomROf{\CF}$ is not only 
    $\NP$-hard but actually $\NP$-complete?
	\item While \SCR is fpt, \cref{thm:finite-set-hardness} implies that there exists a constant $C$ such that \pFGHomR restricted to instances with $\leq C$ constraints is para-$\NP$-hard.
	What is the minimal such $C$? Is \pFGHomR with two connected constraints fpt?
    Is there a sharp threshold from fixed-parameter tractability to para-$\NP$-hardness?
	\item
    The proof of \Cref{th:hardnesscolouredConstant} suggests that in some cases the number of constraints can be decreased by increasing the number of connected components of the constraint graphs.
    Notably, a crucial ingredient in \cref{proviso:fpt} is that the constraint graphs are connected.
    How do the parameters number of constraints and number of connected components affect the complexity of \FGHomR?
    What does the complexity hierarchy under these two parameters look like?
	\item In \cite{freedman_reflection_2007,lovasz_semidefinite_2009}, the functions $f \colon \mathcal{G} \to \mathbb{N}$ which are of the form $f = \hom(-, G)$ for some graph~$G$ were characterised.
	Here, $\mathcal{G}$ denotes the class of all graphs. 
	For which finite graph classes $\mathcal{I}$ does a 
	characterisation of functions $f \colon \mathcal{I} \to 
	\mathbb{N}$ of this form exist?
	Our \cref{thm:finite-set-hardness} implies that in some cases deciding 
	whether a given $f$ is of this form is $\NP$-hard.
    \item  Are there non-trivial examples of combinations of constraint graphs for which reconstructability is tractable?
    Is there an effective description of the yellow area in~\cref{fig:plot-triangles}?
    \item Is $\FGHomR$ self-reducible \cite{schnorr_optimal_1976}? That is, can we efficiently
    construct a graph $G$ that realises the given constraints
    if we have access to an oracle for $\FGHomR$?
    \item What is the computational complexity of deciding whether homomorphism constraints
    are \emph{approximately} reconstructable?
    \item How can one sample graphs satisfying homomorphism constraints uniformly at random?
\end{itemize}

\newpage
\appendix

\clearpage{}\section{Extended Preliminaries}
\label{app:prelim}

Let us briefly state the full formal definitions of the decision problems
that we only briefly described in the main body of the paper.

\dproblem{\FGSubR}
{Pairs $(F_1, h_1), \dots, (F_m,h_m) \in \CF \times \NN$ where $h_1, \dots, h_m$ are given in binary.}
{Is there a graph $G \in \CG$ such that $\sub(F_i, G) = h_i$ for every $i \in [m]$?}

\dproblem{\BoundedFGHomR}
{Pairs $(F_1, h_1), \dots, (F_m,h_m) \in \CF \times \NN$ where $h_1, \dots, h_m$ are given in binary, an integer $n \in \mathbb{N}$ given in unary.}
{Is there a $G \in \CG$ such that $\lvert U(G) \rvert \le n$ and $\hom(F_i, G) = h_i$ for every $i \in [m]$?}

\dproblem{\BoundedFGSubR}
{Pairs $(F_1, h_1), \dots, (F_m,h_m) \in \CF \times \NN$ where $h_1, \dots, h_m$ are given in binary, an integer $n \in \mathbb{N}$ given in unary.}
{Is there a $G \in \CG$ such that $\lvert U(G) \rvert \le n$ and $\sub(F_i, G) = h_i$ for every $i \in [m]$?}
\clearpage{}
\clearpage{}\section{Counting Complexity Classes beyond $\SP$}
\label{sec:ct}

In \cite{DBLP:journals/acta/Wagner86}, Wagner studied the complexity of 
combinatorial problems with inputs described by languages that can succinctly 
express sets of 
exponential size in length $n$. A particular description language of interest 
were integer-expressions such as \[ H = (((0 \cup 1)+(0 \cup 2))+(0\cup4)),\]
introduced in \cite{DBLP:conf/stoc/StockmeyerM73}, defining up 
to exponentially large subsets of $\NN$, as well as giving rise to natural 
complete problems in a hierarchy of counting complexity classes contained in 
$\mathsf{PSPACE}$.

Since integer-expressions $H$ can be represented by acyclic 
directed edge-labelled graphs $G(H)$, any natural number contained in the 
described language $L(H)$ can be witnessed by a path in $G(H)$ in polynomial 
time. As to be expected, the membership problem ``$n \in L(H)?$'' is 
$\NP$-complete. Wagner also defines an ${\NP}^{\SP}$-complete problem: given an 
integer-expression $H$ and a natural number $m$, deciding 
whether there exists an element $a \in L(H)$ described by at least $m$ 
different paths in $G(H)$.

While Wagner's result can serve as a self-contained introduction to the 
notion of counting complexity classes, we will instead use the more recent 
$\mathsf{GapP}$-machinery \cite{DBLP:journals/jcss/FennerFK94} to define our 
own version of a
complete problem for ${\NP}^{\CEP}$, provide some context to the complexity 
class $\CEP$, and avoid the precise definition of succinct description 
languages, such as the combinatorial circuit, boolean-, and integer-expression 
languages used by Wagner. For our purposes, we will work with counting 
accepting computation paths of nondeterministic polynomial-time Turing machines 
and apply the well-known Cook--Levin reduction. 

\subsection{$\NP^\CEP$-Completeness of 
\existsEqualsThreeColouring}
\label{sec:completeness}

Let us recall the counting class terminology used by Fenner, Fortnow, and Kurtz 
\cite{DBLP:journals/jcss/FennerFK94}. We let $\Sigma = \{0,1\}$ and 
interpret $\Sigma^*$ as the set of natural numbers encoded in binary. For any 
nondeterministic Turning machine $M$ running in polynomial time, we define the 
function $\#M \colon \Sigma^* \to \Sigma^*$ such that $\#M(x)$ is the 
number of accepting computation paths of $M$ on $x$ for $x \in \Sigma^*$.
The class $\FP$ denotes all polynomial-time computable functions $k \colon \Sigma^* 
\to \Sigma^*$ and $\SP$ all functions $f \colon \Sigma^* \to \Sigma^*$ definable as 
$f(x) = \#M(x)$, for all $x \in \Sigma^*$. The class $\GapP$ can now be 
described as all functions $g \colon \Sigma^* \to \Sigma^*$ that can be defined as 
$g(x) = k(x) - f(x)$, for some $k \in \FP$ and $f \in \SP$.
Now, we are able to give an alternative definition of $\CEP$.

\begin{lemma}[\cite{DBLP:journals/jcss/FennerFK94}]
	$\CEP$ is gap-definable, i.e. it holds that
	$L \in \CEP$ if and only if there exists $ g \in 
	\GapP$ such that $x \in L \Leftrightarrow g(x) = 
	0$.
\end{lemma}

\begin{corollary}\label{coNPinCEP}
	$\mathsf{coNP}\subseteq \CEP$.
\end{corollary}
\begin{proof}
	Any language $L \in \mathsf{coNP}$ can be recognized by a
	nondeterministic Turing machine $M$ running in polynomial time such that,
	for some polynomial $p$, for all $x \in \Sigma^*$, it holds that 
	$\#M(x) = 2^{p(|x|)}$ for all $x \in L$ and
	$\#M(x) < 2^{p(|x|)}$ otherwise. 
	Hence, $g_M(x) := 2^{p(|x|)} - \#M(x) \in \GapP$, and with $x \in L$ if and 
	only if $g_M(x) = 0$ follows $L \in \CEP$.
\end{proof}

For any language $O \subseteq \Sigma^*$, an \emph{oracle machine} $M^O$ is a 
(nondeterministic) Turing machine with access to an additional tape for queries 
of the form ``$q \in O$'', to which answers $[q \in O] \in \{0,1\}$ 
can be provided in a single computation step.
The class $\NP^\CEP$ is defined as all decision problems $L$ for which there 
exists a language $O \in \CEP$ and a nondeterministic polynomial-time oracle 
machine $M^O$ deciding $L$. We will simply refer to $M^O$ as $\NP^\CEP$ machine.
While \Cref{coNPinCEP} implies $\NP^\NP = \NP^\mathsf{coNP} \subseteq 
\NP^{\CEP}$, the stronger statement $\NP^{\NP^{\cdots^\NP}} = \PH \subseteq 
\NP^\CEP$ follows from Toda's Theorem \cite{toda_pp_1991}, already referred to in \Cref{todascor}. 
Vaguely speaking, the additional nondeterministic computation power of 
$\NP^\CEP$ seems to make up for the ``universal-flavoured'' shortcomings of 
$\CEP$ as class of decision problems, when compared to $\mathsf{PP}$, for 
example.

We say an $\NP^\CEP$ machine $M^O$ is \emph{normalised} if its computation 
can be split into two stages: first, a nondeterministic stage without oracle 
queries, and second, a deterministic stage consisting of a single step before 
the computation halts, where a single oracle query to $O$ can be made.

\begin{lemma}
	Every language $L$ in $\NP^\CEP$ can be decided by a normalised $\NP^\CEP$ 
	machine.
\end{lemma}

\begin{proof}
	Assume there is an $\NP^\CEP$ machine $M^O$ running in polynomial time 
	$p(n)$ that decides $L$, for some $O \in \CEP$. Thus, there are some 
	functions $k,f$ in $\FP$ and $\SP$, respectively, such that $x \in O$ holds 
	if and only if $k(x) - f(x) = 0$.
	
	The normalised $\NP^\CEP$ machine $N^{O'}$ then proceeds as follows: on 
	input $x \in \Sigma^*$, for $m \coloneqq  p(\abs{x})$, it begins by 
	nondeterministically guessing query strings 
	$q_1,\dots,q_m \in \Sigma^{<m}$ as well as values $a_1,\dots,a_m \in 
	\Sigma^{<m}$. This concludes the nondeterministic stage.
	Next, $N^{O'}$ begins to simulate $M^O$ on input $x$. Whenever $M^O$ makes 
	the $i$-th oracle query ``$q'_i \in O$'', $N^{O'}$ rejects, if $q'_i \neq 
	q_i$. 
	Otherwise, $N^{O'}$ computes $k(q_i)$, and returns $[k(q_i) - 
	a_i = 0]$ as oracle answer to $M^O$, then continues the 
	simulation of $M^O$.
	Observe that for any $f \in \SP$, the function
	\[ f^+(x) = \begin{cases} 
		f(x_1)1^m0^m\dots1^m0^mf(x_m), &
		\text{if }x = x_11^m0^m\dots1^m0^mx_m,\text{ and } \forall i 
		\in [m]: x_i \in \Sigma^{<m}, \\
		\epsilon, & \text{else}. \\
	\end{cases}
	\]
	is also contained in $\SP$, by closure of $\SP$ under addition and 
	multiplication with natural numbers. Similarly, $\SP$ contains functions on 
	(encoded) tuples that are projections of $\SP$ functions onto 
	their components. Hence, when fixing a usual encoding of tuples as $\langle 
	* \rangle$,
	\[ O' \coloneqq \{\langle x,y\rangle  \in \Sigma^* \mid y - f^+(x) = 0\} \in 
	\CEP.\]
	Finally, let $(\overline{q},\overline{a}) \coloneqq \langle q_11^m0^m\dots 
	1^m0^mq_m, 
	a_11^m0^m\dots 1^m0^ma_m\rangle$. Once the simulation of $M^O$ is complete, 
	$N^{O'}$ queries its own oracle with ``$(\overline{q},\overline{a}) \in 
	O'$'', and 
	accepts $x$ if $M^O$ accepted and the answer is 
	$(\overline{q},\overline{a}) \in O'$, 
	or 
	rejects, otherwise.
	
	Note that $N^{O'}$ still runs in polynomial time, and is indeed a 
	normalised $\NP^\CEP$ machine. The correctness of its simulation is easily 
	verified.
\end{proof}

Using the classical Cook--Levin reduction from problems $L \in \NP^\CEP$ given 
as normalised $\NP^\CEP$ machine $M^O$, for some $O \in \CEP$, it is readily 
seen that the problem

\dproblem{$\existsEqualsThreeSAT$}
{A $3$-CNF formula 
$\phi(\overline{X},\overline{Y})$ and an $n \in 
\NN$ given in binary.}
{Is there an assignment $\alpha(\overline{X})$ such that there are exactly $n$ 
	assignments $\beta({\overline{Y}})$ 
	such that $\phi(\alpha(\overline{X}),\beta(\overline{Y}))$ is satisfied?}

has to be $\NP^\CEP$-complete under polynomial-time 
reductions. We will transform
a nondeterministic polynomial-time Turing machine 
computation into a polynomial-sized boolean formula, 
and employ the variable set
$\overline{Y}$ to encode the oracle query as a
comparison between the number of satisfying 
assignments $\beta(\overline{Y})$ and some $n \in 
\NN$.
Let us remark that 
$\existsEqualsThreeSAT$ is already 
hard for $\NP^\CEP$, even though its 
input consists of a binary number $n \in 
\NN$, instead of an arbitrary 
function $k \in \FP$ to compare the 
number of assignments 
with, as the definition of $\CEP$ 
might suggest. 

\begin{lemma}
	For all $L \in \NP^\CEP$, it holds that $L \le_p \existsEqualsThreeSAT$. 
	Thus, $\existsEqualsThreeSAT$ is $\NP^\CEP$-complete.
\end{lemma}

\begin{proof}
	We show that every language $L \in \NP^\CEP$ given by a normalised 
	$\NP^\CEP$ machine $M^O$, for some $O \in 
	\CEP$ with $k \in\FP$ and 
	$f \in\SP$, can be reduced to $\existsEqualsThreeSAT$. Since $f \in\SP$, we 
	can assume that there exists a polynomial $p$ such that for all $y \in 
	\Sigma^*$, $f(y) \leq 2^{p(|y|)}$.
	
	\begin{claim}
		$L \le_p \existsEqualsThreeSAT$.
	\end{claim}
	\begin{claimproof}
	For now, we will only consider the first stage of $M^O$, a nondeterministic 
	computation of polynomial length. Given $x \in \Sigma^*$, it is well known 
	that in time polynomial in the running time of $M^O$ we can compute a 
	boolean formula $\phi_{x}(\overline{X})$ which is satisfied if and only if 
	its 
	variables encode a computation path of $M^O$ on $x$ ending in the unique 
	state of the second stage, with $M^O$ ready to accept, should the 
	oracle query return positive. Further, this correspondence between 
	satisfying assignments and computation paths is one-to-one. Let $y \in 
	\Sigma^*$ denote the input to the subsequent oracle query, which is at this 
	moment already encoded in the assignment of variables $\overline{X}_y 
	\subseteq \overline{X}$. Let $n_y \coloneqq 2^{p(|y|)}$.
	
	We continue by considering the second stage of $M^O$: for some $y \in 
	\Sigma^*$, a single oracle query ``$y \in O$'', followed by accepting 
	$x$ if and only if $[y \in O] = 1$. 
	By definition, it holds that $y \in O \iff k(y) - f(y) = 0$, and thus
	$y \in O \iff k(y) + (n_y - f(y)) 
	= n_y.$ 
	Since every \emph{nonnegative} function from $\FP-\SP = \GapP$ is contained 
	in $\SP$, and $\SP$ is closed under addition, it holds that $q(y) \coloneqq k(y) + 
	(n_y - f(y)) \in \SP$. Thus, by 
	the $\SP$-completeness of 
	$\SSat$, 
	we can also compute in polynomial time a boolean formula 
	$\psi_{O}(\overline{X}_y,\overline{Y})$ with exactly $q(y)$ satisfying 
	assignments of 
	$\overline{Y}$, by encoding the addition as disjunction, where assignments 
	of each summand are made disjoint by adding a new variable, 
	and $n_y - f(y)$ by excluding assignments satisfying the formula for 
	$f(y)$, using negation.
	Furthermore, we may assume that the computed formula 
	$\phi_{x}(\overline{X})\wedge\psi_{O}(\overline{X}_y,\overline{Y})$ is in 
	$3$-CNF, by employing the usual transformation.
	
	It remains to verify that $x \in L$ holds if and only if 
	$(\phi_{x}(\overline{X})\wedge\psi_{O}(\overline{X}_y,\overline{Y}),n_y)
	 \in \existsEqualsThreeSAT$. Assuming $x \in L$, it follows that $M^O(x) = 
	 1$, 
	and thus there exists an 
	assignment $\alpha(\overline{X})$ such that 
	$\phi_{x}(\alpha(\overline{X}))$ is 
	satisfied, while $\alpha(\overline{X}_y)$ encodes an input $y$ to the 
	oracle query. Since the second stage of $M^O$ depends only on the outcome 
	of query ``$y \in O$'', it follows that $y 
	\in O$, and thus, $q(y) = k(y) + 
	(n_y - f(y)) = n_y$. Hence, the formula 
	$\phi_{x}(\alpha(\overline{X}))\wedge\psi_{O}(\alpha(\overline{X}_y),\overline{Y})
	 \equiv 
	\psi_{O}(\alpha(\overline{X}_y),\overline{Y})$ has 
	exactly $n_y$ satisfying 
	assignments $\beta(\overline{Y})$ of 
	$\overline{Y}$.
	Conversely, from the existence of an assignment $\alpha(\overline{X})$ 
	such that 
	exactly $n_y$ satisfying 
	assignments $\beta(\overline{Y})$ for 
	$\phi_{x}(\alpha(\overline{X}))\wedge\psi_{O}(\alpha(\overline{X}_y),\overline{Y})$
	 exist, 
	follows that $M^O(x) = 1$.
	\end{claimproof}

	Finally, we will quickly demonstrate that $\existsEqualsThreeSAT \in 
	\NP^\CEP$. On input $x = 
	\langle\phi(\overline{X},\overline{Y}),n\rangle$,
	 an $\NP^\CEP$ machine 
	can guess $\alpha(\overline{X})$, and since $\SSat \in \SP$, with $g(x) 
	\coloneqq n - 
	\SSat(\phi(\alpha(\overline{X}),\overline{Y}))$
	 it 
	follows the query 
	``$\SSat(\phi(\alpha(\overline{X}),\overline{Y}))
	 = n$'' is 
	in $\CEP$.
\end{proof}

The $\NP^\CEP$-completeness of 
$\existsEqualsThreeColouring$ is now 
straightforward, as the well-known reduction between $\threeSAT$ and 
$\threecolouring$ is parsimonious. 
We observe that for every formula $\varphi(\overline{X},\overline{Y})$, this 
property is preserved under partial assignments $\alpha(\overline{X})$.

\clearpage{}
\clearpage{}\section{Hardness Proofs for Uncoloured Graphs}
\label{sec:uncoloured}

In this part of the appendix, we provide adapted versions of the reductions
from \Cref{sec:complexity}
that reduce the corresponding problems to homomorphism reconstructability problems
for uncoloured graphs.
We essentially encode the labels and colours used in the reductions of \Cref{sec:complexity}
by using \emph{Kneser graphs}.
We describe Kneser graphs, their properties that we utilise,
and the general ideas of our gadget constructions in \Cref{sec:kneserGraphs}.
We then build upon that in
Appendices~\ref{sec:uncolouredGeneralHardness}, \ref{sec:uncolouredThreeConstraints}, and \ref{sec:uncolouredFiniteSet}
to adapt the reductions of \Cref{th:NPCEPReduction,th:hardnesscolouredConstant,th:hardnesscolouredFixed},
and, hence, prove \Cref{th:uncolouredGeneralHardness,th:ConstantEncoding,th:hardnessFixed},
respectively;
the construction needed for \Cref{th:uncolouredGeneralHardness}
is arguably the most involved of the three, which is why it comes last.

\subsection{Kneser Graphs}
\label{sec:kneserGraphs}

Let us first recall some basic facts on Kneser graphs from
\cite{hahn_homomorphisms_1997}.
For natural numbers $r,s \in \NN$ with $1 \le r < s/2$,
the \textit{Kneser graph} $K(r,s)$
is the graph whose vertices are the $r$-subsets of $[s]$
where two vertices are connected by an edge
if and only if they are disjoint.
The condition $r < s/2$ guarantees that
every such Kneser graph is connected.
Every Kneser graph is a core \cite[Proposition~3{.}13]{hahn_homomorphisms_1997}.
For natural numbers $k,n \in \NN$ with $n \ge 3$,
the Kneser graph $K(k(n-2), (2k+1)(n-2))$ has chromatic number $n$
and odd girth $2k + 1$ \cite[Proposition~3{.}13]{hahn_homomorphisms_1997}.
This fact is often used to obtain a (possibly infinite) family of pairwise
homomorphically incomparable connected graphs since,
if $h \colon V(F) \to V(G)$ is a homomorphism from a graph $F$
to a graph $G$, then the chromatic number of $G$
has to be greater than or equal to the chromatic number of $F$
while the odd girth of $G$ has to be less than or equal to
the odd girth of $F$.

Böker, Chen, Grohe, and Rattan~\cite{boker_complexity_2019}
proved that there is a class of directed graphs $\CF$ for which
homomorphism indistinguishability is undecidable.
They then provided a construction that allows to encode directed graphs
as undirected ones while preserving homomorphism numbers in some sense,
which was made possible by encoding edge directions using
gadgets consisting of Kneser graphs.
This allowed them to turn $\CF$ into a class of undirected graphs
for which homomorphism indistinguishability is still undecidable.
Similarly, Roth and Wellnitz~\cite{roth_counting_2020} recognized the usefulness of Kneser graphs
in the context of homomorphism counts and used them to encode
every problem $P \in \#W[1]$ as two sets of graphs $\CF$ and $\CG$
such that $P$ is equivalent to the problem of counting homomorphisms
from graphs in $\CF$ to graphs in $\CG$.

Before showing how we can employ Kneser graphs
to adapt the reductions from \Cref{sec:complexity} to
homomorphism reconstructability problems for unlabelled and uncoloured graphs,
let us first describe the construction from \cite{boker_complexity_2019}:
fix four pairwise homomorphically incomparable
connected Kneser graphs $K^1, \dots, K^4$.
Let $\ell \coloneqq \max_{i \in [4]} \lvert V(K^i) \rvert$.
For every $i \in [4]$, we fix a vertex of $K^i$,
which we denote by $t_{i}$ and call the \textit{tip} of $K^i$.
For a directed graph $F$, we construct the undirected graph $F^U$
by
replacing every directed edge $(u,v)$ by the \textit{direction gadget}
from \Cref{fig:directionGadget} and,
for every vertex $u$ of $F$,
add a copy of the \textit{indicator gadget} from \Cref{fig:indicatorGadget}
and identify $u$ of $F$ with the vertex $v$ of the indicator gadget.
Note that $P_i$ denotes a path of length $i \ge 0$
in \cref{fig:gadgets}.
These long paths used in the gadgets
prevent the Kneser graphs from being mapped to other parts
of an encoded graph and, in particular,
from being mapped to the same
Kneser graph in a different copy of the gadget.

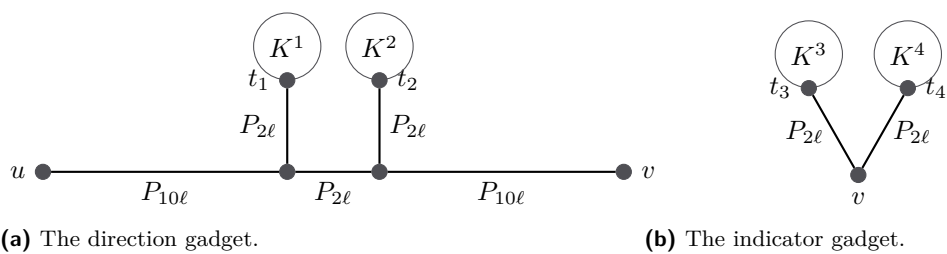
\begin{figure}
	\begin{subfigure}{0.6\textwidth}
		\centering
		\begin{tikzpicture}
			\node[vertex, label={left:$u$}] (u) {};
			\node[vertex, right = 3cm of u] (u2) {};
			\node[vertex, right = 1cm of u2] (v2) {};
			\node[vertex, label={right:$v$}, right = 3cm of v2] (v) {};
			\node[vertex, label={left:$t_1$}, above = 1cm of u2] (u3) {};
			\node[kneser, anchor=south] (uKneser) at (u3) {$K^1$};
			\node[vertex, label={right:$t_2$}, above = 1cm of v2] (v3) {};
			\node[kneser, anchor=south] (vKneser) at (v3) {$K^2$};
			
			\path[thick] (u) edge node[below] {$P_{10\ell}$} (u2)
			(u2) edge node[left] {$P_{2\ell}$} (u3)
			(u2) edge node[below] {$P_{2\ell}$} (v2)
			(v2) edge node[right] {$P_{2\ell}$} (v3)
			(v2) edge node[below] {$P_{10\ell}$} (v);
		\end{tikzpicture}
		\caption{The direction gadget.}
		\label{fig:directionGadget}
	\end{subfigure}
	\begin{subfigure}{0.4\textwidth}
		\centering
		\begin{tikzpicture}
			\node[vertex, label={below:$v$}] (u) {};
			\node[vertex, label={left:$t_{3}$},
			above left = 1.0cm and 0.5cm of u] (u2) {};
			\node[vertex, label={right:$t_{4}$},
			above right = 1.0cm and 0.5cm of u] (u3) {};
			\node[kneser, anchor=south] (uKneser1) at (u2) {$K^3$};
			\node[kneser, anchor=south] (uKneser2) at (u3) {$K^4$};
			
			\path[thick] (u) edge node[left] {$P_{2\ell}$} (u2)
			(u) edge node[right] {$P_{2\ell}$} (u3);
		\end{tikzpicture}
		\caption{The indicator gadget.}
		\label{fig:indicatorGadget}
	\end{subfigure}
	\caption{The gadgets for encoding directed graphs.}
	\label{fig:gadgets}
\end{figure}

Böker, Chen, Grohe, and Rattan have shown that this construction behaves as desired:
if $F$ and $G$ are directed graphs
and $h$ is a homomorphism from $F^U$ to $G^U$,
then the restriction of $h$ to an arbitrary indicator
or direction gadget of $F^U$
is an isomorphism to an indicator or direction gadget of $G^U$, respectively.
In particular,
there is a one-to-one correspondence between homomorphisms
from a directed graph $F$ to a directed graph $G$
and homomorphisms from $F^U$ to $G^U$
modulo automorphisms of the Kneser graphs.

\begin{lemma}[{\cite[Lemma 17]{boker_complexity_2019}}]
    \label{le:encodingPreservesHom}
    Let $a_i \coloneqq \lvert \{\pi \in \Aut(K^i) \mid \pi(t_i) = t_i\}\rvert$
    be the number of automorphisms of $K^i$ that
    stabilise the tip of $K^i$
    for every $i \in [4]$.
    Then,
    for all directed graphs $F$ and $G$,
    \begin{equation*}
        \hom(F^U, G^U) =
        (a_1 \cdot a_2)^{\lvert E(F) \rvert} \cdot
        (a_3 \cdot a_4)^{\lvert V(F) \rvert} \cdot
        \hom(F, G).
    \end{equation*}
\end{lemma}

The reason why this construction works is that
the employed Kneser graphs are connected, cores, and pairwise homomorphically incomparable:
let $K$ be a copy of a Kneser graph occurring in $F^U$.
Since the Kneser graphs are connected, $K$
can only be mapped to a connected subgraph of $G^U$.
Since the Kneser graphs are cores, $K$ cannot be mapped to
the long paths because, otherwise, there would be a homomorphism from $K$ to a bipartite
graphs and hence a single edge, which would contradict $K$ being a core.
Since the Kneser graphs are pairwise homomorphically incomparable, $K$
cannot be mapped to a copy of one of the other Kneser graphs in $G^U$.
Moreover, it is not possible that only a part of $K$ is mapped to a
copy of $K$ or a copy of a different Kneser graph in $G^U$:
the remainder of $K$ would have to be mapped to the long paths,
which would mean that this part could be folded in into a single edge
and we would obtain a homomorphism to a proper subgraph of $K$ in the first case,
contradicting $K$ being a core,
and a homomorphism to one of the different Kneser graphs in the second case,
contradicting the Kneser graphs being pairwise homomorphically incomparable.
Finally, this is also the reason why
a homomorphism from $F^U$ to $G^U$ also has to stabilise the tip
of every copy of $K$: if the tip was not mapped to the tip of a copy of $K$
in $G^U$, then this would mean that $K$ \enquote{pulls} at a long path
and causes that a different copy of a Kneser graph in $F^U$ is partially mapped
to long paths, yielding a contradiction.

We proceed in a similar fashion and provide gadgets built from Kneser graphs that
can be used for encoding labels and colours,
which then allows us to construct reductions to homomorphism
reconstructability problems for (unlabelled and uncoloured) graphs from the reductions
presented in \Cref{sec:complexity}.
Proving the correctness of such a modified reduction, however,
involves more work as only one direction of the correctness proof is
a simple application of a lemma in the vein of \Cref{le:encodingPreservesHom}.
For the \emph{forwards direction},
we have to provide a graph $G$ that satisfies the encoded constraints.
We essentially just choose $G$ as the encoding of the graph used in
the reduction to the labelled/coloured problem and
apply a variant of \Cref{le:encodingPreservesHom} for our specific construction.
For the \emph{backwards direction}, however,
we are given a graph $G$ that satisfies all encoded constraints but
does not have to follow our chosen encoding and use of Kneser graphs.
Hence, we are left to argue that a solution to the problem we are reducing
from can still be salvaged from $G$.
Here, the most important trick we employ is that our reduction can provide additional constraints
that forbid the existence of all non-homomorphic images of the Kneser graphs we use in $G$.
This is usually enough to guarantee that the structure of $G$ is at least somewhat reasonable.

\subsection{Reduction from $\SetSplitting$}
\label{sec:uncolouredThreeConstraints}

\begin{proof}[Proof of \Cref{th:ConstantEncoding}]
	We adapt the reduction of \Cref{th:hardnesscolouredConstant}
	by encoding colours in binary via Kneser graphs.
	Fix four pairwise homomorphically
    incomparable connected Kneser graphs $K^A$, $K^0$, $K^1$, and $K^Z$,
	and fix a vertex, the \textit{tip}, for all of these graphs.
	Let $\ell \coloneqq \max \{\lvert V(K^A) \rvert, \lvert V(K^0) \rvert, 
	\lvert V(K^1) \rvert, \lvert V(K^Z) \rvert\}$ and
	let $k_i$ be the number of automorphisms of $K^i$
	that stabilise the tip for $i \in \{A, 0, 1, Z\}$.

	\begin{figure}
		\centering
		\begin{tikzpicture}
			\node[vertex, label={left:$t$}] (t) {};

			\node[vertex, label={below:$a$}, right = 1.5cm of t] (s) {};
			\node[kneser, anchor=south] (sKneser) at (s) {$K^A$};

			\node[vertex, label={below:$b_1$}, right = 1.5cm of s] (b1) {};
			\node[kneser, anchor=south] (b1Kneser) at (b1) {$K^{0/1}$};

			\node[vertex, label={below:$b_2$}, right = 1.5cm of b1] (b2) {};
			\node[kneser, anchor=south] (b2Kneser) at (b2) {$K^{0/1}$};

			\node[right = 1cm of b2] (dots) {\scriptsize$\ldots$};
			\node[vertex, label={below:$b_m$}, right = 1cm of dots] (bm) {};
			\node[kneser, anchor=south] (bmKneser) at (bm) {$K^{0/1}$};

			\node[vertex, label={below:$z$}, right = 1.5cm of bm] (e) {};
			\node[kneser, anchor=south] (eKneser) at (e) {$K^Z$};

			\path[thick]
			(t) edge node[above] {$P_{2\ell}$} (s)
			(s) edge node[above] {$P_{2\ell}$} (b1)
			(b1) edge node[above] {$P_{2\ell}$} (b2)
			(b2) edge node[above, xshift=5pt] {$P_{2\ell}$} (dots)
			(dots) edge node[above, xshift=-5pt] {$P_{2\ell}$} (bm)
			(bm) edge node[above] {$P_{2\ell}$} (e);
		\end{tikzpicture}
		\caption{The gadgets for the binary encoding of colours.}
		\label{fig:binaryGadget}
	\end{figure}
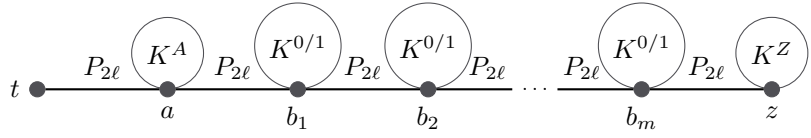

	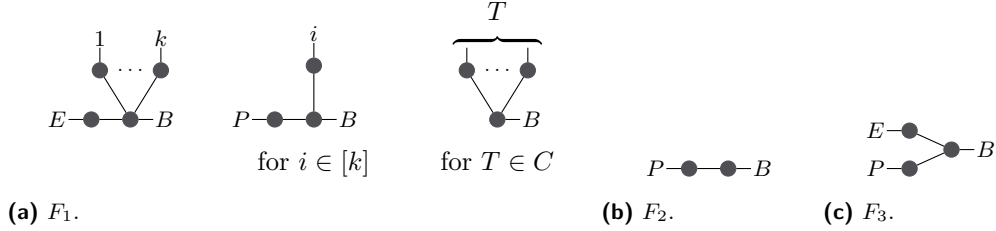
\begin{figure}
		\centering
		\begin{subfigure}[b]{0.55\textwidth}
			\centering
			\begin{tikzpicture}[node distance = 0.5cm]
				\node[vertex] (everyB) {};
				\node[vertex, left = 0.3cm of everyB] (everyE) {};
				\node[vertex, above left = 0.5cm and 0.25cm of everyB] (every1) {};
				\node[vertex, above right = 0.5cm and 0.25cm of everyB] (everyk) {};
				\node[] (everydots) at ($(every1)!0.5!(everyk)$) {\scriptsize$\ldots$};
				\draw
				(everyB) edge (everyE)
				edge (every1)
				edge (everyk)
				;
\node[below = 0.3cm of everyB] {};

				\path
				(every1) edge node[above] {\footnotesize$1$} ($(every1) + (0,0.3)$)
				(everyk) edge node[above] {\footnotesize$k$} ($(everyk) + (0,0.3)$)
				(everyB) edge node[right] {\footnotesize$B$} ($(everyB) + (0.3,0)$)
				(everyE) edge node[left] {\footnotesize$E$} ($(everyE) + (-0.3,0)$)
				;

				\node[vertex, right = 2.2cm of everyB] (partB) {};
				\node[vertex, left = 0.3cm of partB] (partP) {};
				\node[vertex, above = of partB] (parti) {};
				\node[below = 0.17cm of partB] (partDescription) {for $i \in [k]$};
				\draw
				(partB) edge (partP)
				edge (parti)
				;
				\path
				(parti) edge node[above] {\footnotesize$i$} ($(parti) + (0,0.3)$)
				(partB) edge node[right] {\footnotesize$B$} ($(partB) + (0.3,0)$)
				(partP) edge node[left] {\footnotesize$P$} ($(partP) + (-0.3,0)$)
				;

				\node[vertex, right = 2.2cm of partB] (setB) {};
				\node[vertex, above left = 0.5cm and 0.25cm of setB] (set1) {};
				\node[vertex, above right = 0.5cm and 0.25cm of setB] (setk) {};
				\node[] (setdots) at ($(set1)!0.5!(setk)$) {\scriptsize$\ldots$};
				\draw[very thick, decorate, decoration = {calligraphic brace}]
				([xshift = -2pt, yshift = 10pt]set1.west)
				-- node[above, yshift = 5pt] {$T$}
				([xshift = 2pt, yshift = 10pt]setk.east)
				;
				\node[below = 0.2cm of setB] (setDescription) {for $T \in C$};
				\draw
				(setB) edge (set1)
				edge (setk)
				;
				\path
				(setB) edge node[right] {\footnotesize$B$} ($(setB) + (0.3,0)$)
				(set1) edge node[right] {} ($(set1) + (0,0.3)$)
				(setk) edge node[right] {} ($(setk) + (0,0.3)$)
				;
			\end{tikzpicture}
			\caption{$F_1$.}
		\end{subfigure}
		\begin{subfigure}[b]{0.20\textwidth}
			\centering
			\begin{tikzpicture}[node distance = 0.5cm]
				\node[vertex, ] (B) {};
				\node[vertex, left = 0.3cm of B] (P) {};
\draw
				(B) edge (P)
				;
				\path
				(B) edge node[right] {\footnotesize$B$} ($(B) + (0.3,0)$)
				(P) edge node[left] {\footnotesize$P$} ($(P) + (-0.3,0)$)
				;
			\end{tikzpicture}
			\caption{$F_2$.}
		\end{subfigure}
		\begin{subfigure}[b]{0.20\textwidth}
			\centering
			\begin{tikzpicture}[node distance = 0.5cm]
				\node[vertex, ] (B) {};
				\node[vertex, above left = 0.1cm and 0.4cm of B] (E) {};
				\node[vertex, below left = 0.1cm and 0.4cm of B] (P) {};
\draw
				(B) edge (P)
				(B) edge (E)
				;
				\path
				(E) edge node[left] {\footnotesize$E$} ($(E) - (0.3,0)$)
				(B) edge node[right] {\footnotesize$B$} ($(B) + (0.3,0)$)
				(P) edge node[left] {\footnotesize$P$} ($(P) + (-0.3,0)$)
				;
			\end{tikzpicture}
			\caption{$F_3$.}
		\end{subfigure}
		\caption{Some of the constraint graphs used in
			the reduction of \Cref{th:ConstantEncoding}.
            A half-present edge
            indicates the use of a gadget from \Cref{fig:binaryGadget}.}
		\label{fig:hardnessEncodedConstantConstruction}
	\end{figure}

	Given a collection $\CC$ of subsets of a finite set $S$, we again
	assume that $S = [k]$ by re-labelling the elements of $S$.
    Furthermore, we assume that $\CC$ does not contain the empty set;
    otherwise, the instance would be trivial.
	In the original reduction, we constructed graphs $F_1$, $F_2$, and $F_3$
    that used colours $1, \dots, k$ and also $B$ (\enquote{black}),
	$E$ (\enquote{everything}),
	and $P$ (\enquote{partition}).
	Now, we encode each of these $k + 3$ colours by a gadget
	that corresponds to a binary number of length $m \coloneqq \lfloor \log 
	(k + 3) \rfloor + 1$,
    which is shown in \Cref{fig:binaryGadget}.
	Then, for a vertex $v$ of colour $C$,
	we take the gadget for colour $C$ and identify $v$
	with the vertex $t$ of the gadget.
	For a colour $C$, let $a_C$ be the product of the tip-stabilising automorphism
	numbers $k_i$ of all Kneser graphs used in the gadget for $C$.
	The modified graphs $F_1$, $F_2$, and $F_3$ that we use for this
    reduction are depicted in \Cref{fig:hardnessEncodedConstantConstruction}.
	The reduction then produces the following constraints:
	\begin{alphaenumerate}
        \item $\hom(F_1) = (a_1 \cdot \ldots \cdot a_k) \cdot a_B \cdot a_E
		\cdot \prod_{i \in [k]} (a_i \cdot a_B \cdot a_P)
		\cdot \prod_{T \in \CC} (a_B \cdot \prod_{i \in T} a_i)$,
        \item $\hom(F_2) = 2 \cdot a_B \cdot a_P$,
        \item $\hom(F_3) = 0$, and
        \item $\hom(F) = 0$ for a graph $F$ of every isomorphism type in the set
                $\CI \coloneqq \{h(K^i) \mid
                \text{$h$ is a non-injective homomorphism from $K^A$, $K^0$, $K^1$, or $K^Z$ to another graph}\}.$ \label{th:ConstantEncoding:constraintNonInj}
	\end{alphaenumerate}

	The constraints from (\ref{th:ConstantEncoding:constraintNonInj})
    are used in the backward direction of the correctness and
    enforce that a homomorphism
	from one of our Kneser graphs to the given graph $G$ is always injective.
	Note that the number of these graphs
	only depends on the four Kneser graphs;
    it does not depend on the given instance $\CC$ for $\SetSplitting$.
    Additionally,
	for the bounded reconstructability problem,
    we can set the size bound
	to an appropriate number linear in $k$;
    this will become obvious from the following correctness proof.

	Given a partition of $S$ into sets $S_1$ and $S_2$
	such that no subset in $\CC$ is entirely contained in either
	$S_1$ or $S_2$, the graph $G_{S_1, S_2}$ in \Cref{fig:hardnessEncodedConstantGraph}
    constructed from $S_1$ and $S_2$
	satisfies all constraints.
    The reasoning for the first three constraints
    is analogous to \Cref{le:encodingPreservesHom}:
	Every homomorphism from $F_1$, $F_2$, or $F_3$ to $G_{S_1, S_2}$ yields
	more homomorphisms to $G_{S_1, S_2}$ by utilising tip-stabilising automorphisms
    of the Kneser graphs, where the precise number is given by
    the product of the automorphism products $a_C$ of the gadgets used
    in the left-hand-side graph.
	Moreover, these are already all homomorphisms:
	Since the Kneser graphs
    are pairwise homomorphically incomparable
	and their own cores, the long paths in our gadgets
    guarantee that every Kneser graph in a gadget
	has to be injectively mapped to a Kneser graph of its own isomorphism type.
	Moreover, these homomorphisms have to stabilise the tips
	of the Kneser graphs since
    a non-tip-stabilising homomorphism would contradict
	the fact that Kneser graphs are cores by the long paths in the gadgets.
    Then, the use of $K^A$ and $K^Z$ to mark the beginning and the end
    of a binary number and the long paths in the gadgets
    then further imply that a gadget has to be injectively
    mapped to a gadget encoding the same number.
    In particular, $K^A$ and $K^Z$ from the same gadget
    are mapped to copies of $K^A$ and $K^Z$ in the same gadget since,
    otherwise, either $K^0$ or $K^1$ would have to be mapped to the long paths.
    Moreover, the same argument yields that
    the tip of the gadget on the left-hand side
    has to be mapped to the tip of the gadget on the right-hand side
    since we have not identified the tip
    of a gadget with an isolated vertex.
	This also means that original vertices (identified with
	the tip of a gadget) have to be mapped to original
	vertices.
    Finally, the constraints from (\ref{th:ConstantEncoding:constraintNonInj})
    are satisfied since, otherwise, our use of long paths
    means that we would get a homomorphism
    from one of our Kneser graphs to a proper subgraph of itself
    or to one of the other Kneser graphs;
    this is a contradiction in both cases.

	\begin{figure}
		\centering
		\begin{tikzpicture}[node distance = 0.5cm]
			\node[vertex] (everyB) {};
			\node[vertex, left = 0.3cm of everyB] (everyE) {};
			\node[vertex, above left = 0.5cm and 0.25cm of everyB] (every1) {};
			\node[vertex, above right = 0.5cm and 0.25cm of everyB] (everyk) {};
			\node[] (everydots) at ($(every1)!0.5!(everyk)$) {\scriptsize$\ldots$};
			\draw
			(everyB) edge (everyE)
			edge (every1)
			edge (everyk)
			;
			\path
			(every1) edge node[above] {\footnotesize$1$} ($(every1) + (0,0.3)$)
			(everyk) edge node[above] {\footnotesize$k$} ($(everyk) + (0,0.3)$)
			(everyB) edge node[right] {\footnotesize$B$} ($(everyB) + (0.3,0)$)
			(everyE) edge node[left] {\footnotesize$E$} ($(everyE) + (-0.3,0)$)
			;

			\node[vertex, right = 2.5cm of everyB] (partOneB) {};
			\node[vertex, left = 0.3cm of partOneB] (partOneP) {};
			\node[vertex, above left = 0.5cm and 0.25cm of partOneB] (partOne1) {};
			\node[vertex, above right = 0.5cm and 0.25cm of partOneB] (partOnek) {};
			\node[] (partOnedots) at ($(partOne1)!0.5!(partOnek)$) {\scriptsize$\ldots$};
			\draw[very thick, decorate, decoration = {calligraphic brace}]
			([xshift = -2pt, yshift = 10pt]partOne1.west)
			-- node[above, yshift = 3pt] {$S_1$}
			([xshift = 2pt, yshift = 10pt]partOnek.east)
			;
			\draw
			(partOneB) edge (partOneP)
			edge (partOne1)
			edge (partOnek)
			;
			\path
			(partOneB) edge node[right] {\footnotesize$B$} ($(partOneB) + (0.3,0)$)
			(partOneP) edge node[left] {\footnotesize$P$} ($(partOneP) + (-0.3,0)$)
			(partOne1) edge node[right] {} ($(partOne1) + (0,0.3)$)
			(partOnek) edge node[right] {} ($(partOnek) + (0,0.3)$)
			;

			\node[vertex, right = 2.5cm of partOneB] (partTwoB) {};
			\node[vertex, left = 0.3cm of partTwoB] (partTwoP) {};
			\node[vertex, above left = 0.5cm and 0.25cm of partTwoB] (partTwo1) {};
			\node[vertex, above right = 0.5cm and 0.25cm of partTwoB] (partTwok) {};
			\node[] (partTwodots) at ($(partTwo1)!0.5!(partTwok)$) {\scriptsize$\ldots$};
			\draw[very thick, decorate, decoration = {calligraphic brace}]
			([xshift = -2pt, yshift = 10pt]partTwo1.west)
			-- node[above, yshift = 3pt] {$S_2$}
			([xshift = 2pt, yshift = 10pt]partTwok.east)
			;
			\draw
			(partTwoB) edge (partTwoP)
			edge (partTwo1)
			edge (partTwok)
			;
			\path
			(partTwoB) edge node[right] {\footnotesize$B$} ($(partTwoB) + (0.3,0)$)
			(partTwoP) edge node[left] {\footnotesize$P$} ($(partTwoP) + (-0.3,0)$)
			(partTwo1) edge node[right] {} ($(partTwo1) + (0,0.3)$)
			(partTwok) edge node[right] {} ($(partTwok) + (0,0.3)$)
			;
		\end{tikzpicture}
		\caption{The graph constructed from $S_1$ and $S_2$ in
			the correctness of \Cref{th:ConstantEncoding}.
            A half-present edge
            indicates the use of a gadget from \Cref{fig:binaryGadget}.}
		\label{fig:hardnessEncodedConstantGraph}
	\end{figure}
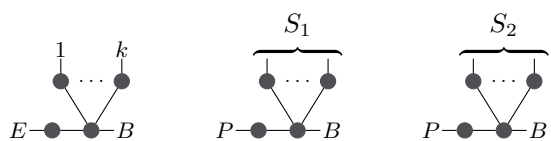

	Conversely, let $G$ be a graph that satisfies all constraints.
	By the constraints of (\ref{th:ConstantEncoding:constraintNonInj}),
    a homomorphism from the other constraint graphs to $G$
	has to be injective on the copies of 
	the Kneser graphs $K^A$, $K^0$,
	$K^1$, and $K^Z$ which means that
	we obtain a distinct homomorphism
	for every non-trivial automorphism of these Kneser graphs.
	For $F_1$ this means, there is exactly one homomorphism from $F_1$
	to $G$ up to tip-stabilising automorphisms of the Kneser graphs.
	In particular, for every component,
	there is exactly one vertex in $G$ to which
	the $B$-gadget tip can be mapped to.
	By the constraint for $F_2$ and a similar argument,
	there are at most $2$ vertices $v_1$ and $v_2$ in $G$ that
	a $B$-gadget tip from $F_2$, i.e.\ the $B$-$P$-graph, can be mapped to.
    Hence, the $B$-gadget tip in an $i$-$B$-$P$ component of $F_1$
    is mapped to $v_1$ or $v_2$ but, by the constraint for $F_1$, not to both.
	Let $S_1, S_2 \subseteq [k]$ be the sets of colours $i$
	for which the $B$-gadget tip of the $i$-$B$-$P$-component
	are mapped to $v_1$ and $v_2$, respectively,
	where we let $S_2 = \varnothing$
	if we only have the one vertex $v_1$;
	then, $S_1$ and $S_2$ form a partition of $S$.
	By the constraint for $F_3$, the $B$-gadget tip
	of the $1$-$k$-$B$-$E$-component of $F_1$
	has to be mapped to a vertex $v_0$ that is distinct from $v_1$ and $v_2$, and moreover,
	every $B$-gadget tip of a $T$-$B$-component is mapped to $v_0$
	and cannot be mapped to $v_1$ or $v_2$
	as this would yield a larger number of homomorphisms
    than the constraint for $F_1$ allows.
	This means that no set $T \in \CC$ is a subset of $S_1$ or $S_2$:
    otherwise, we would obtain a homomorphism
    from the $T$-$B$-component for $T$
    to $G$ that maps the $B$-gadget tip to $v_1$ or $v_2$.
\end{proof}

\subsection{Reduction from $\QuadraticPolynomial$}
\label{sec:uncolouredFiniteSet}

\newcommand{\colourfont}[1]{{#1}}

\newcommand{\pictoColouredPathOf}[3]{\hspace{-3pt}
    \tikz[baseline=-3pt]{
        \node[smallvertex, label={[label distance = -2pt]above:\footnotesize$#2$}] (R) {};\node[smallvertex, label={[label distance = -2pt]left:\footnotesize$#1$}, left = 0.1cm of R] (A) {};\node[smallvertex, label={[label distance = -2pt]right:\footnotesize$#3$}, right = 0.1cm of R] (M) {};\draw (A) edge (R);
        \draw (A) edge (M);}
    \hspace{-2pt}
}

\begin{proof}[Proof of \Cref{th:hardnessFixed}]
	We adapt the reduction of \Cref{th:hardnesscolouredFixed}
	and encode the colours used there by Kneser graphs.
    Since we only used a fixed number of colours in the original reduction,
    we simply take a Kneser graph for every colour:
	fix a set of pairwise homomorphically incomparable connected Kneser graphs
	$K^C$ for $C \in \{\colourfont{R}, \colourfont{M}, \colourfont{M_1}, \colourfont{M_2}, \colourfont{A}, \colourfont{X}, \colourfont{B}, \colourfont{Y}\} \eqqcolon \CC$.
    Here, we include the additional colour $M$ and add a leaf of colour $M$
    to every $R$-coloured vertex of the original reduction
    to eliminate isolated vertices,
    which simplifies the construction.
	For $C \in \CC$, fix an arbitrary vertex of $K^C$, the \emph{tip},
	and let $a_C$ be the number of automorphisms of $K^C$ that stabilise its tip.
	Let $\ell \coloneqq \max_{C \in \CC} \lvert V(K^C) \rvert$.

	\begin{figure}
		\centering
		\begin{subfigure}[b]{0.45\textwidth}
			\centering
			\begin{tikzpicture}[node distance = 0.5cm]
				\node[vertex, label={above:\small$\colourfont{R}$}] (R) {};
				\node[vertex, label={below:\small$\colourfont{A}$}, below left = 0.5cm and 0.25cm of R] (A) {};
				\node[vertex, label={below:\small$\colourfont{X}$}, below right = 0.5cm and -0.15cm of R] (X1) {};
				\node[vertex, label={below:\small$\colourfont{X}$}, below right = 0.5cm and 0.25cm of R] (X2) {};
				\node[vertex, label={above:\small$\colourfont{B}$}, above left = 0.5cm and 0.25cm of R] (B) {};
				\node[vertex, label={above:\small$\colourfont{Y}$}, above right = 0.5cm and 0.25cm of R] (Y) {};
				\node[vertex, label={right:\small$\colourfont{M}$}, right = 0.3cm of R] (C) {};
				\draw
				(R) edge (A)
				edge (X1)
				edge (X2)
				edge (B)
				edge (Y)
				edge (C)
				;
			\end{tikzpicture}
			\caption{$F_\colourfont{poly}$.}
		\end{subfigure}
		\begin{subfigure}[b]{0.45\textwidth}
			\centering
			\begin{tikzpicture}
				\node[vertex] (R) {};
				\node[vertex, below left = 1.0cm and 0.9cm of R] (A) {};
				\node[vertex, below right = 1.0cm and -0.15cm of R] (X1) {};
				\node[vertex, below right = 1.0cm and 0.9cm of R] (X2) {};
				\node[vertex, above left = 1.0cm and 0.9cm of R] (B) {};
				\node[vertex, above right = 1.0cm and 0.9cm of R] (Y) {};
				\node[vertex, right = 0.7cm of R] (C) {};

				\node[kneser, anchor=south] (RKneser) at (R)   {\small$K^{\colourfont{R}}$};
				\node[kneser, anchor=south] (BKneser) at (B)   {\small$K^{\colourfont{B}}$};
				\node[kneser, anchor=south] (YKneser) at (Y)   {\small$K^{\colourfont{Y}}$};
				\node[kneser, anchor=north] (AKneser) at (A)   {\small$K^{\colourfont{A}}$};
				\node[kneser, anchor=north] (X1Kneser) at (X1) {\small$K^{\colourfont{X}}$};
				\node[kneser, anchor=north] (X2Kneser) at (X2) {\small$K^{\colourfont{X}}$};
				\node[kneser, anchor=west] (CKneser) at (C)    {\small$K^{\colourfont{M}}$};
				\draw
				(R) edge node[left] {\small$P_{2\ell}$} (A)
				edge node[left, pos = 0.75, xshift = 3pt] {\small$P_{2\ell}$} (X1)
				edge node[right] {\small$P_{2\ell}$} (X2)
				edge node[left]  {\small$P_{2\ell}$} (B)
				edge node[right] {\small$P_{2\ell}$} (Y)
				edge node[below] {\small$P_{2\ell}$} (C)
				;
			\end{tikzpicture}
			\caption{$F'_\colourfont{poly}$.}
		\end{subfigure}
		\caption{Example of an encoded graph in
			the reduction of \Cref{th:hardnessFixed}.}
		\label{fig:hardnessFixedConstruction}
	\end{figure}
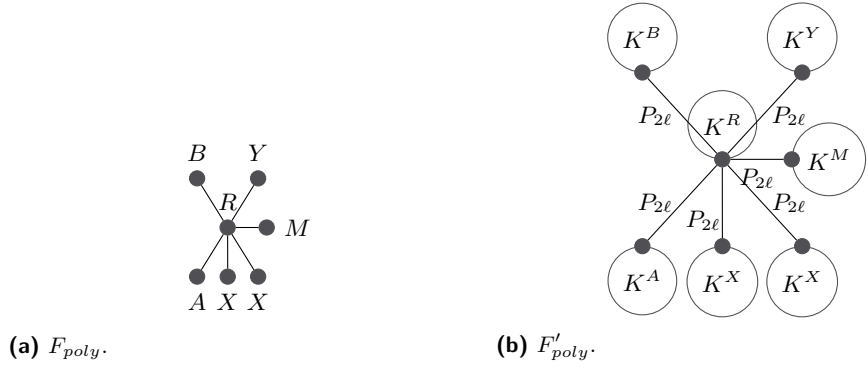

    Let $F$ be a coloured graph
    with colouring function $c \colon V(F) \to \CC$.
    We encode a constraint $\hom(F) = n$ by
    $[\hom(F) = n] \coloneqq \hom(F') = n'$, where
    $F'$ is obtained from $F$ by first introducing a copy of $K^{c(v)}$
    for every vertex $v \in V(F)$ and identifying the tip of $K^{c(v)}$
    with $v$
    and then replacing every original edge of $F$
    by a path $P_{2\ell}$ and
    $n' \coloneqq  n \cdot \prod_{v \in V(F)} a_{c(v)}$.
    See
    \Cref{fig:hardnessFixedConstruction} for an example.

	Given an instance $(a,b,c)$ of \QuadraticPolynomial,
    the reduction then produces the following constraints,
	where (\ref{enum:constraintARC})--(\ref{enum:constraintM2ARXC})
    essentially encode the constraints from the reduction of
	\Cref{th:hardnesscolouredFixed}
    and (\ref{enum:constraintNonInjective}) uses the set
    $\CI \coloneqq \{h(K^C) \mid
        C \in \CC,\, \text{$h$ is a non-injective homomorphism from $K^C$ to another graph}\}$.
    \begin{multicols}{2}
        \begin{alphaenumerate}
            \item $[\hom(\pictoColouredPathOf{A}{R}{M}) = a + 1]$,\label{enum:constraintARC}
            \item $[\hom(\pictoColouredPathOf{B}{R}{M}) = b + 1]$,
            \item $[\hom(F_\colourfont{poly}) = c]$.\label{enum:constraintPoly}
            \item $[\hom(\pictoColouredEdgeOf{R}{M}) = 2]$,\label{enum:constraintRC}
            \item $[\hom(\pictoColouredPathOf{M_1}{R}{M}) = 1]$,\label{enum:constraintM1RC}
            \item $[\hom(\pictoColouredPathOf{M_2}{R}{M}) = 1]$,\label{enum:constraintM2RC}
            \item $[\hom(\tikz[baseline=-3pt]{
                \node[smallvertex, label={[label distance = -2pt]above:\footnotesize$\colourfont{R}$}] (R) {};\node[smallvertex, label={[label distance = -2pt]left:\footnotesize$\colourfont{M_1}$}, above left = 0.04cm and 0.12cm of R] (M) {};\node[smallvertex, label={[label distance = -2pt]left:\footnotesize$\colourfont{M_2}$}, below left = 0.04cm and 0.12cm of R] (MT) {};\node[smallvertex, label={[label distance = -2pt]right:\footnotesize$\colourfont{M}$}, right = 0.1cm of R] (C) {};
                \draw (R) edge (M)
                edge (MT)
                edge (C);
            }) = 0]$,\label{enum:constraintM1M2RC}
            \item $[\hom(\tikz[baseline=-3pt]{
                \node[smallvertex, label={[label distance = -2pt]below:\footnotesize$\colourfont{R}$}] (R) {};\node[smallvertex, label={[label distance = -2pt]left:\footnotesize$\colourfont{M_1}$}, left = 0.1cm of R] (M) {};\node[smallvertex, label={[yshift = 3pt, label distance = -2pt]left:\footnotesize$\colourfont{B}$}, above left = 0.12cm and 0.04cm of R] (B) {};\node[smallvertex, label={[yshift = 3pt, label distance = -2pt]right:\footnotesize$\colourfont{Y}$}, above right = 0.12cm and 0.04cm of R] (Y) {};\node[smallvertex, label={[label distance = -2pt]right:\footnotesize$\colourfont{M}$}, right = 0.1cm of R] (C) {};
                \draw (R) edge (M)
                edge (B)
                edge (Y)
                edge (C);
            }) = 1]$,\label{enum:constraintM1BRYC}
            \item $[\hom(\tikz[baseline=-3pt]{
                \node[smallvertex, label={[label distance = -2pt]below:\footnotesize$\colourfont{R}$}] (R) {};\node[smallvertex, label={[label distance = -2pt]left:\footnotesize$\colourfont{M_2}$}, left = 0.1cm of R] (M) {};\node[smallvertex, label={[yshift = 3pt, label distance = -2pt]left:\footnotesize$\colourfont{A}$}, above left = 0.12cm and 0.04cm of R] (B) {};\node[smallvertex, label={[yshift = 3pt, label distance = -2pt]right:\footnotesize$\colourfont{X}$}, above right = 0.12cm and 0.04cm of R] (Y) {};\node[smallvertex, label={[label distance = -2pt]right:\footnotesize$\colourfont{M}$}, right = 0.1cm of R] (C) {};
                \draw (R) edge (M)
                edge (B)
                edge (Y)
                edge (C);
            }) = 1]$.\label{enum:constraintM2ARXC}
\item $\hom(F) = 0$ for a graph $F$ of every isomorphism type in $\CI$.\label{enum:constraintNonInjective}
        \end{alphaenumerate}
    \end{multicols}
    Note that every graph in $\CI$ is connected since our Kneser graphs are,
    is not bipartite as Kneser graphs are cores,
    and has at most $\ell$ vertices by the definition of $\ell$.
    Moreover, note that
    (\ref{enum:constraintNonInjective})
    only results in a finite number of constraints
    since we only use finitely many Kneser graphs.
	Finally, let $\CF$ be the set of graphs that occur in these constraints.
    Then, $\CF$ is finite and independent of the instance $(a,b,c)$.

	Let $(a,b,c)$ be an instance of $\QuadraticPolynomial$
    with $x,y \in \NN$ such that $ax^2 + by = c$.
    Let $G_{a,b,x,y}$ be the coloured graph from \Cref{fig:hardnessFixedGraph}
    and denote its colouring function by $c$.
	Let $G'_{a,b,x,y}$ be
	the graph obtained from $G_{a,b,x,y}$
	by the same encoding as before, i.e.\ first introducing a copy of $K^{c(v)}$
	for every vertex $v \in V(G_{a,b,x,y})$ and identifying the tip of $K^{c(v)}$
	with $v$
	and then replacing every original edge of $G_{a,b,x,y}$
	by a path $P_{2\ell}$.
	Let us verify that $G'_{a,b,x,y}$ satisfies all constraints.
	First, we observe that
	$G_{a,b,x,y}$ satisfies every unencoded constraint $\hom(F) = n$
	from (\ref{enum:constraintARC})--(\ref{enum:constraintM2ARXC}).
    Let $\hom(F') = n' \coloneqq [\hom(F) = n]$.
	Every homomorphism from $F$ to $G_{a,b,x,y}$ yields
	$\prod_{v \in V(G_{a,b,x,y})} a_{c(v)}$ homomorphisms from
	$F'$ to $G'_{a,b,x,y}$ by utilising tip-stabilising automorphisms
    of the Kneser graphs.
	Moreover, these are already all homomorphisms:
	Since our Kneser graphs are pairwise homomorphically incomparable
	and their own cores, every Kneser graph in $F'$
	has to be injectively mapped to a Kneser graph of its own isomorphism type.
	Moreover, these homomorphisms have to stabilise the tips
	of the Kneser graphs since $F$ is not an isolated vertex,
	which means that a non-tip-stabilising homomorphism would contradict
	the fact that Kneser graphs are cores by our use of long paths.
	This also means that original vertices (identified with
	the tip of a Kneser graph) have to be mapped to original
	vertices. Moreover, since there are no adjacent vertices
	in $F$ that have the same colour, adjacent original vertices
	cannot be mapped to the same original vertex and
	are hence mapped to adjacent
	original vertices.
    Finally, the constraints from (\ref{enum:constraintNonInjective})
    are satisfied since, otherwise, our use of long paths
    means that we would get a homomorphism
    from one of our Kneser graphs to a proper subgraph of itself
    or to one of the other Kneser graphs;
    this is a contradiction in both cases.

	\begin{figure}
		\centering
		\begin{tikzpicture}
			\node[vertex, label={above:$\colourfont{R}$}] (leftR) {};
			\node[vertex, label={below:$\colourfont{X}$}, below right = 0.5cm and 0.1cm of leftR] (leftX1) {};
			\node[vertex, label={below:$\colourfont{X}$}, below right = 0.5cm and 0.7cm of leftR] (leftXx) {};
			\node[vertex, label={below:$\colourfont{A}$}, below left = 0.5cm and 0.7cm of leftR] (leftA1) {};
			\node[vertex, label={below:$\colourfont{A}$}, below left = 0.5cm and 0.1cm of leftR] (leftAa) {};
			\node[vertex, label={above:$\colourfont{Y}$}, above right = 0.5cm and 0.25cm of leftR] (leftY) {};
			\node[vertex, label={above:$\colourfont{B}$}, above left= 0.5cm and 0.25cm of leftR] (leftB) {};
			\node[vertex, label={right:$\colourfont{M}$}, right = 0.3cm of leftR] (leftC) {};
			\node[vertex, label={left:$\colourfont{M_1}$}, left = 0.3cm of leftR] (leftM) {};

			\node[] (Xdots) at ($(leftX1)!0.5!(leftXx)$) {\scriptsize$\ldots$};
			\draw[very thick, decorate, decoration = {calligraphic brace}]
			([xshift = 2pt, yshift = -16pt]leftXx.east)
			-- node[below, yshift = -3pt] {$x$}
			([xshift = -2pt, yshift = -16pt]leftX1.west)
			;

			\node[] (Adots) at ($(leftA1)!0.5!(leftAa)$) {\scriptsize$\ldots$};
			\draw[very thick, decorate, decoration = {calligraphic brace}]
			([xshift = 2pt, yshift = -16pt]leftAa.east)
			-- node[below, yshift = -3pt] {$a$}
			([xshift = -2pt, yshift = -16pt]leftA1.west)
			;

			\draw
			(leftR) edge (leftY)
			edge (leftB)
			edge (leftM)
			edge (leftX1)
			edge (leftXx)
			edge (leftA1)
			edge (leftAa)
			edge (leftC)
			;

			\node[vertex, label={above:$\colourfont{R}$}, right = 4cm of leftR] (rightR) {};
			\node[vertex, label={below:$\colourfont{Y}$}, below right = 0.5cm and 0.1cm of rightR] (rightX1) {};
			\node[vertex, label={below:$\colourfont{Y}$}, below right = 0.5cm and 0.7cm of rightR] (rightXx) {};
			\node[vertex, label={below:$\colourfont{B}$}, below left = 0.5cm and 0.7cm of rightR] (rightA1) {};
			\node[vertex, label={below:$\colourfont{B}$}, below left = 0.5cm and 0.1cm of rightR] (rightAa) {};
			\node[vertex, label={above:$\colourfont{X}$}, above right = 0.5cm and 0.25cm of rightR] (rightY) {};
			\node[vertex, label={above:$\colourfont{A}$}, above left = 0.5cm and 0.25cm of rightR] (rightB) {};
			\node[vertex, label={right:$\colourfont{M}$}, right = 0.3cm of rightR] (rightC) {};
			\node[vertex, label={left:$\colourfont{M_2}$}, left = 0.3cm of rightR] (rightM) {};

			\node[] (Xdots) at ($(rightX1)!0.5!(rightXx)$) {\scriptsize$\ldots$};
			\draw[very thick, decorate, decoration = {calligraphic brace}]
			([xshift = 2pt, yshift = -16pt]rightXx.east)
			-- node[below, yshift = -3pt] {$y$}
			([xshift = -2pt, yshift = -16pt]rightX1.west)
			;

			\node[] (rightAdots) at ($(rightA1)!0.5!(rightAa)$) {\scriptsize$\ldots$};
			\draw[very thick, decorate, decoration = {calligraphic brace}]
			([xshift = 2pt, yshift = -16pt]rightAa.east)
			-- node[below, yshift = -3pt] {$b$}
			([xshift = -2pt, yshift = -16pt]rightA1.west)
			;

			\draw
			(rightR) edge (rightY)
			edge (rightB)
			edge (rightM)
			edge (rightX1)
			edge (rightXx)
			edge (rightA1)
			edge (rightAa)
			edge (rightC)
			;

		\end{tikzpicture}
		\caption{The modified graph $G_{a,b,x,y}$ constructed from $a$, $b$, $x$, and $y$ in
			the reduction of
			\Cref{th:hardnessFixed}.}
		\label{fig:hardnessFixedGraph}
	\end{figure}
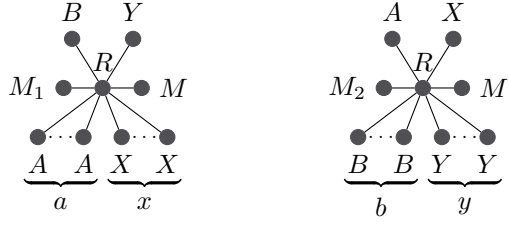

	Conversely, assume that there is a graph $G$ that satisfies all constraints produced
    by the reduction.
    By the constraints from (\ref{enum:constraintNonInjective}),
    a homomorphism from one of the constraint graphs from
    (\ref{enum:constraintARC})--(\ref{enum:constraintM2ARXC})
    to $G$ has to be injective on all copies of our Kneser graphs
    used in the constraint graph.
    Hence, a single homomorphism from the constraint graph of
	(\ref{enum:constraintM1RC}) to $G$
	yields $a_{\colourfont{M_1}} \cdot a_{\colourfont{R}} \cdot a_{\colourfont{C}}$
	homomorphisms to $G$ by considering tip-stabilising automorphisms of the Kneser graphs.
	Since the constraint enforces that we have
	precisely this number of homomorphisms from the constraint graph of
	(\ref{enum:constraintM1RC}) to $G$,
	this means that every such homomorphism has to map the tip of the copy of $K^\colourfont{R}$
	to the same vertex of $G$, say~$v_1$.
	Analogously, every homomorphism has to map the tip of the copy of $K^\colourfont{R}$
	in the constraint graph from~(\ref{enum:constraintM2RC}) to the same vertex of $G$, say $v_2$.
    \begin{claim}
        \label{cl:v1v2AreDistinct}
        The vertices $v_1$ and $v_2$ are distinct.
    \end{claim}
    \begin{claimproof}
        Assume that $v_1 = v_2$. Then, there are homomorphisms $h_1$
        and $h_2$ from the encoding of (\ref{enum:constraintM1RC}) and
        (\ref{enum:constraintM2RC}), respectively, to $G$
        that map the tip of the respective copy of $K^\colourfont{R}$
        to $v_1 = v_2$.
        Then, however, we get a homomorphism from the constraint graph of
        (\ref{enum:constraintM1M2RC}), which contradicts the constraint.
        Hence, $v_1 \neq v_2$.
    \end{claimproof}

	Now, consider the constraint from (\ref{enum:constraintRC}).
    \begin{claim}
        \label{cl:eitherV1OrV2}
        Every homomorphism from the constraint graph of (\ref{enum:constraintRC})
        has to map the tip of the copy of $K^\colourfont{R}$
        to $v_1$ or $v_2$.
    \end{claim}
    \begin{claimproof}
        A homomorphism from the constraint graph of
        (\ref{enum:constraintM1RC}) to $G$
        already yields $a_{\colourfont{R}} \cdot a_{\colourfont{M}}$
        homomorphisms from the constraint graph of (\ref{enum:constraintRC})
        to $G$.
        Then,
        a homomorphism from the encoding of
        (\ref{enum:constraintM2RC}) to $G$
        yields
        $a_{\colourfont{R}} \cdot a_{\colourfont{M}}$ further homomorphisms,
        i.e.\ these homomorphisms are distinct from the ones obtained
        from (\ref{enum:constraintM1RC}) since these map the tip
        of the copy of $K^{\colourfont{R}}$ to $v_2$ and not $v_1$,
        which is distinct from $v_2$ by \Cref{cl:v1v2AreDistinct}.
        Hence, by the number in (\ref{enum:constraintRC}),
        the claim follows.
    \end{claimproof}

	\newcommand{\pictoRootedKOf}[1]{
        \hspace{-3pt}
		\tikz[baseline=-3pt]{
			\node[smallvertex] (K) {};\node[smallvertex, label={[label distance = -2pt]left:$r$}, left = 0.4cm of K] (r) {};\node[kneser, anchor=west, inner sep = 0pt, minimum size = 0.7cm] (Kneser) at (K) {\footnotesize$K^{\colourfont{#1}}$};\draw (r) edge node[above, yshift = -1pt] {\footnotesize$P_{2\ell}$} (K);
		}
	}
	\newcommand{\pictoRootedKB}{\pictoRootedKOf{B}}

	Consider a homomorphism from the constraint graph of
	(\ref{enum:constraintPoly}) to $G$.
    By \Cref{cl:eitherV1OrV2}, such a homomorphism has to map the tip
	of the copy of $K^\colourfont{R}$ to either $v_1$ or $v_2$.
	Hence,
	\begin{equation*}
		\hom(F_{\colourfont{poly}}', G)
		= \hom(F_{\colourfont{poly}}', G;\, \text{tip of $K^{\colourfont{R}}$} \mapsto v_1)
		+ \hom(F_{\colourfont{poly}}', G;\, \text{tip of $K^{\colourfont{R}}$} \mapsto v_2).
	\end{equation*}
    By \Cref{cl:eitherV1OrV2} and the constraint from (\ref{enum:constraintRC}),
    we further have
	\begin{align*}
		\hom(F_{\colourfont{poly}}', G;\, \text{tip of $K^{\colourfont{R}}$} \mapsto v_1)
		= a_\colourfont{R} \cdot a_\colourfont{M}
		{}\cdot{} &\hom(\pictoRootedKB; r \mapsto v_1)\\
		{}\cdot{} &\hom(\pictoRootedKOf{Y}; r \mapsto v_1)\\
		{}\cdot{} &\hom(\pictoRootedKOf{A}; r \mapsto v_1)\\
		{}\cdot{} &\hom(\pictoRootedKOf{X}; r \mapsto v_1)^2.
	\end{align*}

	\begin{claim}
        \label{cl:homEquationsFiniteSet}
		\begin{enumerate}
			\item $\hom(\pictoRootedKB; r \mapsto v_1) = a_{\colourfont{B}}$.
			\item $\hom(\pictoRootedKOf{Y}; r \mapsto v_1) = a_{\colourfont{Y}}$.
			\item $\hom(\pictoRootedKOf{A}; r \mapsto v_1) = a_{\colourfont{A}} \cdot a$.
		\end{enumerate}
	\end{claim}
	\begin{claimproof}
		For the first claim, we use the constraint from (\ref{enum:constraintM1BRYC}).
		The tip of the copy of $K^\colourfont{R}$ can only be mapped to $v_1$ or $v_2$;
		more precisely, it can only be mapped to $v_1$ by the constraint from (\ref{enum:constraintM1M2RC}).
		Hence, all $a_{\colourfont{M_1}} \cdot a_\colourfont{B} \cdot a_\colourfont{R} \cdot a_\colourfont{Y} \cdot a_\colourfont{M}$
		homomorphisms from the constraint graph of (\ref{enum:constraintM1BRYC}) to $G$
		map the tip of the copy of $K^\colourfont{R}$ to $v_1$.
		Moreover, all these homomorphisms are already obtained from
        the tip-stabilising automorphisms of the Kneser graphs,
		which means that the first equality must hold.
		The second equality is proven analogously.

		For the third inequality, we use the constraint from (\ref{enum:constraintARC}).
        Its constraint graph has precisely
		$a_{\colourfont{A}} \cdot a_{\colourfont{R}} \cdot a_{\colourfont{M}} \cdot (a+1)$ homomorphisms to $G$.
		The tip of the copy of $K^\colourfont{R}$ can only be mapped to $v_1$ or $v_2$
        by \Cref{cl:eitherV1OrV2}.
        Consider the constraint graph from (\ref{enum:constraintM2ARXC}):
        analogously to the proof of the first equation, the tip of the copy of
		of $K^\colourfont{R}$ can only be mapped to $v_2$.
		Hence, due to the constraint from (\ref{enum:constraintM2ARXC}), we have
		precisely $a_{\colourfont{A}} \cdot a_{\colourfont{R}} \cdot a_{\colourfont{M}}$
		homomorphisms from the constraint graph of (\ref{enum:constraintARC})
		that map the tip of the copy of $K^{\colourfont{R}}$ to $v_2$.
		Hence, the remaining $a_{\colourfont{A}} \cdot a_{\colourfont{R}} \cdot a_{\colourfont{M}} \cdot a$ homomorphisms
		map this tip to $v_1$.
		Then, the claim follows with \Cref{cl:eitherV1OrV2} and the constraint from (\ref{enum:constraintRC}).
	\end{claimproof}

    To finish the proof, let $x \coloneqq \hom(\pictoRootedKOf{X}; r \mapsto 
    v_1) / a_\colourfont{X}$
	and observe that this is a natural number.
	Then, with \Cref{cl:homEquationsFiniteSet}, we have
	\begin{align*}
		\hom(F_{\colourfont{poly}}', G;\, \text{tip of $K^{\colourfont{R}}$} \mapsto v_1)
		= a_\colourfont{R} \cdot a_\colourfont{M} \cdot a_\colourfont{B} \cdot a_\colourfont{Y} \cdot a_\colourfont{A} \cdot a \cdot a_\colourfont{X}^2 \cdot x^2.
	\end{align*}
	Let $y \coloneqq \hom(\pictoRootedKOf{Y}; r \mapsto v_2) / 
	a^\colourfont{Y}$.
	By symmetry, we get
	\begin{align*}
		\hom(F_{\colourfont{poly}}', G;\, \text{tip of $K^{\colourfont{R}}$} \mapsto v_2)
		= a_\colourfont{R} \cdot a_\colourfont{M} \cdot a_\colourfont{B} \cdot b \cdot a_\colourfont{Y} \cdot y \cdot a_\colourfont{A} \cdot a_\colourfont{X}^2.
	\end{align*}
	Hence,
	\begin{align*}
		\hom(F_{\colourfont{poly}}', G)
		= a_\colourfont{R} \cdot a_\colourfont{M} \cdot a_\colourfont{B} \cdot a_\colourfont{Y} \cdot a_\colourfont{A} \cdot a_\colourfont{X}^2 (a \cdot x^2 + b \cdot y).
	\end{align*}
	By the encoded constraint, this implies that $a \cdot x^2 + b \cdot y = c$.
\end{proof}

\subsection{Reduction from $\existsEqualsThreeColouring$}
\label{sec:uncolouredGeneralHardness}

We will replace labels used in the proof of $\NP^\CEP$-hardness, again 
employing the binary encoding via incomparable fixed Kneser graphs. Since the 
labelled constraint graph $F'$ we produced in the reduction of 
\cref{th:NPCEPReduction} can have an arbitrary structure and 
thus can contain undesirable images of the Kneser graphs, we do an additional 
clean-up step and replace the edges of $F'$ with bidirectional gadgets. 

\begin{proof}[Proof of \Cref{th:uncolouredGeneralHardness}]
As before, given an instance $x = (F,S,k)$ of 
	$\existsEqualsThreeColouring$,
	let $m := |S|$ and $S = \{s_1,\dots,s_m\}$. 
	For any homomorphism $c \colon F[S] \to K_3$ of $F$ with exactly $k$ 
	extensions $\hat{c} \colon F \to K_3$, there will be 
	at least one graph consistent with all of our constraints,
	and none if $x = (F,S,k) \notin \existsEqualsThreeColouring$.
	
	We fix the pairwise homomorphically
    incomparable Kneser graphs $K^C$ for $C \in \CC \coloneqq \{I_0, I_1, D_0, D_1, 
    A, Z, 0, 1\}$
    and a vertex as the \emph{tip} for each of these Kneser graphs.
	Let $\ell \coloneqq \max_{C \in \CC} \lvert V(K^C) \rvert$,
	and for $C \in \CC$,
	let $a_C$ be the number of automorphisms of $K^C$
	that stabilise the tip. 
	Let $\langle \ell_1 \rangle,\dots,\langle \ell_m \rangle$
    denote gadgets as shown in \Cref{fig:binaryGadget}
    that encode the labels $\ell_1, \dots, \ell_m$ in binary by using
	the Kneser graphs $K^A$, $K^Z$, $K^0$, and $K^1$
	and fix the vertex $t$ as their own tip.
    Let $\mathsf{I}$ denote the \emph{indicator gadget}
    from \Cref{fig:indicatorGadget}
    built from $K^{I_0}$ and $K^{I_1}$, and
	let $\mathsf{BD}$ denote the
	\emph{bidirectional gadget}
    obtained by taking two copies of the the direction gadget from
    \Cref{fig:directionGadget}
    built from $K^{D_0}$ and $K^{D_1}$
    in opposing directions
    and identifying their end vertices.
	For a label ${\ell}$, let $a_{\ell}$ be the product of the
    numbers $a_C$ of tip-stabilising automorphisms of the Kneser graphs $K^C$
    used in $\langle {\ell} \rangle$.
    Analogously, define
    $a_{\mathsf{I}} \coloneqq a_{I_0} a_{I_1}$
    and
	$a_{\mathsf{BD}}\coloneqq a_{D_0}^2 a_{D_1}^2$
    as the number of automorphisms of $\mathsf{I}$ and $\mathsf{BD}$, respectively,
    that stabilise all vertices
    but the non-tip vertices of the Kneser graphs.

	Let $F'$ be the labelled graph obtained from instance $x = (F,S,k)$ by
	assigning label~$\ell_i$ to
	vertex $s_i \in S$ for every $i \in [m]$.
	Next, we will describe how labelled graphs can be encoded as constraint
	graphs without labels.
For an $L$-labelled graph $H$ with $L \subseteq \{\ell_1, \dots, \ell_m\}$
    and a natural number $n$,
	let $[\hom(H) = n] \coloneqq \hom([H]) = n'$, 
	where
	$[H]$ is obtained from $H$ by first introducing a copy of $\langle 
	{\ell} \rangle$ for every label ${\ell}$ present in $H$
    and identifying the tip of $\langle {\ell} \rangle$ with
    the vertex that $\ell$ is assigned to,
    then introducing an indicator gadget $\mathsf{I}$
    for every original vertex $v$ of $H$ and identifying its tip with $v$,
	and finally replacing every original edge of~$H$ by a bidirectional gadget 
	$\mathsf{BD}$, and let
	$n' \coloneqq  n
		\cdot a_{\mathsf{I}}^{|V(H)|}
		\cdot a_{\mathsf{BD}}^{|E(H)|}
		\cdot \prod_{\ell \in L} a_{\ell}$.
To create a constraint similar to the non-injectivity constraint 
	in the reduction from $\QuadraticPolynomial$, we define
	the sets of non-injective homomorphic images
	\begin{equation*}
        \CI_1 \coloneqq \{h(K^C) \mid C \in \CC, 
        \text{ $h$ is a non-injective homomorphism from $K^C$ to another 
        graph}\}, \text{ and}
    \end{equation*}
    \begin{equation*}
		\CI_2 \coloneqq  \{h([K_3]) \mid
		\text{ $h$ is a non-injective homomorphism from $[K_3]$ to 
		another graph}\}.
    \end{equation*}
	Since the number of isomorphism types of graphs in $\CI_1$ and $\CI_2$ is 
	finite, and the size of label gadgets is 
	polynomial in $m$, we can add the following constraints using 
	encoded labels:
	\begin{multicols}{2}
		\begin{alphaenumerate}
			\item\label{B:hcfg} $[\hom(F') = k]$,
			\item\label{B:hcellany} $[\hom(\tikz[baseline=-3pt]{
				\node[smallvertex]
				(ell_any) {}
			}) = 3]$,
			\item\label{B:hcelli} [$\hom(\tikz[baseline=-3pt]{
				\node[smallvertex, 
				label={right:\small$\ell_i$}] (ell_i) 
				{}
			}) = 1]$ for $i\in[m]$,
                \label{B:hcellanyelli}
            \vfill\null
            \columnbreak
			\item\label{B:hedges} $[\hom(\tikz[baseline=-3pt]{
				\node[smallvertex] (L) 
				{};\node[smallvertex,
				right = 
				0.3cm of L] (R) {};
				\draw (L) edge (R);
			}) = 6]$,
			\item \label{B:hK3any}$[\hom({K_3}) = 6]$,
\item \label{ct:IConstraint} $\hom(F, G) = 0$ for a graph $F$ of 
			every isomorphism type in $\CI_1$ and $\CI_2$. 
		\end{alphaenumerate}
	\end{multicols}

	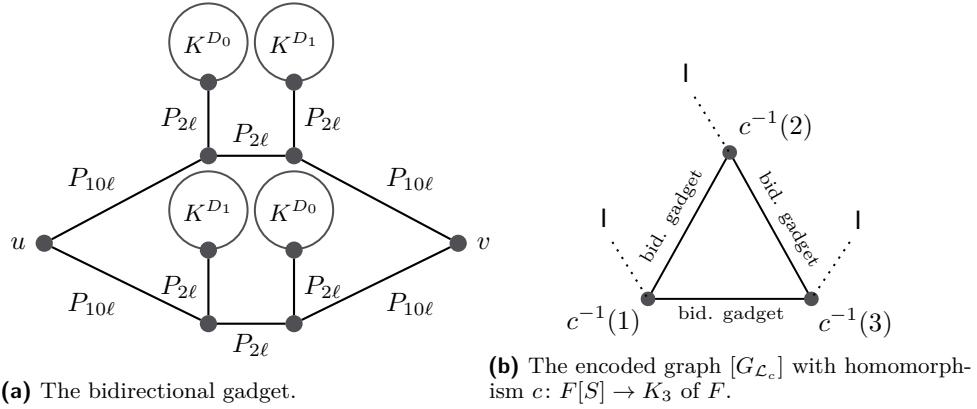
\begin{figure}
\centering
		\begin{subfigure}[b]{0.45\textwidth}
			\tikzset{every picture/.style={line width=0.75pt}}
\begin{tikzpicture}
				\node[vertex, label={left:$u$}] (u) {};
				\node[vertex, above right = 1cm and 2cm of u] (u2) {};
				\node[vertex, right = 0.9cm of u2] (v2) {};
				\node[vertex, label={right:$v$}, below right = 1cm and 2cm of 
				v2] (v) {};
				\node[vertex, above = 0.75cm of u2] (u3) {};
				\node[kneser, anchor=south] (uKneser) at (u3) {\small$K^{D_0}$};
				\node[vertex, above = 0.75cm of v2] (v3) 
				{};
				\node[kneser, anchor=south] (vKneser) at (v3) {\small$K^{D_1}$};
				
				\path[thick] (u) edge node[above left] {$P_{10\ell}$} (u2)
				(u2) edge node[left] {$P_{2\ell}$} (u3)
				(u2) edge node[above] {$P_{2\ell}$} (v2)
				(v2) edge node[right] {$P_{2\ell}$} (v3)
				(v2) edge node[above right] {$P_{10\ell}$} (v);
				
				\node[vertex, below=2cm of u2] (u2') {};
				\node[vertex, below=2cm of v2] (v2') {};
				\node[vertex, below=2cm of u3] (u3') {};
				\node[kneser, anchor=south] (uKneser) at (u3') 
				{\small$K^{D_1}$};
				\node[vertex, below=2cm of v3] (v3') {};
				\node[kneser, anchor=south] (vKneser) at (v3') 
				{\small$K^{D_0}$};
				
				\path[thick] (u) edge node[below left] {$P_{10\ell}$} (u2')
				(u2') edge node[left] {$P_{2\ell}$} (u3')
				(u2') edge node[below] {$P_{2\ell}$} (v2')
				(v2') edge node[right] {$P_{2\ell}$} (v3')
				(v2') edge node[below right] {$P_{10\ell}$} (v);
			\end{tikzpicture}
			\caption{The bidirectional gadget.}
			\label{fig:B:hardnessNPC=PEdge}
		\end{subfigure}
\begin{subfigure}[b]{0.45\textwidth}
			\centering
			
			\tikzset{every picture/.style={line width=0.75pt}} 

			\begin{tikzpicture}[x=0.75pt,y=0.75pt,yscale=-1,xscale=1]

\draw   (99.23,45.35) -- (140,119) -- (58.45,119) -- cycle ;
\draw  [color=lipicsGray  ,draw 
				opacity=1 ][fill=lipicsGray  
				,fill 
				opacity=1 ][line width=0.75]  (137.31,120.67) .. controls 
				(136.39,119.18) and (136.85,117.23) .. (138.33,116.31) .. 
				controls (139.82,115.39) and (141.77,115.85) .. (142.69,117.33) 
				.. controls (143.61,118.82) and (143.15,120.77) .. 
				(141.67,121.69) .. controls (140.18,122.61) and (138.23,122.15) 
				.. (137.31,120.67) -- cycle ;
\draw  [color=lipicsGray  ,draw 
				opacity=1 ][fill=lipicsGray  
				,fill 
				opacity=1 ][line width=0.75]  (96.53,47.02) .. controls 
				(95.61,45.53) and (96.07,43.58) .. (97.56,42.66) .. controls 
				(99.04,41.74) and (100.99,42.2) .. (101.92,43.68) .. controls 
				(102.84,45.17) and (102.38,47.12) .. (100.89,48.04) .. controls 
				(99.41,48.96) and (97.46,48.5) .. (96.53,47.02) -- cycle ;
\draw  [color=lipicsGray  ,draw 
				opacity=1 ][fill=lipicsGray  
				,fill 
				opacity=1 ][line width=0.75]  (55.76,120.67) .. controls 
				(54.84,119.18) and (55.3,117.23) .. (56.78,116.31) .. controls 
				(58.27,115.39) and (60.22,115.85) .. (61.14,117.33) .. controls 
				(62.06,118.82) and (61.6,120.77) .. (60.12,121.69) .. controls 
				(58.63,122.61) and (56.68,122.15) .. (55.76,120.67) -- cycle ;
\draw  [dash pattern={on 0.84pt off 2.51pt}]  (40.65,90.35) -- 
				(58.45,119) ;
\draw  [dash pattern={on 0.84pt off 2.51pt}]  (159.2,90.35) -- 
				(140,119) ;
\draw  [dash pattern={on 0.84pt off 2.51pt}]  (81.65,17.35) -- 
				(99.45,46) ;
				
\draw (102.29,23) node [anchor=north west][inner sep=0.75pt]   
				[align=left] {$\displaystyle c^{- 1}( 2)$};
\draw (142,122) node [anchor=north west][inner sep=0.75pt]   
				[align=left] {$\displaystyle c^{- 1}( 3)$};
\draw (16.29,121) node [anchor=north west][inner sep=0.75pt]   
				[align=left] {$\displaystyle c^{- 1}( 1)$};
\draw (51.69,98.86) node [anchor=north west][inner sep=0.75pt]  
				[font=\scriptsize,rotate=-300] [align=left] {bid. gadget};
\draw (71.78,121.41) node [anchor=north west][inner 
				sep=0.75pt]  [font=\scriptsize] [align=left] {bid. gadget};
\draw (119.51,51.9) node [anchor=north west][inner sep=0.75pt]  
				[font=\scriptsize,rotate=-60.35] [align=left] {bid. gadget};
\draw (33,72) node [anchor=north west][inner sep=0.75pt]   
				[align=left] {$\mathsf{I}$};
\draw (160,73) node [anchor=north west][inner sep=0.75pt]   
				[align=left] {$\mathsf{I}$};
\draw (74,-1) node [anchor=north west][inner sep=0.75pt]   
				[align=left] {$\mathsf{I}$};

			\end{tikzpicture}
			\caption{The encoded graph $[G_{\CL_c}]$ with 
				homomorphism $c \colon F[S] \to K_3$ of $F$.}
			\label{fig:B:hardnessNPC=PK3c}
		\end{subfigure}
        \caption{The bidirectional gadget $\mathsf{BD}$ for encoding labelled 
        graphs and an example of the graph $[G_{\CL_c}]$ satisfying the 
        constraints used in the reduction.}
	\end{figure}

	Let $(F,S,k)$ be an instance of $\existsEqualsThreeColouring$, and thus, 
	assume there exists a homomorphism $c \colon F[S] \to K_3$ of $F$, with 
	exactly $k$ extensions $\hat{c} \colon F \to K_3$.
	We define the colour classes of $F[S]$ with respect to $c$ as
	$L_i \coloneqq  \{\ell_j \mid c(s_j) = i\}$
	and their induced partition as 
	$\CL_c \coloneqq  \{L_i \mid i\in[3]\}$.
	Let $G_{\CL_c}$ denote the $\smallKthree$ with its three vertices labelled 
	by the labels in $L_1$, $L_2$, and $L_3$, respectively.
	Observe that $G_{\CL_c}$ satisfies every unencoded constraint from 
	(\ref{B:hcfg})--(\ref{B:hK3any}). 
	The encoding as $[G_{\CL_c}]$ is illustrated in 
	\Cref{fig:B:hardnessNPC=PK3c}. For $\BoundedFHomROf{\CG}$, we can also 
	compute the size of $[G_{\CL_c}]$ efficiently since 
	every label gadget appears exactly once, and set it as the additional size 
	constraint.
	
	\begin{claim}
		\label{cl:k3ca}
		$[G_{\CL_c}]$ satisfies constraints (\ref{B:hcfg}) through 
		(\ref{B:hK3any}).
	\end{claim}
	\begin{claimproof}
		Let $\hom(H) = n$ be an unencoded constraint satisfied by $G_{\CL_c}$,
        where $H$ is an $L$-labelled graph.
		For each homomorphism from $H$ to $G_{\CL_c}$,
		$[G_{\CL_c}]$ has at least
        $a_{\mathsf{I}}^{|V(H)|} \cdot a_{\mathsf{BD}}^{|E(H)|}
		\cdot \prod_{\ell \in L} a_{\ell}$ homomorphisms
		from $[H]$, corresponding to the number of automorphisms between 
		Kneser graphs.
		
		Further, there are no other homomorphisms from $[H]$ to $[G_{\CL_c}]$. 
		Since the
		Kneser graphs are pairwise homomorphically incomparable and their 
		own cores, they can only be mapped injectively to an image of the same 
		isomorphism type as $[G_{\CL_c}]$ consists of only Kneser graphs 
		connected 
		by long paths. This implies that homomorphisms have to stabilise the 
		tips of Kneser graphs, and hence, the tips of label gadgets, as well as 
		those of bidirectional gadgets.
        Note that, in particular, this holds for isolated vertices
        since the indicator gadget uses two Kneser graphs.
		Thus, original vertices of $H$ are mapped to original vertices of 
		$G_{\CL_c}$, and since the tips of bidirectional gadgets can 
		only be mapped to different original vertices in $[G_{\CL_c}]$ (cf. 
		\Cref{fig:B:hardnessNPC=PK3c}), original edges of $H$ are mapped to 
		original edges of $G_{\CL_c}$.
	\end{claimproof}
	
	Finally, $[G_{\CL_c}]$ satisfies constraint (\ref{ct:IConstraint}) as,
	otherwise, 
	since $[G_{\CL_c}]$ consists of only Kneser graphs connected by long paths,
	there would exist at least one homomorphism from a proper subgraph of a 
	Kneser graph to itself or between Kneser graphs of different isomorphism 
	types; a contradiction.
	
	Conversely, let $G$ be a graph satisfying all constraints. Since constraint 
	(\ref{B:hcellany}) enforces $3\cdot a_{\mathsf{I}}$ homomorphisms to $G$
	and homomorphisms from copies of the Kneser graphs $K^{I_0}$ and $K^{I_1}$
	have to be injective by constraint (\ref{ct:IConstraint}),
	there are exactly three images of indicator gadgets
	$[\tikz[baseline=-3pt]{
		\node[smallvertex]
		(ell_any) {}
	}]$
	in $G$ by the
	tip-stabilising automorphisms of Kneser graphs.
	Denote the vertices in $G$ to which tips of indicator gadgets are mapped to
	by $v_1$, $v_2$, and $v_3$. While every homomorphism that is 
	injective on all 
	Kneser graphs necessarily 
	requires the existence of at least $a_{\mathsf{I}}$ distinct homomorphisms,
    without further 
	consideration, the three images of
	$[\tikz[baseline=-3pt]{
		\node[smallvertex]
		(ell_any) {}
	}]$
	in $G$ might very well not be disjoint and in particular, could share 
	the same tip.
	\begin{claim}
		\label{cl:ct1}
		The vertices $v_1$, $v_2$, and $v_3$ are distinct.
	\end{claim}
	\begin{claimproof}
		Up to tip-stabilising automorphisms, we require exactly six 
		homomorphisms from $[K_3]$ to $G$, which have to be 
		injective by constraint (\ref{ct:IConstraint}). Hence, $G$ contains 
		$[K_3]$ as a subgraph, and thus, three disjoint images of 
        the indicator gadget.
	\end{claimproof}
	\begin{claim}
		\label{cl:ct2}
		Homomorphisms from
        $[\tikz[baseline=-3pt]{
			\node[smallvertex] (ell_1) 
			{}
		}],
        [\tikz[baseline=-3pt]{
			\node[smallvertex, 
			label={right:\small$\ell_1$}] (ell_1) 
			{}
		}],\dots,
		[\tikz[baseline=-3pt]{
			\node[smallvertex, 
			label={right:\small$\ell_m$}] (ell_m) 
			{}
		}]$ to $G$ map their tips to $v_1$, $v_2$, or 
		$v_3$.
	\end{claim}
	\begin{claimproof}
		Follows from (\ref{B:hcellany}) and the definition of $v_1$, $v_2$, and $v_3$.
	\end{claimproof}
	\begin{claim}
		\label{cl:ct3}
		Every homomorphism from
		$[\tikz[baseline=-3pt]{
			\node[smallvertex, label={left:\small$\ell$}] (L) {};\node[smallvertex, label={right:\small$\ell'$}, right = 
			0.3cm of L] (R) {};
			\draw (L) edge (R);
		}]$
		to $G$ maps the tips of
		$\langle \ell \rangle$ and $\langle \ell' 
		\rangle$ to two distinct vertices from $v_1$, $v_2$, or $v_3$.
	\end{claim}
	\begin{claimproof}
		By constraint (\ref{B:hedges}), all images of
		$[\tikz[baseline=-3pt]{
			\node[smallvertex] (L) 
			{};\node[smallvertex, 
			right = 
			0.3cm of L] (R) {};
			\draw (L) edge (R);
		}]$
		in $G$ are contained in the subgraph copy of 
		$[K_3]$ in $G$.
	\end{claimproof}
	
	From \cref{cl:ct1} and \cref{cl:ct2} follows that every homomorphism to $G$
	from a graph that contains gadgets
    $\langle \ell_1 \rangle,\dots,\langle \ell_m \rangle$ on vertices
    marked by indicator gadgets
	induces a partition of 
	labels $\CL \coloneqq  \{L_1,L_2,L_3\}$ into at most three parts with 
	respect to mapping their tips to either $v_1$, $v_2$, or $v_3$.
	For any set of labels partitioned in this way, let $G_{\CL}$ 
	denote the $\smallKthree$ with its three vertices labeled 
	by the labels in $L_1$, $L_2$, and $L_3$, respectively.
	
	\begin{claim}
        Let $L \coloneqq \{\ell_1, \dots, \ell_m\}$.
		For every $L$-labelled graph $F$
        and for every partition 
		of $L$ into $\CL = \{L_i\}_{i\in[3]}$, it holds that
		\[ \hom([F], G) = \hom(F,G_{\CL})
			\cdot a_{\mathsf{I}}^{|V(F)|}
			\cdot a_{\mathsf{BD}}^{|E(F)|}
			\cdot\prod_{i \in [m]}
				a_{\ell_i}.\]
	\end{claim}
	\begin{claimproof}
		Let $c' \colon [F] \to G$ be a homomorphism. By \cref{cl:ct2}, $c'$ 
		maps encoded vertices of $F$ to $v_1$, $v_2$, or $v_3$.
		With \cref{cl:ct3} follows that vertices 
		incident to an encoded edge of $F$ are mapped to distinct vertices 
		from $v_1$, $v_2$, or $v_3$.
		Furthermore, up to automorphisms of label gadgets, 
		constraint (\ref{B:hedges}) and the non-injectivity constraint 
		(\ref{ct:IConstraint}) enforce that there is precisely one homomorphism 
		from an encoded edge to every pair of distinct vertices $u,v \in 
		\{v_1,v_2,v_3\}$.
		Now, analogously to the proof for encodings of directed edges in 
		\cref{le:encodingPreservesHom}, any homomorphism from $F$ to $G_{\CL}$ 
		corresponds to a homomorphism from $[F]$ to $G$ modulo automorphisms of 
		Kneser graphs.
	\end{claimproof}

	Finally, from constraint (\ref{B:hcfg}) follows that $F'$ has exactly $k$ 
	homomorphisms to $G_{\CL}$ with its three vertices labelled by 
	$L_1$, $L_2$, and $L_3$, respectively. 
	Thus, $(F,S,k)$ is in $\existsEqualsThreeColouring$.
\end{proof}
\clearpage{}
\clearpage{}\section{Binomial Equations is $\NP$-complete}
\label{sec:binomialEquations}

\dproblem{\BPoly}{Natural numbers $a$, $b$, and $c$ in binary encoding.}
{Are there natural numbers $x$ and $y$ such that $ax(x-1) + by = c$?}

\begin{theorem}
	\BPoly is $\NP$-complete.
\end{theorem}
\begin{proof}
	Membership in $\NP$ follows as for \QuadraticPolynomial.
	By \cite[Lemma~5.1]{moore_nature_2011}, the variant of \QuadraticPolynomial with the additional constraint that $x \equiv 1 \mod 2$ is $\NP$-complete.
	The problem remains $\NP$-complete when further restricted to instances with $c \geq a$.
	Let $a,b,c$ be an instance of this problem.

	\begin{claim}
		$ax^2 + by = c$, $x \equiv 1 \mod 2$ has a solution in $\mathbb{N}$ if and only if $4aX(X-1) + bY = c-a$ has a solution in $\mathbb{N}$.
	\end{claim}
	\begin{claimproof}
		Given a solution $x,y \in \mathbb{N}$ to the first system of equations, define $X \coloneqq \frac{x+1}{2}$ and $Y \coloneqq y$.
		Then 
		\[
			4aX(X-1) + bY = a(x^2 - 1) + by = c-a.
		\]
		Conversely, given a solution $X, Y \in \mathbb{N}$ to the second equation, 
		distinguish cases:
		If $X \geq 1$, define $x \coloneqq 2X - 1$ and $y \coloneqq Y$. Then
		\[
			ax^2 + by = 4aX(X-1) + a + bY = c. 
		\]
		If $X = 0$, define $x \coloneqq 1$ and $y \coloneqq Y$. Then
		$
			ax^2 + by = a + bY = c.
		$
	\end{claimproof}
	Hence, \BPoly is $\NP$-complete.
\end{proof}

\clearpage{}
\clearpage{}\section{Triangle subgraphs on $n+1$ vertices}
\label{sec:triangles}
In this section, the following theorem is proven:

\thmtriangles*

In the following, we prepare for the proof of \cref{thm:triangles}, in which we show that there exists a degree sequence and a corresponding tree, such that the complement of said tree on $n+1$ vertices has the required triangle subgraph count.

We define multisets that correspond to degree sequences of trees. We consider the following inductively defined set of multisets of positive integers. 
Let $p \in \mathbb{N}$. The set $\mathcal{D}^p$ is defined as follows:
\begin{enumerate}
	\item the multiset $D_0^p \coloneqq \multiset{2, \dots, 2, 1,1}$ with $p$ copies of $2$ and two copies of $1$ is in $\mathcal{D}^p$.
	\item if $D \in \mathcal{D}^p$ contains an entry $c \geq 2$ at least twice, define $D'$ as the multiset obtained from $D$ by replacing two entries of $c$ by $c+1$ and $c-1$. Then $D' \in \mathcal{D}^p$. We call this operation \textit{splitting} and write $D \to D'$.
\end{enumerate}
We remark that a similar construction was described in \cite[Definition~2]{wagner_class_2006}.

A \emph{chain} in $\mathcal{D}^p$ is a collection $D_1, \dots, D_m$ such that $D_1 \to D_2 \to \dots \to D_m$.
Note that all $D \in \mathcal{D}^p$ are comprised of $p+2$ entries summing to $2p+2$.
Since the sum of squares of entries strictly increases in every splitting step, the length of any chain in $\mathcal{D}^p$ is bounded by some function of $p$.
Write $s(p)$ for number of multisets in the longest chain in $\mathcal{D}^p$.
The sequence $(s(p))_{p \in \mathbb{N}}$ is \href{https://oeis.org/A121924}{\texttt{A121924}} in The On-Line Encyclopedia of Integer Sequences.

\begin{lemma} \label{obs:monot}
	If $p \leq q$ then $s(p) \leq s(q)$.
\end{lemma}
\begin{proof}
	Clearly, for any chain in $\mathcal{D}^p$ there is a chain in $\mathcal{D}^q$ that corresponds to the chain in $\mathcal{D}^p$ in the sense that the same splitting operations were performed and the $q-p$ additional $2$s in $D_0^q$ were ignored.
\end{proof}

\begin{lemma} \label{lem:lowerbound}
	For $p = \binom{x}{2}$ with $x \in \mathbb{N}$, there exists a chain in $\mathcal{D}^p$ of length $\binom{x}{3}$.
	Hence, $s(\binom{x}{2}) \geq \binom{x}{3}$ for all $x \in \mathbb{N}$.
\end{lemma}
\begin{proof}
	By induction on $x$, we prove that $D_0^p \in \mathcal{D}^p$ can be transformed into $\multiset{x, x-1, \dots, 1, \dots, 1}$ in at least $\binom{x}{3}$ steps.
	Here, $\multiset{x, x-1, \dots, 1, \dots, 1}$ denotes the multiset containing every $2 \leq y \leq x$ precisely once and whose entries sum to $2p +2$.
	
	If $x = 2$ then $p = 1$ and $\binom{x}{3} = 0$. The initial sequence $D_0^p$ is as desired.
	
	By the inductive hypothesis, we can transform $D_0^{\binom{x}{2}}$ into $\multiset{x,x-1,\dots, 1, \dots, 1}$ in at least $\binom{x}{3}$ many steps. 
	Write $T_x \coloneqq \multiset{2,\dots,2}$ for the multiset containing $x$ copies of $2$.
	Since $\binom{x+1}{2} - \binom{x}{2} = x$, the initial sequence $D_0^{\binom{x+1}{2}} = D_0^{\binom{x}{2}} \cup T_x$ can be transformed into $\multiset{x,x-1,\dots, 1, \dots, 1} \cup T_x$ in at least $\binom{x}{3}$ many steps.
	
	We can now transform
	\begin{itemize}
		\item  $\multiset{x,x-1,\dots, 1, \dots, 1} \cup T_x$ to $\multiset{x+1,x-1,\dots, 1, \dots, 1} \cup T_{x-1}$ in $x-1$ steps by using one extra $2$,
		\item  $\multiset{x+1,x-1,\dots, 1, \dots, 1} \cup T_{x-1}$ to $\multiset{x+1,x, x-2,\dots, 1, \dots, 1} \cup T_{x-2}$ in $x-2$ steps by using one extra $2$.
		\item \dots
	\end{itemize}
	In this way, we end up at the desired multiset in $\sum_{y = 1}^{x-1} y = \binom{x}{2}$ steps consuming all extra $2$s.
	Observing $\binom{x}{3} + \binom{x}{2} = \binom{x+1}{3}$, we conclude.
\end{proof}

\begin{corollary} \label{cor:slarge}
	For $n \geq 130$, $s(\lfloor n/2 \rfloor - 1) \geq n-2$.
\end{corollary}
\begin{proof}
	Let $x \in \mathbb{N}$ be such that $\binom{x}{2} \leq \lfloor n/2 \rfloor -1 \leq \binom{x+1}{2}$. 
	Then $x \geq \sqrt{2\lfloor n/2 \rfloor - 2} -1 \geq \sqrt{n-3} - 1$.
	By \cref{obs:monot,lem:lowerbound},
	\[
	s(\lfloor n/2 \rfloor -1) \geq s(\binom{x}{2}) \geq \binom{x}{3} \geq \binom{\sqrt{n-3}-1}{3}.
	\]
	Clearly, $\binom{\sqrt{n-3}-1}{3} \in \Theta(n^{3/2})$ is larger than $n-2$ for $n$ large enough. Evaluating both functions numerically shows that taking $n \geq 130$ suffices.
\end{proof}

In the next step, we construct trees from the multisets in $\mathcal{D}^p$.

\begin{lemma} \label{lem:multiset-tree}
	For every $D \in \mathcal{D}^p$, there exists a tree $T$ with degree sequence $D$.
\end{lemma}
\begin{proof}
	For $D_0^p$, define $T_0$ as the tree with vertices $V(T_0) \coloneqq \{v_i ~|~ i \in [p+2]\}$ and $E(T_0) \coloneqq \{(v_i, v_{i+1}) ~|~ i \in [p+1]\}$, cf.\@ \cref{fig:tree-construction}. Clearly, $T_0$ is a tree with degree sequence $D_0^p$. 
	
	Now by structural induction we assume that for $D_1 \to D_2$ we have a tree $T_1$ of degree sequence $D_1$ given and $D_2$ results from changing two entries of value $c \geq 2$ to $c-1$ and $c+1$. Let $v_{1}, v_{2} \in V(T_1)$ with $v_1 \neq v_2$ be two vertices with $\deg(v_1) = \deg(v_2) = c$. Since $T_1$ is a tree and $c \geq 2$ there must be a vertex $v_{3} \neq v_{2}$ with $v_1 v_3 \in E(T_1)$, $v_2 v_3 \notin E(T_1)$ and $v_3$ is not on the path from $v_1$ to $v_2$. Now let $T_2$ be a tree with $V(T_2)\coloneqq V(T_1)$ and $E(T_2) \coloneqq (E(T_1) \setminus \{v_1 v_3\}) \cup v_2 v_3$. Clearly, $\deg(v_3)$ remains unchanged while $\deg(v_2)$ increases by one and $\deg(v_1)$ decreases by one such that $D_2$ is the degree sequence of $T_2$. Since $T_1$ is a tree and $v_3$ is not on the path between $v_1$ and $v_2$, removing the edge $v_1 v_3$ creates two subtrees, one including $v_1$ and $v_2$, while the other includes $v_3$. Adding the edge $v_2 v_3$ connects the two subtrees again, therefore $T_2$ is also a tree.
\end{proof}
\begin{figure}
	\centering
	\begin{tikzpicture}[scale=1.2]
		\tikzstyle{highlighted} = [fill=lipicsYellow, draw=lipicsYellow];
		\node [vertex] (0) at (0, -1) {};
\node [vertex,highlighted] (1) at (0, -1.5) {};
		\node [vertex] (2) at (0, -2) {};
		\node [vertex] (3) at (0, -2.5) {};
		\node [vertex] (4) at (0, -3) {};
		\node [vertex,highlighted] (5) at (0, -3.5) {};
		\node [vertex] (6) at (0, -4) {};

		\node [vertex] (7) at (2, -1.5) {};
		\node [vertex, highlighted] (8) at (2, -2) {};
		\node [vertex] (9) at (2, -2.5) {};
		\node [vertex, highlighted] (10) at (2, -3) {};
		\node [vertex] (11) at (2, -3.5) {};
		\node [vertex] (12) at (2, -4) {};
		
		\node [vertex] (13) at (2.5, -3.5) {};
		\node [vertex] (14) at (4, -2) {};
		\node [vertex] (15) at (4, -2.5) {};
		\node [vertex,highlighted] (16) at (4, -3) {};
		\node [vertex,highlighted] (17) at (4, -3.5) {};
		\node [vertex] (18) at (4, -4) {};
		
		\node [vertex] (19) at (4.5, -3.5) {};
		\node [vertex] (20) at (4.5, -3) {};
		\node [vertex] (21) at (6, -2) {};
		\node [vertex,highlighted] (22) at (6, -2.5) {};
		\node [vertex,highlighted] (23) at (6, -3) {};
		\node [vertex] (24) at (6, -3.5) {};
		
		\node [vertex] (25) at (6, -4) {};
		\node [vertex] (26) at (6.5, -3.5) {};
		\node [vertex] (27) at (6.5, -3) {};
		
		\node [vertex] (28) at (8.5, -2.5) {};
		\node [vertex] (29) at (8, -2.5) {};
		\node [vertex] (30) at (8, -3) {};
		\node [vertex] (31) at (8, -3.5) {};
		\node [vertex] (32) at (8, -4) {};
		\node [vertex] (33) at (8.5, -3.5) {};
		\node [vertex] (34) at (8.5, -3) {};
		\draw (0) to (1);
		\draw (1) to (2);
		\draw (2) to (3);
		\draw (3) to (4);
		\draw (4) to (5);
		\draw (5) to (6);
		\draw (7) to (8);
		\draw (8) to (9);
		\draw (9) to (10);
		\draw (10) to (11);
		\draw (11) to (12);
		\draw (11) to (13);
		\draw (14) to (15);
		\draw (15) to (16);
		\draw (16) to (17);
		\draw (17) to (18);
		\draw (17) to (19);
		\draw (16) to (20);
		\draw (21) to (22);
		\draw (22) to (23);
		\draw (23) to (24);
		\draw (24) to (25);
		\draw (24) to (26);
		\draw (24) to (27);
		\draw (29) to (30);
		\draw (30) to (31);
		\draw (31) to (32);
		\draw (31) to (33);
		\draw (31) to (34);
		\draw (30) to (28);
	\end{tikzpicture}
	\caption{Trees corresponding to a chain of multisets in $\mathcal{D}^5$. 
		The leftmost tree corresponds to $\multiset{2,2,2,2,2,1,1}$ and the rightmost tree corresponds to $\multiset{4,3,1,1,1,1,1}$. The yellow vertices have the same degree and were chosen for the operation described in \cref{lem:multiset-tree}.}
	\label{fig:tree-construction}
\end{figure}
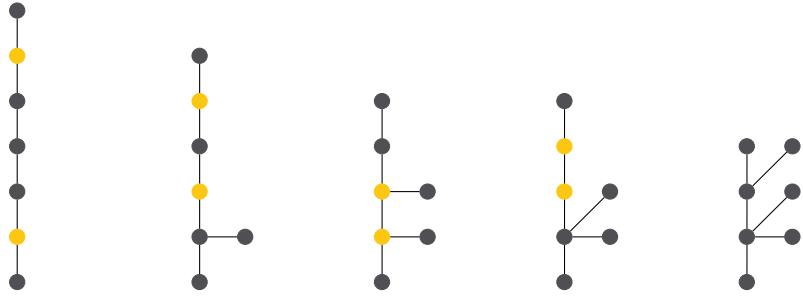

Given a graph $G$, write $\overline{G}$ for its \emph{complement}, i.e.\@ the graph with vertex set $V(G)$ and edge set $\binom{V(G)}{2} \setminus E(G)$.
Alternatively, one may think of $\overline{G}$ as the graph obtained from a clique on the vertex set of $G$ by deleting all edges in $G$. \cref{thm:nr-triangles} gives a formula for calculating the number of triangles in $\overline{G}$ based on $G$.

\begin{lemma}
	\label{thm:nr-triangles}
	Let $G$ be a graph on $n$ vertices. Then
	\[
	\sub(\smallKthree, \overline{G}) = \binom{n}{3} - \sub(\smallKtwo, G) (n-2) + \sum_{v \in V(G)} \binom{\deg_G(v)}{2} - \sub(\smallKthree, G).
	\]
\end{lemma}
\begin{proof}
	Write $V \coloneqq V(G)$.
	For $e \in \binom{V}{2}$, write $D_e = \{t \in \binom{V}{3} \mid e \subseteq t\}$ for the set of $3$-sets of vertices containing $e$.
	Observe that $|D_e| = n-2$ for every $e$.
	A $3$-set of vertices induces a triangle in $\overline{G}$ if and only if none of its edges is in $G$.
	Hence, $\sub(\smallKthree, \overline{G}) = \binom{n}{3} - |\bigcup_{e \in E(G)} D_e|$.
	Then by Inclusion--Exclusion,
	\begin{align*}
		\left|\bigcup_{e \in E(G)} D_e \right|
		& = \sum_{\emptyset \neq E \subseteq E(G)} (-1)^{|I|+1} \left|\bigcap_{e \in E} D_e \right| \\
		& = \sum_{e \in E(G)} |D_e| - \sum_{\substack{\{e_1, e_2\} \subseteq E(G) \\ |e_1 \cap e_2| = 1}} 1  + \sum_{\substack{\{e_1, e_2,e_3\} \subseteq E(G) \\ |e_1 \cap e_2| = |e_1 \cap e_3| = |e_2 \cap e_3| = 1}} 1 \\
		&= \sub(\smallKtwo, G) (n-2) - \sum_{v \in V(G)} \binom{\deg_G(v)}{2} + \sub(\smallKthree, G).
	\end{align*}
	This concludes the proof.
\end{proof}

This concludes the preparations for the proof of \cref{thm:triangles}.

\begin{proof}[Proof of \cref{thm:triangles}]
For $h \leq \binom{n-1}{3}$, this follows directly from \cref{thm:reconstruct-cliques}. To close the gap to $h \leq \binom{n}{3}$ we show that for $\binom{n-1}{3} \leq h \leq \binom{n}{3}$ there is a cycle-free graph $G$ on $n+1$ vertices such that $\sub(\smallKthree, \overline{G})=h$.

This graph is constructed based on the multisets discussed above.
Let $p \leq n-1$.
For $D \in \mathcal{D}^p$, write $T(D)$ for a tree as in \cref{lem:multiset-tree}.
Since the sum of the entries of $D$ is $2p + 2$ and $D$ is the multisets of degree of $T(D)$, 
the tree $T(D)$ has $p + 1$ edges and $p+2 \leq n+1$ vertices.
Write $F(D)$ for the forest obtained from $T(D)$ by adding isolated vertices such that $F(D)$ has $n+1$ vertices.
Write $G(D)$ for the complement of $F(D)$.

By \cref{thm:nr-triangles},
\begin{equation} \label{eq:Gp1}
	\sub(\smallKthree, G(D_0^p)) = \binom{n+1}{3} - (p+1)(n-1) + p.
\end{equation}
Clearly, increasing $p$ decreases $\sub(\smallKthree, G(D_0^p))$.
More precisely, the gap between the $\smallKthree$-counts in $G(D_0^p)$ and $G(D_0^{p+1})$ is
\begin{equation} \label{eq:gap}
	\sub(\smallKthree, G(D_0^p)) - \sub(\smallKthree, G(D_0^{p+1})) = n-2.
\end{equation}
Lastly, by \cref{eq:Gp1},
\begin{equation}
	\sub(\smallKthree,G(D^{\lfloor n/2 \rfloor -1}_0))
	\geq \binom{n+1}{3} - \binom{n}{2} + \frac{n-3}{2} \geq \binom{n}{3}.
\end{equation}
Since $h \leq \binom{n}{3}$, there exists $p \geq \lfloor n/2 \rfloor - 1$ such that $\sub(\smallKthree,G(D_0^{p+1})) < h \leq \sub(\smallKthree, G(D_0^{p}))$.
By \cref{eq:gap}, the gap between this upper and lower bound is $n-2$.
By \cref{thm:nr-triangles}, if $D, D' \in \mathcal{D}^p$ are such that $D \to D'$ then
\begin{equation} \label{eq:gap1}
	\sub(\smallKthree, G(D)) = \sub(\smallKthree, G(D')) -1
\end{equation}
since $\binom{d}{2} + \binom{d}{2} = \binom{d+1}{2} + \binom{d-1}{2} - 1$ for all $d \in \mathbb{N}$.
By \cref{obs:monot,cor:slarge}, $s(p) \geq n-2$.
Hence, there exists a chain $D_0^p \to D_1 \to \dots \to D_{n-3}$ in $\mathcal{D}^p$.
By \cref{eq:gap1}, one of these multisets $D_i$ is such that $h = \sub(\smallKthree, G(D_i))$.
\end{proof}\clearpage{}
\clearpage{}\section{Material Omitted in \Cref{sec:parameterized}}
\label{app:parameterised}

\begin{proof}[Proof of \cref{lem:homindsum}]
	Write $\surj(F, G)$ for the number of homomorphisms $h \colon F \to G$ such that $h$ is surjective on vertices.  The following equality is well-known, cf.\@ \cite{lovasz_large_2012}. For all graphs $F$ and $G$,
	\begin{equation}\label{lem:surjhom}
		\surj(F, G) = \sum_{U \subseteq V(G)} (-1)^{\abs{V(G) \setminus U}} \hom(F, G[U]).
	\end{equation}
	
	We partition the homomorphisms from $F$ to $G$ by their images. This yields the first step in the following chain of equalities. Subsequently, we rearrange the sums:
	\begin{align*}
		\hom(F, G)
		&= \sum_{\substack{U \subseteq V(G)\\\abs{U}\leq k}} \surj(F, G[U]) \\
		&\overset{\eqref{lem:surjhom}}{=} \sum_{\substack{U \subseteq V(G)\\\abs{U}\leq k}} \sum_{U' \subseteq U} (-1)^{|U \setminus U'|} \hom(F, G[U'])\\
		&=\sum_{H : \abs{H} \leq k} \hom(F, H) \sum_{\substack{ U \subseteq V(G)\\\abs{U}\leq k}}  (-1)^{\abs{U}-\abs{H}} \sum_{U' \subseteq U} \delta_{G[U'] \cong H}\\
		&=\sum_{H : \abs{H} \leq k} \hom(F, H) \sum_{\substack{ U' \subseteq V(G)\\|U'| = |H|}} 
		\delta_{G[U'] \cong H}
		\sum_{\substack{ U' \subseteq U \subseteq V(G)\\\abs{U}\leq k}}
		(-1)^{\abs{U}-\abs{H}}\\
		&=\sum_{H : \abs{H} \leq k} \hom(F, H) \sum_{\substack{ U' \subseteq V(G)\\|U'| = |H|}} 
		\delta_{G[U'] \cong H}
		\sum_{j=0}^{k- \abs{H}} \binom{\abs{G}-\abs{H}}{j} (-1)^j
		\\
		&=\sum_{H : \abs{H} \leq k} \hom(F, H) \indsub(H, G)
		\sum_{j=0}^{k- \abs{H}} \binom{\abs{G}-\abs{H}}{j} (-1)^j \\
		&=\sum_{H : \abs{H} \leq k} \hom(F, H) \indsub(H, G) (-1)^{k-\abs{H}}\binom{\abs{G}-\abs{H}-1}{k-\abs{H}}.
	\end{align*}
	The last equality follows from a standard identity involving binomial coefficients, cf.\@ \cite[Section~24.1.1.II.B]{abramowitz_handbook_1965}. This is the only point where the assumption $|G| > k$ matters.
\end{proof}

\begin{proof}[Proof of \cref{lem:bezout}]
	Let $y$ denote the greatest common divisor of the $y_1, \dots, y_n$.
	Clearly, $X \subseteq y\mathbb{N}$.
	By Bezóut's identity, write $y = \sum \gamma_iy_i$ for some $\gamma_i \in \mathbb{Z}$. If all $\gamma_i$ are non-negative then $y \in X$ and $y\mathbb{N} = X$. Otherwise, let $c = - \min \gamma_i > 0$. Let $N = cy_1 \sum y_i$. Clearly, $N \in X \cap y\mathbb{N}$. Let $z \in y\mathbb{N}$ be greater than $N$. Write $z = N + ay_1 + by$ for some $a \geq 0$ and $0 \leq b < y_1/y$. Then 
	\[
	z = N + ay_1 + by = cy_1\sum y_i + ay_1 + b\sum \gamma_i y_i
	= (cy_1+a+b\gamma_i)y_1 + \sum_{i \geq 2} (cy_1 + b\gamma_i)y_i
	\]
	If $\gamma_i \geq 0$ then clearly $cy_1 + a +b\gamma_i \geq cy_1 + b\gamma_i \geq 0$. If $\gamma_i < 0$ then $cy_1 + a +b\gamma_i \geq cy_1 + y_1 \gamma_i/y \geq \gamma_i(1/y -1)y_1 \geq 0$. Hence, $z \in X$.
\end{proof}
\newcommand{\pile}{\textup{\textsc{$p$-IntegerLinearEquations}}\xspace}
\begin{theorem}[\cite{damschke_sparse_2013}] \label{thm:damschke}
	The following problem is in $\FPT$:
	\pproblem{\pile}{A  matrix $A \in \mathbb{N}^{k \times m}$, a vector $b \in \mathbb{N}^k$, both given in binary}{$km$}
	{Does there exist $x \in \mathbb{N}^m$ such that $Ax = b$?}
\end{theorem}
\begin{proof}
	We outline how the theorem follows from the results in \cite{damschke_sparse_2013}.
	For a set of columns $C$, write $A_C$ for the submatrix of $A$ with columns $C$ and $x_C$ analogously for the subvector of $x$ corresponding to $C$.
	A set of columns $C \neq \emptyset$ is \emph{feasible} if the system $A_C x_C = b$ has a solution where all entries of $x_C$  are positive integers. A set of columns $C$ is \emph{minimal feasible} if $C$ is feasible but no $C' \subsetneq C$ is feasible.
	By \cite[Theorem~4]{damschke_sparse_2013}, the minimal feasible sets of the input system can be enumerated in $\FPT$.
	Observe that $Ax = b$ is feasible iff there exists a minimal feasible set. This concludes the proof.
\end{proof}\clearpage{}

\end{document}